\documentclass[bj]{imsart}

\RequirePackage{amsthm,amsmath,amsfonts,amssymb}
\RequirePackage[numbers]{natbib}

\RequirePackage[colorlinks,citecolor=blue,urlcolor=blue]{hyperref}
\RequirePackage{graphicx}
\graphicspath{{figures/}}
\usepackage{xfrac}
\usepackage{bbm}
\allowdisplaybreaks
\usepackage{mathtools}

\DeclarePairedDelimiter\floor{\lfloor}{\rfloor}
\usepackage{booktabs} 
\usepackage{caption} 
\usepackage{subcaption} 
\usepackage{graphicx}
\usepackage{pgfplots}
\usepackage[all]{nowidow}
\usepackage[utf8]{inputenc}
\usepackage{tikz}
\usetikzlibrary{er,positioning,bayesnet}
\usepackage{multicol}
\usepackage{algpseudocode,algorithm,algorithmicx}
\usepackage[inline]{enumitem} 

\definecolor{orange}{HTML}{FF7F0E}
\definecolor{green}{HTML}{2CA02C}

\newcommand{\M}{\mathcal{M}}
\newcommand\numberthis{\addtocounter{equation}{1}\tag{\theequation}}

\usepackage{tikz}
\usetikzlibrary{matrix}

\pgfplotsset{compat=1.17}

\startlocaldefs
\newtheorem{theorem}{Theorem}

\newtheorem{lemma}[theorem]{Lemma}
\newtheorem{proposition}[theorem]{Proposition}

\theoremstyle{remark}
\newtheorem*{remark}{Remark}
\theoremstyle{definition}
\newtheorem{assumption}[theorem]{Assumption}
\theoremstyle{definition}

\numberwithin{theorem}{section}

\endlocaldefs

\begin{document}

\begin{frontmatter}
\title{Adaptive schemes for piecewise deterministic Monte Carlo algorithms}
\runtitle{Adaptive schemes for PDMC algorithms}

\begin{aug}


\author[A]{\fnms{Andrea} \snm{Bertazzi}\ead[label=e1,mark]{a.bertazzi@tudelft.nl}},
\author[A]{\fnms{Joris} \snm{Bierkens}\ead[label=e2,mark]{joris.bierkens@tudelft.nl}}
\address[A]{Delft Institute of Applied Mathematics, TU Delft, Mekelweg 4, 2628 CD Delft, the Netherlands. \printead{e1,e2}}
\end{aug}

\begin{abstract}
The Bouncy Particle sampler (BPS) and the Zig-Zag sampler (ZZS) are continuous time, non-reversible Monte Carlo methods based on piecewise deterministic Markov processes. Experiments show that the speed of convergence of these samplers can be affected by the shape of the target distribution, as for instance in the case of anisotropic targets.
We propose an adaptive scheme that iteratively learns all or part of the covariance matrix of the target and takes advantage of the obtained information to modify the underlying process with the aim of increasing the speed of convergence.
Moreover, we define an adaptive scheme that automatically tunes the refreshment rate of the BPS or ZZS.
We prove ergodicity and a law of large numbers for all the proposed adaptive algorithms. Finally, we show the benefits of the adaptive samplers with several numerical simulations.
\end{abstract}
\begin{keyword}[class=MSC]
            \kwd[Primary ]{60J25}
            \kwd[; secondary ]{65C05}
\end{keyword}
\begin{keyword}
\kwd{Adaptive Markov process Monte Carlo}
\kwd{piecewise deterministic Markov processes}
\kwd{bouncy particle sampler}
\kwd{zig-zag sampler}
\kwd{ergodicity}
\end{keyword}

\end{frontmatter}


\section{Introduction}
Piecewise Deterministic Markov processes (PDMP) have been recently used for Monte Carlo sampling as a continuous time alternative to Markov chain Monte Carlo (MCMC) methods. The Bouncy Particle sampler (BPS) and the Zig-Zag sampler (ZZS), introduced in \cite{BPS} and \cite{ZZ} respectively, are primary examples of this new class of methods, with notable predecessors \cite{PetersDeWith2012, MichelKapferKrauth2014, Monmarche2016}.     
Both samplers are based on continuous time dynamics defined by piecewise linear, deterministic trajectories and changes in the velocity of the process at random times. The simplicity of these dynamics is such that the underlying processes can be simulated exactly in many relevant scenarios. 
A key aspect of these samplers is the non-reversibility of the underlying Markov process. This property has been observed to result in a lower asymptotic variance of Monte Carlo estimates for moments of the target density \cite{andrieu2019peskuntierney,bierkens_duncan_2017}. 
Moreover, these algorithms can be naturally modified to achieve asymptotically exact subsampling. In the Bayesian setting, this means that Piecewise Deterministic Monte Carlo (PDMC) algorithms need to access only a small portion of the data set at every iteration without introducing a bias. 
The reader is referred to \cite{fearnhead2018,vanetti2017piecewisedeterministic} for an in depth description of the general methodology of PDMC algorithms. 

Several papers have studied the convergence properties of the BPS and the ZZS in recent years. It was first observed in \cite{BPS} that the BPS can fail to be ergodic unless a refreshment of the velocity vector is performed at random times. Exponential ergodicity of the BPS was proved in \cite{BPSexp_erg,BPS_Durmus} for different distributions from which refreshments of the velocity vector can be drawn. Similarly, in \cite{ZZerg} the ZZS was shown  to converge exponentially to its invariant distribution under reasonable assumptions and without the need of velocity refreshments. Exponential convergence in the $L^2$ sense was established for both samplers in \cite{andrieu2019hypocoercivity} using the hypocoercivity framework, and recently, using Poincaré inequalities in space-time, in \cite{Lu2020}.
A study of the scaling limits was conducted in \cite{bierkens2018highdimensional}, giving also a criterion to choose the refreshment rate of BPS. The asymptotic variance of these processes has been studied in e.g. \cite{andrieu2019peskuntierney,bierkens_duncan_2017}.
It is also possible to design PDMC algorithms with non-linear trajectories, see e.g. \cite{vanetti2017piecewisedeterministic, bierkens2020boomerang}.

Although PDMC sampling methods offer some important benefits as mentioned above, computation remains expensive, which requires us to investigate possible performance improvements. In particular, a strong performance degradation is observed when the target distribution $\pi$ is anisotropic. 
\begin{figure}[b]
\centering
\begin{subfigure}{0.49\textwidth}
  \centering
    \includegraphics[width=\textwidth]{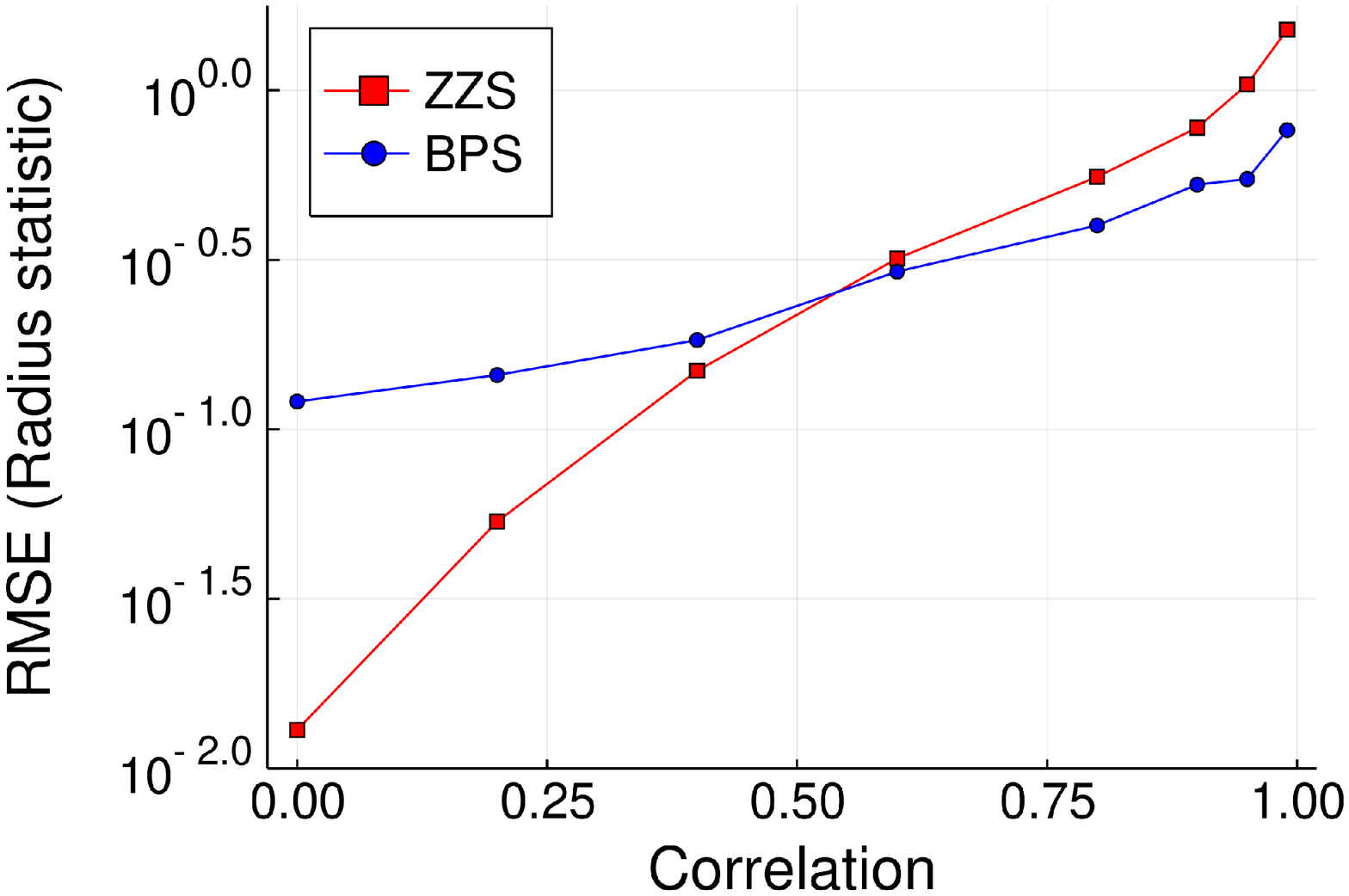}
\end{subfigure}
\begin{subfigure}{0.49\textwidth}
  \centering
    \includegraphics[width=\textwidth]{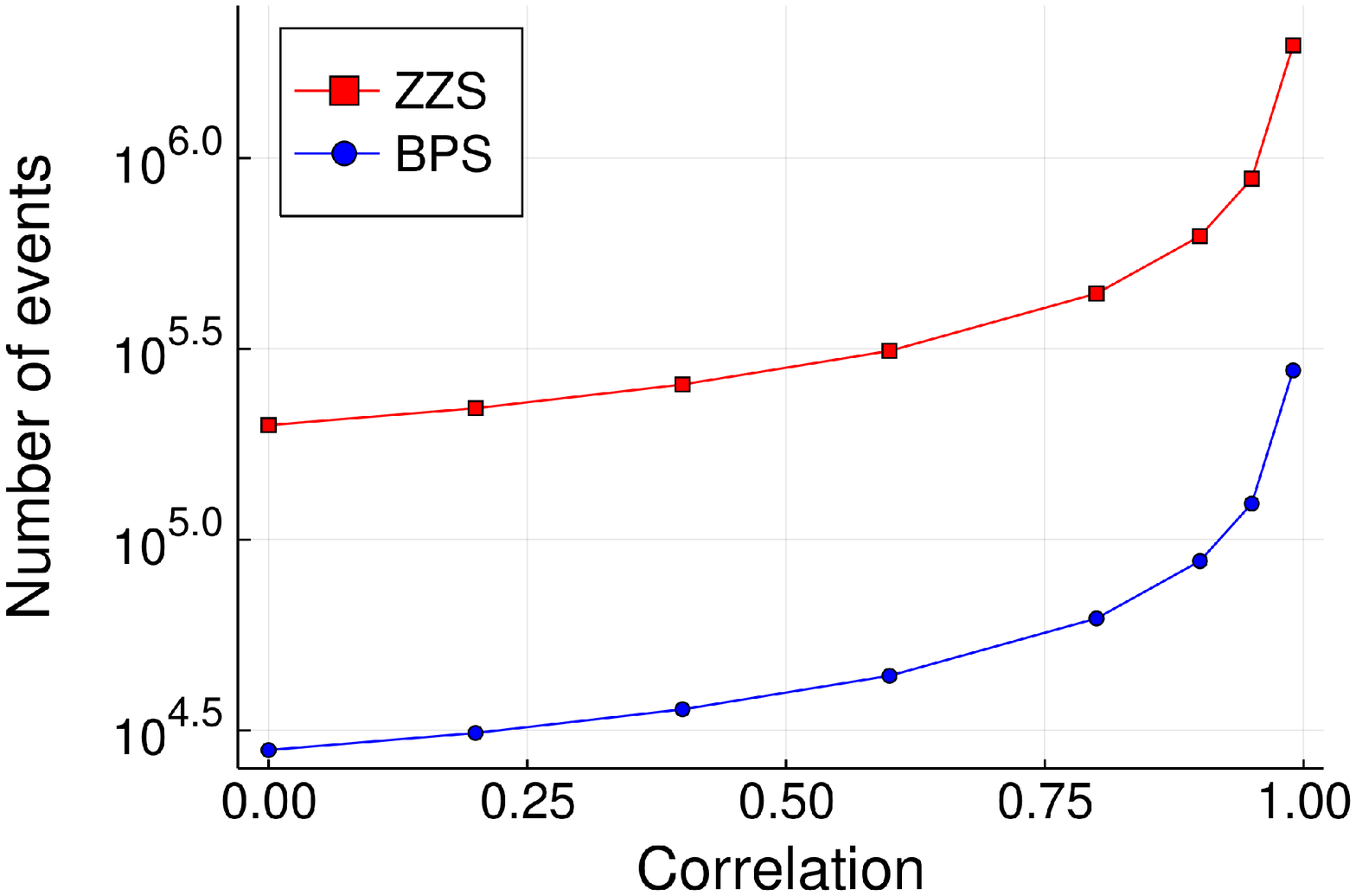}
\end{subfigure}
\caption{Root mean squared error (RMSE) for the radius statistic (left) and number of events (right) for $50$-dimensional Gaussian targets for several values of the correlation between all components. The continuous time horizon is $T=10^4$.}
\label{fig:intro}
\end{figure}
Figure \ref{fig:intro} illustrates this phenomenon in the case of Gaussian targets as a function of the correlation between all components. The performance drop occurs  due to a combination of decreasing accuracy of the estimates and increasing computational complexity of the algorithms, which is implied by the growing number of velocity change events.
Our idea to improve this issue is to let the process learn (part of) the covariance matrix $\Sigma_\pi$ and take advantage of it to enhance the mixing properties.
The covariance estimate is used to linearly transform the target in such a way that it becomes more isotropic, i.e. with unitary covariance matrix. The standard samplers are then run targeting the transformed version of $\pi$, and the obtained sample is finally re-transformed to be approximately from $\pi$. 
The procedure is applied iteratively, and once a new estimate of $\Sigma_\pi$ is computed, it is used to define the linear transformation of $\pi$. The estimate will eventually be close to the true covariance matrix and the process targets an isotropic version of $\pi$. This scheme can also be interpreted as an application of a linear transformation directly to the standard ZZS and BPS. The natural applications of this procedure are then targets with elliptical level curves, although performance improvements can be observed also for distributions that deviate from this class.  

Furthermore, we address the problem of automatically choosing the rate of velocity refreshments. In \cite{bierkens2018highdimensional} the authors consider a BPS with a specific target and  derive that in the limit it is optimal to have a ratio of number of refreshments over number of total events of $0.7812$. In this paper we use this criterion as a basis to define an adaptive algorithm that iteratively adjusts the refreshment rate to obtain the right ratio. This same adaptive scheme can be applied to the ZZS. Indeed it could be the case that adding velocity refreshments to the ZZ process leads to a faster convergence to the invariant measure, although it comes with a larger asymptotic variance. For an analysis of these two results we refer to \cite{vialaret2019convergence} and \cite{bierkens_duncan_2017} respectively. 

Both the schemes we discussed take advantage of what the process has learned up until the current time to tune a parameter or to improve the performance. This idea is at the core of adaptive Markov chain Monte Carlo algorithms. For an introduction to this area we refer to \cite{andrieu_thoms,robros_examples}, while standard results on convergence of these methods can be found among others in \cite{BaiRob,fort_conv,fort_CLT,Roberts2007CouplingAE}. It is well known that adaptive MCMC algorithms can lose the right invariant measure if not applied with care (see for instance \cite{Roberts2007CouplingAE} and the examples therein). Therefore we study in depth the convergence properties of the proposed algorithms. To fit into the existing adaptive MCMC literature we let the adaptation happen at fixed points in time. The main challenge consists of establishing a simultaneous geometric drift condition for a family of BPS's (see Lemma \ref{lemma:driftABP}) and a simultaneous small set condition for a family of ZZ processes (see Lemma \ref{lemma:smallZZ}). The former result is obtained taking advantage of the Lyapunov function found in \cite{BPS_Durmus}, while the latter is proved by extending the one-dimensional case considered in Lemma \ref{lemmma:smallsetZZ1D}.
The ergodicity and a law of large numbers for the proposed adaptive PDMC algorithms are then established in Theorem~\ref{thm:ergADZZ} and Theorem~\ref{thm:ergodicityABPS}.

In Section~\ref{sec:adap_schemes} we introduce the adaptive schemes, while in Section~\ref{sec:theory} the theoretical aspects of the algorithms are studied. The proofs of the two main theorems can be found in Section \ref{sec:proofs}. In Section~\ref{sec:num_exp} the adaptive BPS and ZZS are tested empirically on various Gaussian targets, on a Bayesian logistic regression problem with correlated data, and on a mixture of two Gaussian distributions. Appendix \ref{app:implementation} gives details on the implementation of adaptive PDMC (with and without subsampling), as well as an alternative adaptive scheme for the refreshment rate.  Appendix \ref{sec:other_proofs} contains the proofs of other theorems stated in Section \ref{sec:theory}, together with more technical results. 

\section{The adaptive schemes}\label{sec:adap_schemes}

We are interested in building adaptive strategies to make the ZZS and the BPS choose the refreshment rate themselves and/or converge faster to the target density. We begin with an introduction of the standard versions of both samplers, followed by a characterisation of the preconditioned processes in Section \ref{sec:preconditioned_pdmps} and a discussion on the choice of the transformation matrix in Section \ref{sec:transf_matrix}.
Finally, the adaptive algorithms are defined in Section \ref{sec:algorithms}.

\subsection{The standard ZZS and BPS}\label{sec:std_algs}
Let the target density $\pi$ be defined on $\mathcal{X} \subset \mathbb{R}^d$ as
\begin{equation}\notag
    \pi(\xi)=\frac{1}{Z} \exp{(-U(\xi))},
\end{equation}
where $U(\xi)$ is called potential or energy function, and $\xi \in \mathcal{X}$. 
Let us now define the standard ZZS with invariant measure $\pi$. 
Throughout the paper the position and velocity vectors at time $t$ of the standard processes, both ZZS and BPS, are denoted respectively as $\Xi(t)$ and $\Theta(t)$. We distinguish between the Zig-Zag Sampler (i.e. a PDMC algorithm) and the Zig-Zag Process (i.e. the Markov process on which the algorithm is based). The Zig-Zag (ZZ) process is a Markov process $(\Xi(t),\Theta(t))_{t \ge 0}$ with state space $E=\mathbb{R}^d\times \{-1,+1\}^d$ that follows a linear trajectory in the position space with velocity $\Theta$ until one of its $d$-inhomogeneous Poisson clocks rings. When the $i$-th clock rings the velocity of the $i$-th coordinate switches sign. In mathematical terms, this means that the new velocity vector is $F_i \Theta$, where $F_i: \{ -1,+1 \}^d \to \{ -1,+1 \}^d$ is the flip operator defined as
\begin{equation}\notag
(F_i \theta)_j = 
    \begin{cases}
        \theta_j \qquad &\text{if } j\neq i, \\
        -\theta_j \qquad &\text{if } j= i.
    \end{cases}
\end{equation}
The rates of the Poisson clocks are defined for $i=1,\dots,d$ as
$$
    \lambda_i(\xi,\theta)= (\theta_i \partial_i U(\xi) )_+ + \gamma_i(\xi,\theta) ,
$$
where $(\theta_i \partial_i U(\xi) )_+ = \max (0,\theta_i \partial_i U(\xi))$ and $\gamma_i(\xi,\theta):E \to \mathbb{R}_+$ is called excess switching rate and is such that $\gamma_i(\xi,\theta)=\gamma_i(\xi,F_i\theta)$. It was shown in \cite[Theorem 2.2]{ZZ} that this choice of switching rates ensures that $\mu=\pi\times \text{Unif}(\{-1,+1 \}^d)$ is the stationary distribution of the process. 
Moreover, the ZZ process is characterised by its infinitesimal generator
\begin{equation}\label{eq:gen_stdZZ}
    \mathcal{L} f(\xi,\theta) = \langle \theta, \nabla f(\xi,\theta) \rangle + \sum_{i=1}^d \lambda_i(\xi,\theta) (f(\xi,F_i \theta)-f(\xi,\theta)),
\end{equation}
where function $f$ should be in the domain of the generator $\mathcal{D}(\mathcal{L})$. The linear trajectories are represented in the first term, while the second term represents the event of a velocity flip.

Similarly, we denote the BPS as $(\Xi(t),\Theta(t))_{t \ge 0}$, but in this case the state space is $E=\mathbb{R}^d\times \mathbb{R}^d$ and so $\Theta(t) \in \mathbb{R}^d$. In contrast to the ZZS, the BPS has two Poisson clocks. The first one depends on the  gradient of the energy function and has inhomogeneous rate
$
    \lambda (\xi,\theta) = \langle \theta, \nabla U(\xi) \rangle _+ = \max(0, \langle \theta, \nabla U(\xi) \rangle).
$
At event time the particle is reflected on the level curve of the potential $U$ following an elastic bounce, and thus preserving the norm of the velocity vector. After a bounce the new velocity vector is given by 
$$
    R(\xi) \theta = \theta - 2 \frac{\langle \nabla U(\xi), \theta \rangle}{||\nabla U(\xi)||^2} \nabla U(\xi).
$$
It was observed in \cite{BPS} that the BPS needs refreshments of the velocity vector in order to be ergodic. This brings us to the second Poisson clock, which has rate $\lambda_{r}:E\to \mathbb{R}_+$. This is referred to as refreshment rate and should thus be strictly positive. When this clock rings, the velocity vector is refreshed by sampling from a distribution $\psi$. Possible choices are the Gaussian distribution $\psi = \mathcal{N}(0,\mathbbm{1}_d)$, or $\psi = \text{Unif}(\mathbb{S}^{d-1})$, where $\mathbb{S}^{d-1}$ is the surface of the unit hypersphere. In the analysis that follows we focus on the former distribution. 
The infinitesimal generator of the BPS is defined for any $f \in \mathcal{D}(\mathcal{L})$ as
\begin{equation}\notag
    \begin{aligned}
    \mathcal{L} f(\xi,\theta) &= \langle \theta , \nabla f(\xi,\theta) \rangle + \lambda(\xi,\theta) \big(f(\xi,R(\xi)\theta)-f(\xi,\theta) \big) \\ 
    & \,\, + \lambda_{r}(\xi,\theta) \int (f(\xi,\theta')-f(\xi,\theta)) \psi(d\theta').
    \end{aligned}
\end{equation}
The invariant measure of the BPS defined by the infinitesimal generator above is $\mu = \pi \times \psi$ as shown in \cite[Proposition 1]{BPS}.
\subsection{Applying a linear transformation to the ZZS and BPS}\label{sec:preconditioned_pdmps}
In this section we suppose a matrix $M\in \mathbb{R}^{d\times d}$ is given. We then wish to define a transformation scheme encoded by $M$, which we should think of being such that, for a suitable choice of $M$, it gives a  ``more isotropic" version of the target, and analyse its effects on the PDMC samplers. 

The transformation scheme encoded by $M$ consists of a linear transformation of the state space, which defines a new target distribution $\Tilde{\pi}_M$ given by
\begin{equation*}
    \tilde{\pi}_M(\xi) \coloneqq \frac{1}{\tilde{Z}_M} \exp(- \tilde{U}_M(\xi)),
\end{equation*}
with $\tilde{U}_M(\xi)=U(M\xi)$ and $\tilde{Z}_M=Z/\lvert \text{det} M \rvert$. The idea is to apply the transformation to the target distribution $\pi$ and simulate the standard PDMC sampler $(\Xi_t, \Theta_t)_{t\geq 0}$ with the resulting
target $\Tilde{\pi}_M$. 
Then the last thing to do is transform the obtained sample, which is approximately from $\Tilde{\pi}_M$, by applying the inverse transformation.
For this reason it is important that the matrix $M$ is invertible, and thus that we can go from one state space to the other.
This procedure is illustrated in Figure~\ref{fig:trans_scheme}. 
An equivalent option is to simulate directly the process $(X_t,\Theta_t)_{t\geq 0}$ that results from the scheme in Figure~\ref{fig:trans_scheme}. We will conveniently alternate between these two formulations when studying the ergodic properties of the samplers, while we will use the latter formulation for our experiments. The dynamics of process $(X_t,\Theta_t)_{t\geq 0}$ are studied in the remainder of this section. 
\begin{figure}[h]
\centering
\normalsize
\begin{tikzpicture}[every node/.style={midway}]
\matrix[column sep={20em,between origins},
        row sep={7em}] at (0,0)
{ \node(R)   {$\pi(x)$}  ; & \node(S) {$\Tilde{\pi}_M(\xi)$}; \\
  \node(R/I) {$(X_t,\Theta_t)$}; & \node(ZZ) {$(\Xi_t, \Theta_t)$};  \\ };
\draw[<-] (R/I) -- (ZZ) node[anchor=south]  {$X_t=M\Xi_t$};
\draw[->] (R)   -- (S) node[anchor=south] {$\xi = M^{-1} x$};
\draw[->] (S)   -- (ZZ) node[anchor=west] {$\Tilde{\lambda}_M(\xi,\theta)$};
\draw[->] (R/I)   -- (R) node[anchor=east] {$\lambda(x,\theta)$};
\end{tikzpicture} 
\caption{Transformation scheme.}
\label{fig:trans_scheme}
\end{figure}
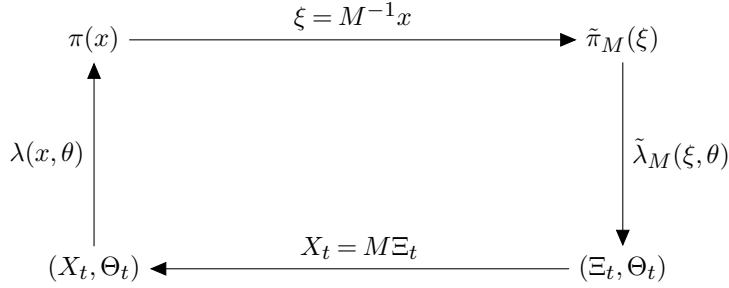

Let us first focus on the case in which $(\Xi_t, \Theta_t)_{t\geq 0}$ is a standard ZZS with excess switching rate $\gamma$. In this case the switching rates are $\Tilde{\lambda}_{M,i} (\xi,\theta) = (\theta_i \partial_i \Tilde{U}_M(\xi) )_+ + \gamma_{i}(\xi,\theta)$ for $i=1,\dots,d$. Note that, unless stated otherwise, we will always use a tilde to indicate quantities related to the standard PDMC samplers with transformed target. In the next proposition we find the generator of the preconditioned ZZS. For a characterisation of the domain of the extended generator we refer to \cite[Theorem 26.14]{Davis1993MarkovM}.
\begin{proposition}\label{prop:generatorZZ}
Let $M \in\mathbb{R}^{d\times d}$ be an invertible matrix. Let $(\Xi(t),\Theta(t))_{t\geq0}$ be a standard ZZS with target $\Tilde{\pi}_M$ and excess switching rates $\gamma:E \to \mathbb{R}^d_+$. The process $(X(t),\Theta(t))_{t\geq0} = (M \Xi(t),\Theta(t))_{t\geq0}$ has extended generator $(\mathcal{L}_M,\mathcal{D}(\mathcal{L}_M))$ where for any $h \in \mathcal{D}(\mathcal{L}_M)$
\begin{equation}\label{eq:genAZZ}
    \mathcal{L}_M \, h(x,\theta) = \langle M \theta , \nabla_x h(x,\theta) \rangle + \sum_{i=1}^d \lambda^M_i(x,\theta) (h(x,F_i \theta)-h(x,\theta)),
\end{equation}
in which for $i=1,\dots,d$
\begin{equation}\label{eq:switch_rates_AZZ}
    \lambda^M_i(x,\theta)= \tilde{\lambda}_{M,i} (M^{-1}x,\theta) = (\theta_i \langle M_i, \nabla U(x) \rangle )_+ + \gamma_i(M^{-1}x,\theta),
\end{equation}
where $M_i$ denotes the $i$-th column of $M$.
\end{proposition}

\begin{proof}
All the proofs of this section are postponed to Appendix~\ref{app:proofs_section_2}.
\end{proof}
\noindent Proposition~\ref{prop:generatorZZ} shows that the transformed process is again a PDMP with linear trajectories between jumps. The transformation affects the velocity of the process, which is now $v=M\theta$, and the switching intensities, which are as defined in \eqref{eq:switch_rates_AZZ}. In particular, the available velocities for a fixed transformation $M$ are in the following set:
$$
    V \coloneqq  \big\{ v: v= M \theta , \, \theta \in \{-1,+1\}^d \big\}.
$$
When a switch of the $i$-th velocity of the underlying standard ZZ process happens, the velocities of the transformed process change according to the operator $\bar{F}_i v = M F_i \theta = M F_i (M^{-1} v)$ for $i=1,\dots, d$. Therefore, all components of the velocity are possibly affected by any single event. If $M$ is diagonal the behaviour is more similar to the standard ZZ process and in particular we have that $\bar{F}_i\equiv F_i$. In the proposition below we check that $(X(t),\Theta(t))_{t\geq 0}$ targets the correct density function. 
\begin{proposition}\label{prop:invarianceAZZ}
Consider the same setting of Proposition~\ref{prop:generatorZZ}. Then, for any invertible $M \in\mathbb{R}^{d\times d}$, the modified ZZ process $(X(t),\Theta(t))_{t \ge 0}$ has invariant distribution $\mu=\pi \times \textnormal{Unif}(\{-1,+1 \}^d )$.
\end{proposition}

Let us now apply the same transformation scheme shown in Figure~\ref{fig:trans_scheme} to the BPS. In this case the switching rate of the standard BPS  with target $\tilde{\pi}_M$ is $\tilde{\lambda}_M(\xi,\theta) = \langle \theta, \nabla\tilde{U}_M(\xi)\rangle_+$, while reflections on the level curves of $\Tilde{U}_M$ are obtained by applying operator $\Tilde{R}_M$.
The following result, analogous to Proposition~\ref{prop:generatorZZ}, studies the transformed process.
\begin{proposition}\label{prop:generatorBPS}
Let $M \in\mathbb{R}^{d\times d}$ be an invertible matrix. Let $(\Xi(t),\Theta(t))_{t\geq0}$ be a standard BPS with target $\Tilde{\pi}_M$ and refreshment rate $\lambda_r: E \to \mathbb{R}_+$. The process $(X(t),\Theta(t))_{t\geq0} = (M \Xi(t),\Theta(t))_{t\geq0}$ has extended generator $(\mathcal{L}_M,\mathcal{D}(\mathcal{L}_M))$ where for any $h \in \mathcal{D}(\mathcal{L}_M)$
\begin{equation}\label{eq:genABPS}
    \begin{aligned}
    \mathcal{L}_{M} h(x,\theta) &= \langle M \theta , \nabla_x h(x,\theta) \rangle + \lambda_M(x,\theta) \big(h(x,R_M(x)\theta)-h(x,\theta) \big) \\ 
    & \,\, + \lambda_r(M^{-1} x,\theta) \int (h(x,\theta')-h(x,\theta)) \psi(d\theta'),
    \end{aligned}
\end{equation}
where we defined 
\begin{equation}\label{eq:refl_rate_ABPS}
    \lambda_M(x,\theta) = \Tilde{\lambda}(M^{-1}x,\theta) = \langle M\theta, \nabla U(x) \rangle_+
\end{equation}
and 
\begin{equation}\label{eq:refl_op_ABPS}
    R_M(x) \theta = \tilde{R}_M(M^{-1}x) \theta = \theta - 2 \frac{\langle M^T \nabla U(x), \theta \rangle}{||M^T \nabla U(x)||^2} M^T \nabla U(x) .
\end{equation}
\end{proposition}
\noindent Once again the true velocity of the process is $v=M\theta$. When a velocity refreshment takes place the new $\theta$ is sampled from $\psi=\mathcal{N}(0_d,\mathbbm{1}_d)$, while the new velocity $v$ is $v=M \theta \sim \mathcal{N}(0_d,M M^T)$. Observe also that the reflection rate is $\lambda_M(x,\theta)= \langle v, \nabla U(x) \rangle_+$ and thus preserves the same structure as in the standard BPS. It follows that the complexity of the simulation of event times remains unchanged.
Finally consider the reflection operator in \eqref{eq:refl_op_ABPS}. This corresponds to a reflection in the opposite direction to the gradient in the transformed space, i.e. $\nabla_\xi \tilde{U}(\xi) = M^T \nabla U(x)$. After the bounce the process moves with velocity 
\begin{equation}\notag
    v = M (R_M(x) \theta) = v - 2 \frac{\langle \nabla U(x), v \rangle}{||M^T \nabla U(x)||^2} M M^T \nabla U(x).
\end{equation}
This implies that  $\langle v,\nabla U(x)\rangle = \langle M (R_M(x)\theta), \nabla U(x)\rangle = - \langle M\theta,\nabla U(x) \rangle = -\langle v,\nabla U(x)\rangle$.
\begin{proposition}\label{prop:invarianceABPS}
Consider the same setting of Proposition~\ref{prop:generatorBPS}. Then, for any invertible $M \in\mathbb{R}^{d\times d}$, the transformed BPS $(X(t),\Theta(t))_{t \ge 0}$ has invariant distribution $\mu=\pi \times \psi$.
\end{proposition}

\subsection{Choosing the transformation matrix}\label{sec:transf_matrix}
As explained above, we wish to transform the target to mitigate its anisotropies. To this end, some alternative choices of the transformation matrix $M$ are the following:
\begin{enumerate}[label=\alph*)]
    \item $M = \sqrt{\text{Cov}_\pi(X)}$: this transformation is such that the target  $\Tilde{\pi}_M$ has unitary covariance matrix. The downside of this choice is the additional $\mathcal{O}(d^3)$ computations that are introduced by the calculation of the square root of the covariance;
    \item $M$ is a rotation matrix ($\det M  = 1$) such that the transformed density has a certain angle. Although an interesting case, it is not clear whether there is an optimal angle that speeds up the convergence;
    \item $M$ is the diagonal matrix with $M_{ii} = \sqrt{\text{Var}_\pi (X_i)}$ for any $1 \le i \le d$. This choice introduces $\mathcal{O}(d)$ computations due to the square root of the variances, which is a negligible additional computational burden. However, correlations in the target are not picked up and only a rescaling of the axes is performed. The main advantage of this choice is that the scenario in which some components are explored quickly and others slowly is avoided.
\end{enumerate}
Both the first and the third option can potentially change the expected number of switching events, and this could be an inconvenience in certain settings. It is not difficult to modify these transformations in such a way that the expected switching rate is enforced to be close to that of the original standard PDMC algorithm. For example, consider the transformed BPS with generator as in Proposition~\ref{prop:generatorBPS}. Then in stationarity we have
$$
    \mathbb{E}_\mu(\langle M\Theta,\nabla U(X))\rangle_+) = \mathbb{E}_\pi \lVert M^T \nabla U(X)\rVert_2 \leq \, \lVert M\rVert_2 \, \mathbb{E}_\pi \lVert \nabla U(X)\rVert_2.
$$
Since the standard case corresponds to $M=\mathbb{I}_d$, we can normalise any $M$ by dividing it by its Frobenius norm. Then the upper bound  is the same for all such choices of $M$ and the expected switching intensity will be close to the standard case.  This does not make a difference from a computational point of view as it just amounts to reparametrization of the time parameter, but prevents unpredictable behaviour of the algorithm.

Naturally the options above are not available in practice as the covariance matrix is unknown. It is the goal of the next section to propose an adaptive scheme that overcomes this issue.

\subsection{Adaptive PDMC algorithms}\label{sec:algorithms}
In the previous sections we defined the transformation scheme and we discussed the effect it has on the underlying process, together with different choices of the preconditioning matrix. We now describe how this idea can be applied in practice by designing an adaptive PDMC algorithm. Our general strategy is to simulate the process in continuous time and store the states of the process at discrete times. Then at predefined times the stored states are used to update the adaptation parameters. In addition to the adaptive preconditioner, we incorporate an adaptation of the refreshment rate, which makes its choice automatic.

Let us then define a family of Markov semigroups by $\mathcal{P} \coloneqq\{(P_\Gamma^t)_{t\geq 0}: \Gamma \in \mathcal{Y}\}$, in which $\Gamma$ is the adaptation parameter, $\mathcal{Y}$ is a compact space, and $(P_\Gamma^t)_{t\geq 0}$ is the semigroup of a modified ZZS or BPS. The modification is given by the adaptation parameter, which is then $\Gamma =(M,\lambda_r)$ for BPS and  $\Gamma =(M,\gamma)$ for ZZS. Thus $\mathcal{Y}$ is a suitable compact space of preconditioners and refreshment rates/excess switching rates. Naturally, it is also possible to choose $\Gamma = M$ or $\Gamma = \lambda_r$ only. 
Now that we have defined a family of Markov processes, we define a rule that establishes how to choose a $P_\Gamma \in \mathcal{P}$ at every iteration. Let us begin by introducing a discretisation step $\Delta t$, which defines a discretisation of the time variable. At each time step $n\in\mathbb{N}$, which corresponds to continuous time $t=n\Delta t$, the adaptive scheme can update the parameter $\Gamma_n$ based on the new information available, that is the new state of the process $(X_n,\Theta_n)$. This defines a random sequence $\{\Gamma_n\}_{n\geq 0}$. Once $\Gamma_n$ is computed, the next state of the process is given by $(X_{n+1},\Theta_{n+1}) \sim P^{\Delta t}_{\Gamma_n}((X_n,\Theta_n),\cdot)$. Then one updates the parameter, obtains the next state of the process, and so on. 

The definition above defines the core ideas, which are written in pseudo-code form in Algorithm~\ref{alg_APDMP}. A few issues remain to be clarified. A first question is how to simulate the PDMP semigroup of either ZZS or BPS. Details on how the processes can be simulated in the case of a target with dominated Hessian can be found in Appendix~\ref{sec:bdd_hessian}. In a big data setting (large number of observations, moderate dimensionality of the problem) it can be beneficial to take advantage of subsampling techniques that can be implemented with PDMC algorithms. In Appendix~\ref{sec:subsampling} details can be found on how to make use of subsampling in the context of the adaptive schemes here discussed. For further information on the general implementation of ZZS and BPS we refer to \cite{ZZ,BPS}.
A second aspect of Algorithm~\ref{alg_APDMP} we focus on is the introduction of set $B$, and thus of the auxiliary sequence of random variables $(Q_n)_{n\geq 0 }$. The idea is to update the adaptation parameter $\Gamma_n$ only if $(X_n,\Theta_n) \in B$. This is useful from a theoretical point of view as it ensures that the process remains bounded in probability. 
Note that set $B$ is defined by the user and can be chosen large. 
The auxiliary variable $Q_n$ is updated even if the process is outside of $B$ and then, as soon as the process enters $B$, or if it was already in $B$, we let $\Gamma_n = Q_n$. A third characteristic of Algorithm~\ref{alg_APDMP} is the sequence $\{ p_n\}_{n\geq 0}$. This is a sequence for which $p_n\in [0,1]$ for all $n\in\mathbb{N}$ and  $p_n \to 0$ as $n \to \infty$. The meaning is that at time step $n$ we update the parameter $\Gamma_n$ with probability $p_n$ (assuming $(X_n,\Theta_n)\in B)$, and with remaining probability $(1-p_n)$ we set $\Gamma_n = \Gamma_{n-1}$. This choice is helpful when proving ergodicity of the adaptive scheme, as it enforces that the quantity of adaptation diminishes and eventually vanishes.
\begin{algorithm}[t]
\caption{Adaptive PDMC sampler\label{alg_APDMP}}
\begin{algorithmic}[1]
\normalsize

    \State \textbf{Input:} family of kernels $\mathcal{P}=\{P_\Gamma: \Gamma \in \mathcal{Y}\}$
    \State \textbf{Input:} initial condition $(x,\theta) \in E$, $\Gamma_0 \in \mathcal{Y}$, set $B$, $\Delta t$, $\{ p_n\}_{n\geq 0}$, number of steps $N$
    \State \textbf{Output:} sequence of discrete samples $\{X_n,\Theta_n \}_{n = 0}^N$
    \State Initialise $n=0$, $(X_0,\Theta_0)=(x,\theta)$, $Q_0=\Gamma_0$
    \While{$n \le N$}
        \State $(X_{n+1},\Theta_{n+1}) \sim P^{\Delta t}_{\Gamma_{n}} ((X_{n}, \Theta_{n}),\cdot  )$
        \State $Q_{n+1} = \textbf{update}(Q_{n},(X_{n+1},\Theta_{n+1})) $
        \If{$(X_{n+1},\Theta_{n+1}) \in B$}
        \State With probability $p_n$, set   $\Gamma_{n+1} = Q_{n+1}$ \\
               \qquad \quad With probability $1-p_n$, set  $\Gamma_{n+1} = \Gamma_n$
        \EndIf
        \State $n =n +1$
    \EndWhile


\end{algorithmic}
\end{algorithm}

The function $\textbf{update}(Q_n,(X_{n+1},\Theta_{n+1}))$ outputs the updated parameter given the new observation $(X_{n+1},\Theta_{n+1})$. As suggested in \cite{andrieu_thoms}, the estimation of the covariance matrix can be done sequentially, or online, by applying 
\begin{align}\label{eq:update_covariance}
\begin{split}
    \hat{\mu}_{n+1} &= \hat{\mu}_{n} + r_{n+1}(X_{n+1}-\hat{\mu}_{n}),  \\
    \hat{\Sigma}_{n+1} &= \hat{\Sigma}_{n} + r_{n+1}((X_{n+1}-\hat{\mu}_{n}) (X_{n+1}-\hat{\mu}_{n})^T -\hat{\Sigma}_{n} ) .
\end{split}
\end{align}
Here $\{ r_n\}_{n\geq0}$ is a positive, decreasing sequence such that $r_n \to 0$ as $n \to \infty$. In our simulations we choose $r_n=1/n$. Equation \eqref{eq:update_covariance} is then used to define $M_{n+1}$ such that $\hat{\Sigma}_{n+1} = M_{n+1}^T M_{n+1}$. The same principle can be used if one is not interested in estimating the full covariance matrix, but only  the diagonal, or more generally only a subset of it. More advanced estimation techniques can be employed to preserve any existing conditional independence structure in the target, as discussed in \cite{wallin_bolin}.  We remark that $\hat{\Sigma}_{n+1}$ needs to be positive definite in order for $M_{n+1}$ to be invertible as required. This property is achieved by choosing $\hat{\Sigma}_0 = \mathbf{1}_{d\times d}$, i.e. the identity matrix, and then observing that the second equation in \eqref{eq:update_covariance} can be reformulated as 
\begin{equation*}
    \hat{\Sigma}_{n+1} = (1-r_{n+1})\hat{\Sigma}_{n} + r_{n+1}(X_{n+1}-\hat{\mu}_{n}) (X_{n+1}-\hat{\mu}_{n})^T.
\end{equation*}
Indeed $(X_{n+1}-\hat{\mu}_{n}) (X_{n+1}-\hat{\mu}_{n})^T$ is non-negative definite and by induction $\hat{\Sigma}_{n}$ is positive definite, and therefore $\hat{\Sigma}_{n+1}$ is itself positive definite. Moreover, in Section \ref{sec:theory} we will see that to show ergodicity of the adaptive algorithms it is required that $M_{n}$ lies in a compact space of positive definite matrices. Observe that positive definiteness follows from the fact that $M_n$ is the square root of a positive definite matrix, while it is sufficient to set bounds on the norm of $M_n$ in order to force it to be in a compact space. In particular we can impose that $M_n$ is not updated if the norm of the new estimate is outside of a user chosen interval $[M_{\text{min}},M_{\text{max}}]$. As $M_{\text{min}}$ and $M_{\text{max}}$ can be chosen arbitrarily small and large respectively, this condition is not restrictive in practice, although the choice of the cut-off value may influence convergence properties of the algorithm.
Then refreshment rate of the BPS is assumed to be constant and is updated iteratively as follows. At time step $n$, $n_\text{refl}$ reflections took place and thus we estimate the average reflection rate as $\overline{\lambda}_\text{refl}(n) = n_\text{refl} / (n\Delta t)$. Therefore, using the optimality criterion in \cite{bierkens2018highdimensional} we have
\begin{equation}
    \frac{\lambda_r^n}{\lambda_r^n + \overline{\lambda}_\text{refl}(n)} = 0.7812  \implies \lambda_r^n = \frac{0.7812}{0.2188}\,\, \overline{\lambda}_\text{refl}(n).
\end{equation}
An alternative adaptive strategy for the refreshment can be found in Appendix \ref{sec:adap_strategy_appendix}.
The scheme above can be applied to the excess switching rate of the ZZS. Although the analysis in \cite{bierkens_duncan_2017} suggests that the best choice in terms of asymptotic variance is $\gamma \equiv 0$, adding some diffusivity could speed up the convergence to the target measure. In practice the user can select the wanted ratio of velocity switches over total events and proceed as above. However, a criterion to choose this ratio is currently unavailable for ZZS, and thus in this paper we limit ourselves to a theoretical study of this option.

Finally, we remark that in practice it is not reasonable to update the parameters at every iteration. The main reason for this is the computational cost of such an operation. In the most general case, the task of learning all components of $\Sigma$ takes $\mathcal{O}(d^2)$ operations, while the computation of its square root, which is needed to obtain the transformation matrix $M$, is an $\mathcal{O}(d^3)$ operation.
Therefore it is rather inconvenient to perform this at every time step. To avoid this issue it is sufficient to define the adaptive scheme such that adaptations happen every $n_{\textnormal{adap}}$ time steps, where $n_{\textnormal{adap}}$ is a user-defined integer. A possible choice is for instance $n_{\textnormal{adap}}=1000$. This modification is beneficial also because it allows the process to explore the target distribution before updating the parameters. Similarly, it is reasonable to update the refreshment rate based on the previous $n_{\textnormal{adap}}$ time steps, as in the long term this allows to stabilise around the wanted ratio. The covariance matrix can be updated as in \eqref{eq:update_covariance} by simply processing the entire batch of $n_{\textnormal{adap}}$ data points one at a time.

\section{Theoretical results}\label{sec:theory}
In the context of adaptive MCMC algorithms, convergence to the target density is usually proved with simultaneous drift conditions and small set conditions. In Section \ref{sec:theory_AMCMC} we introduce the notation and the main existing theorems we make use of, and we extend these results to more general conditions in Theorem~\ref{thm:erg_adaptiveMCMC}. In Section \ref{sec:ergodicity_continuous} we state Theorem~\ref{thm:ergodicity_continuoustime}, which shows ergodicity for an adaptive MCMC algorithm based on a continuous time process. In this result, the assumptions are formulated directly in continuous time. Finally, Theorems~\ref{thm:ergADZZ} and~\ref{thm:ergodicityABPS} in Section \ref{sec:ergodicity_adapPDMP} show that the adaptive ZZS and the adaptive BPS discussed in Section~\ref{sec:algorithms} are ergodic and satisfy a weak law of large numbers under reasonable growth conditions on the potential.

\subsection{Theory of adaptive MCMC}\label{sec:theory_AMCMC}
We denote the parameter that specifies the kernel as $\Gamma \in \mathcal{Y}$. At time step $n$ a $\mathcal{Y}$-valued random variable $\Gamma_n$ determines which transition kernel will be used to move to the next step. From here on each Markov transition kernel $P_\Gamma$ is assumed to define a Markov chain that has $\mu$ as stationary measure, and moreover it is aperiodic and irreducible. An adaptive MCMC algorithm is then said to be \textit{ergodic} if
\begin{equation}\label{eq:ergodicityAdap}
    \lim_{n \to \infty} \lVert P( Z_n \in \cdot \,\,| z_0, \Gamma_0) - \mu(\cdot) \rVert_{\textnormal{TV}} = 0 \qquad \text{for all } z_0 \in E, \Gamma_0 \in \mathcal{Y} ,
\end{equation}
where $||\cdot||_{\textnormal{TV}}$ is the total variation distance, i.e. $\lVert \mu - \nu \rVert_{\text{TV}} = \sup_{A\subset E} \lvert \mu(A)-\nu(A)\rvert$. A crucial quantity turns out to be the $\varepsilon$-convergence time function $M_\varepsilon : E \times \mathcal{Y} \to \mathbb{N}$, defined as
\begin{equation}\notag
    M_\varepsilon(z,\Gamma) = \inf \{n \geq 1: \lVert P^n_\Gamma(z,\cdot) - \mu(\cdot) \rVert_{\textnormal{TV}} \leq \varepsilon   \}.
\end{equation}
The next theorem, proved in \cite{Roberts2007CouplingAE}, is arguably the most important result for establishing ergodicity of adaptive MCMC methods.
\begin{theorem}[Theorem 2 in \cite{Roberts2007CouplingAE}] \label{thm:RobRos}
Consider an adaptive MCMC algorithm on a state space $E$ with adaption parameter in a space $\mathcal{Y}$. Let $\mu$ be stationary for $P_\Gamma$ for each $\Gamma \in \mathcal{Y}$. The adaptive algorithm is ergodic if the two following conditions hold:
\begin{enumerate}[label=(\alph*)]
    \item (Containment condition) For all $ z_0 \in E, \Gamma_0 \in \mathcal{Y}$, $\varepsilon > 0$ the sequence $\{M_\varepsilon(Z_n,\Gamma_n) \}_{n=0}^\infty$ is bounded in probability given $z_0, \Gamma_0$;
    \item (Diminishing adaptation) The following limit holds in probability: 
    \begin{equation}\label{eq:dim_adap}
        \lim_{n \to \infty} \left( \sup_{z\in E} \lVert P_{\Gamma_{n+1}}(z,\cdot) - P_{\Gamma_{n}}(z,\cdot)\rVert_{\textnormal{TV}} \right)=0.
    \end{equation}
\end{enumerate}
\end{theorem}
\noindent The boundedness of $\{M_\varepsilon(Z_n,\Gamma_n) \}_{n=0}^\infty$ can be rephrased as for all $ z_0 \in E, \Gamma_0 \in \mathcal{Y}$, $\delta > 0$, there exists $N \in \mathbb{N}$ such that $P(M_\varepsilon(Z_n,\Gamma_n) \leq N | z_0, \Gamma_0) \geq 1-\delta$, for all $n\in\mathbb{N}$. 

We are then interested in sufficient conditions that imply containment. A first case is the following, and was studied in \cite{BaiRob}.
\begin{assumption}[\cite{BaiRob}]\label{ass:SGE}
The family $\{P_\Gamma: \Gamma \in \mathcal{Y}\}$ is \textit{simultaneously geometrically ergodic} (SGE), i.e. there are $C \in \mathcal{B}(E)$, some integer $n_0 \geq 1$, a function $V:E \to[1,\infty)$, $\delta > 0$, $0 < \lambda <1$, and $b<\infty$, such that $\sup_{z \in C} V(z) < \infty$, $\mu(V) < \infty$, and
\begin{enumerate}[label=(\alph*)]
    \item $C$ is a uniform $(\nu_\Gamma,\delta,n_0)$-small set, i.e. for each $\Gamma$, there exists a probability measure $\nu_\Gamma(\cdot)$ on $C$ such that $P_\Gamma^{n_0} (z,\cdot) \geq \delta \nu_\Gamma(\cdot)$ for all $z \in C$;
    \item (\textit{simultaneous geometric drift condition}) $P_\Gamma V \leq \lambda V + b \mathbbm{1}_C$ for all $\Gamma \in \mathcal{Y}$.
\end{enumerate}
\end{assumption}
\noindent Then \cite[Theorem 3]{BaiRob} establishes that an SGE family satisfies the containment condition. In Section \ref{sec:ergodicity_adapPDMP} we use this result to show that containment holds for the adaptive ZZS when the class of preconditioners is restricted to diagonal matrices.

In practice it is often hard to show that the family of Markov kernels is SGE, as it is not trivial to find a Lyapunov function that satisfies the simultaneous geometric drift condition. In \cite{craiu2015} the authors introduced a way around this problem, although in a different context, and in \cite{chimisov2018adapting} this was applied to adaptive MCMC. The fundamental idea is that it is possible to weaken the simultaneous drift condition by allowing adaptations only at time steps $n$ at which the process $Z_n$ is inside of a compact set $B$. This means that, defining an auxiliary random process $\{ Q_n\}_{n\ge1}$ that contains the current adaptation parameter independently of the position of $Z_n$, $\Gamma_n$ is updated as
\begin{equation}\label{eq:adap_compact}
\Gamma_{n+1} = 
\begin{cases}
\Gamma_n \qquad &\text{if} \, Z_{n+1} \notin B, \\
Q_{n+1} \qquad & \text{if} \, Z_{n+1} \in B.
\end{cases}
\end{equation}
This modification avoids unbounded detours of the process by sticking to the same ergodic kernel once the process exits a fixed compact set. The compact set can be chosen arbitrarily large, and therefore in most applications the process will not exit from it.  

With this is mind we introduce the following sets of assumptions, which we show in Theorem~\ref{thm:erg_adaptiveMCMC} to be sufficient to enforce the containment condition.
\begin{assumption}\label{ass:mineZZ}
    Let $\{P_\Gamma: \Gamma \in \mathcal{Y}\}$ be a family of discrete time Markov chains with state space $E$. There are $C \in \mathcal{B}(E)$, an integer $n_0 \geq 1$, a class of functions $\{ V_\Gamma:E \to[1,\infty): \Gamma \in \mathcal{Y}  \} $, $\delta > 0$, $0 < \lambda <1$, and $b<\infty$, such that $\sup_{z \in C,\, \Gamma \in \mathcal{Y}} V_\Gamma (z) < \infty$, $\mu(V_\Gamma) < \infty$, and
    \begin{enumerate}[label=(\alph*)]
    \item $C$ is a uniform $(\nu_\Gamma,\delta,n_0)$-small set, i.e. for each $\Gamma\in\mathcal{Y}$, there exists a probability measure $\nu_\Gamma(\cdot)$ on $C$ such that $P_\Gamma^{n_0} (z,\cdot) \geq \delta \nu_\Gamma(\cdot)$ for all $z \in C$;
    \item  for each $\Gamma \in \mathcal{Y}$, $z \in E$, $P_\Gamma V_\Gamma(z) \leq \lambda V_\Gamma(z) + b \mathbbm{1}_C (z)$;
    \item the adaptation parameter is allowed to be updated only if the process is inside of a compact set $B$, as defined in \eqref{eq:adap_compact}.
\end{enumerate}
\end{assumption}
\begin{assumption}\label{ass:mineassBPS}
Let $\{P_\Gamma: \Gamma \in \mathcal{Y}\}$ be a family of discrete time Markov chains with state space $E$. There exist $\alpha, \lambda \in (0,1)$, $C_1 >0$, $C_2 > 2 C_1$, a class of functions $\{ V_\Gamma:E \to[1,\infty): \Gamma \in \mathcal{Y}  \} $ with $\mu(V_\Gamma) < +\infty$, such that
\begin{enumerate}[label=(\alph*)]
    \item for each $\Gamma \in \mathcal{Y}$, for all $x,y \in E$ such that $V_\Gamma(x) + V_\Gamma (y) \le C_2$ it holds that
    \begin{equation*}
        \lVert P_\Gamma(x,\cdot) - P_\Gamma(y,\cdot) \rVert_{\textnormal{TV}} \leq 2(1-\alpha);
    \end{equation*}
    \item for each $\Gamma \in \mathcal{Y}$ and for any $z\in E$, $P_\Gamma V_\Gamma(z) \le \lambda V_\Gamma(z) + C_1(1-\lambda)$;
    \item the adaptation parameter is allowed to be updated only if the process is inside of a compact set $B$, as defined in \eqref{eq:adap_compact}.
\end{enumerate}
\end{assumption}
\begin{theorem}\label{thm:erg_adaptiveMCMC}
Consider a family of discrete time Markov transition kernels $\{P_\Gamma : \Gamma \in \mathcal{Y} \}$. Assume that all kernels $P_\Gamma$ are aperiodic, irreducible, and have stationary measure $\mu$. Suppose the adaptive algorithm satisfies the diminishing adaptation condition, i.e. assumption (b) in Theorem~\ref{thm:RobRos}, and let either Assumption~\ref{ass:mineZZ} or Assumption~\ref{ass:mineassBPS} hold. Then the containment condition holds and the adaptive MCMC algorithm is ergodic. 
\end{theorem}
\noindent The proof can be found in Appendix~\ref{app:proof_erg_adaptiveMCMC}.
\begin{remark}
A weak law of large numbers (WLLN) for bounded and measurable functions follows immediately from containment and diminishing adaptation by Theorem 3.4 in \cite{pompe2018framework}. Therefore under the conditions of Theorem~\ref{thm:erg_adaptiveMCMC} a WLLN holds.
\end{remark}

\subsection{Convergence properties of adaptive MCMC algorithms based on continuous time Markov processes}\label{sec:ergodicity_continuous}
It could be the case, as it is in the present work, that one is interested in defining an adaptive scheme based on a family of continuous time Markov processes in continuous time. In this case a grid for the time variable needs to be introduced in order to indicate the times at which the adaptation occurs. In fact the adaptive chain only sees the process at times $m\Delta t$, where $\Delta t>0$ is the step size and $m\in \mathbb{N}$. Although the resulting chain is in discrete time, it is in most cases easier to work directly with the continuous time process. The following result, which is analogous to Theorem~\ref{thm:erg_adaptiveMCMC}, is helpful in this sense.

\begin{theorem}\label{thm:ergodicity_continuoustime}
Consider a family of Markov processes with generators $\{\mathcal{L}_\Gamma : \Gamma \in \mathcal{Y} \}$, each being irreducible, aperiodic, and having $\mu$ as invariant measure. Consider a grid for the time variable with step $\Delta t$.
Consider an adaptive scheme that at times $m\Delta t$, with $m\in \mathbb{N}$, chooses a process from the aforementioned family. Furthermore, suppose that the adaptive scheme satisfies the diminishing adaptation condition \eqref{eq:dim_adap} for $P \coloneqq P^{\Delta t}$.
Finally assume one of the following two sets of conditions holds:
\begin{enumerate}
\item There exist a set $C \in \mathcal{B}(E)$, $t_0>0$, a class of functions $\{ V_\Gamma:E \to[1,\infty): \Gamma \in \mathcal{Y}  \} $, $\delta > 0$, $A_1,A_2>0$, such that for each $\Gamma$, $\sup_{z \in C, \Gamma \in \mathcal{Y}} V_\Gamma (z) < \infty$, $\pi(V_\Gamma) < \infty$, and
\begin{enumerate}[label=(\alph*)]
    \item for each $\Gamma\in\mathcal{Y}$ there exists a probability measure $\nu_\Gamma$ such that $P_\Gamma^{t_0} (z,\cdot) \geq \delta \nu_\Gamma(\cdot)$ for all $z \in C$;
    \item for each $\Gamma\in\mathcal{Y}$ and $z\in E$ it holds that $\mathcal{L}_\Gamma V_\Gamma (z) \leq -A_1 V_\Gamma(z) + A_2 \mathbbm{1}_C(z)$;
    \item it holds that $\Delta t = m t_0$, for some $m\in\mathbb{N}$.
\end{enumerate}
\item There exist $A_1,A_2 >0$, $C_2 > 2 A_2/A_1$, a class of functions $\{ V_\Gamma:E \to[1,\infty): \Gamma \in \mathcal{Y}  \} $ with $\pi(V_\Gamma) < +\infty$, such that
\begin{enumerate}[label=(\alph*)]
    \item for each $\Gamma \in \mathcal{Y}$, for all $x,y \in E$ such that $V_\Gamma(x) + V_\Gamma (y) \le C_2$, there exists $\alpha,t_0>0$ such that
    \begin{equation*}
        \lVert P^{t_0}_\Gamma(x,\cdot) - P^{t_0}_\Gamma(y,\cdot) \rVert_{\textnormal{TV}} \leq 2(1-\alpha);
    \end{equation*}
    \item for each $\Gamma \in \mathcal{Y}$ and for any $z\in E$, $\mathcal{L}_\Gamma V_\Gamma(z) \le - A_1 V_\Gamma(z) + A_2$;
    \item it holds that $\Delta t = t_0$.
\end{enumerate}
\end{enumerate}
If the adaptation parameter is allowed to be updated only if the process is inside of a compact set $B$, as defined in \eqref{eq:adap_compact}, then the adaptive algorithm satisfies the containment condition and is thus ergodic.
\end{theorem}
\noindent The proof of this theorem can be found in Appendix \ref{app:proof_ergodicity_continuoustime}.
\begin{remark}
The restrictions on the step size can be milder than as stated in Theorem~\ref{thm:ergodicity_continuoustime}. For instance, if the minorisation  condition (1a) of Theorem \ref{thm:ergodicity_continuoustime} holds for all $t\geq t_0$, then one is free to choose any step size $\Delta t>0$.
Furthermore, in both cases the assumption that the parameter can be updated only if the process is inside of a compact set at the  adaptation time can be dropped when a simultaneous geometric drift condition holds (Assumption~\ref{ass:SGE}(b)).
\end{remark}

\subsection{Convergence properties of adaptive PDMC algorithms}\label{sec:ergodicity_adapPDMP}
Relying on Theorem~\ref{thm:ergodicity_continuoustime}, in this section we turn our attention to the ergodicity of adaptive PDMC algorithms. The proofs of the two theorems are postponed to Section \ref{sec:proofs}.
First, let us consider the adaptive ZZS. We assume the following conditions on the potential.
\begin{assumption}[Growth Condition 3 in \cite{ZZerg}]\label{ass:ZZ_growth}
    $U \in \mathcal{C}^2(\mathbb{R}^d)$ and
    \begin{equation*}
        \lim_{\lVert x\rVert \to \infty} \frac{\max(1,\lVert \nabla^2 U(x) \rVert)}{\lVert \nabla U(x) \rVert} = 0, \qquad  \lim_{\lVert x \rVert \to \infty} \frac{\lVert\nabla U(x)\rVert}{ U(x)} = 0.
    \end{equation*}
\end{assumption}
\noindent Let us now state the ergodicity result for the adaptive ZZS.
\begin{theorem}\label{thm:ergADZZ}
Let $\mathcal{M}$ be a compact set of positive definite matrices and let $\Lambda$ be a set of excess switching rates $\gamma:E\to\mathbb{R}^d_+$ for which there are $0<\gamma_\textnormal{min} \leq \gamma_\textnormal{max} <\infty$ such that for all $\gamma \in \Lambda$
\begin{equation}\notag
    \gamma_\textnormal{min} \le \gamma(x,\theta) \le \gamma_\textnormal{max}  \qquad \text{for all } (x,\theta) \in E.
\end{equation}
Let $\mathcal{P} = \{P_{M,\gamma}: M \in \mathcal{M}, \gamma \in \Lambda \}$ be a family of preconditioned Zig-Zag processes with generators defined by Equation \eqref{eq:genAZZ}. Suppose Assumption~\ref{ass:ZZ_growth} holds and assume either of the following conditions holds:
\begin{enumerate}[label=\alph*)]
    \item $\mathcal{M} = \{M\in\mathbb{R}^{d \times d}: M_{ii} \in [a,b], M_{jk}=0 \quad \text{for all } j \neq k \}$ with $b>a>0$ ;
    \item $\mathcal{M}$ has no additional restrictions but adaptations are allowed only inside of a compact set $B$,
\end{enumerate}
and let $\Delta t$ be the discretisation step. Then the containment condition holds. Moreover, if the adaptive strategy is as described in Section~\ref{sec:algorithms} and is such that $p_n \to 0$ as $n\to\infty$, then the diminishing adaptation condition holds and thus for $\mu = \pi \times \textnormal{Unif}(\Theta)$:
\begin{equation}\label{eq:AZZisergodic}
    \lim_{n \to \infty} || \mathbb{P}( (X_n,\Theta_n) \in \cdot \,\,| x_0, \theta_0, \gamma_0) - \mu(\cdot) ||_{TV} = 0 \qquad \text{for all } (x_0,\theta_0) \in E, \gamma_0 \in \Lambda.
\end{equation}
Finally, for any bounded and measurable $f:\mathbb{R}^d \to \mathbb{R}$ a weak law of large numbers holds, i.e.
\begin{equation}\label{eq:AZZsatisfiesWLLN}
    \frac{\sum_{n=1}^N f(X_n)}{N} \to \pi(f) \qquad \text{in probability}.
\end{equation}
\end{theorem}
\begin{remark}
The time discretisation step $\Delta t$ can be chosen freely and is not subject to constraints. 
Moreover, we remark that under condition (a) on $\M$ the adaptive algorithm is SGE, i.e. satisfies Assumption~\ref{ass:SGE}, while under condition (b) it satisfies the first set of conditions in Theorem~\ref{thm:ergodicity_continuoustime}. Thus if one is interested in learning only the diagonal elements of the covariance, then it is possible to take $B=\mathbb{R}^d$ and allow adaptations independently of the state of the process.
\end{remark}
\begin{remark}
It was shown in \cite{ZZerg} that the ZZS is geometrically ergodic under Assumption \ref{ass:ZZ_growth} also in the case $\gamma = 0$, whereas in Theorem \ref{thm:ergADZZ} we require $\gamma(x,\theta)\geq \gamma_\textnormal{min} >0$. This extra assumption is convenient when proving a simultaneous small set condition (see Lemma \ref{lemma:smallZZ}). Based on similar arguments as in \cite{ZZerg}, we expect the statement of Theorem \ref{thm:ergADZZ} to remain valid even in the case $\gamma_{\textnormal{min}}=0$. In practice, one is free to choose $\gamma_{\textnormal{min}}$ very small and thus this assumption does not represent a severe limitation.
\end{remark}


Below we introduce a set of assumptions that is used to show ergodicity of the adaptive BPS. Here we limit our attention to the case of $\psi = \mathcal{N}(0,\mathbbm{1}_d)$. 
\begin{assumption}[Assumptions A1, A2, and A7 in \cite{BPS_Durmus}]\label{ass:BPS_growth}
    Let $U:\mathbb{R}^d \to [0,\infty)$ satisfy
    \begin{enumerate}[label=(\alph*)]
    \item $U \in \mathcal{C}^2(\mathbb{R}^d)$, and $x \to \lVert\nabla U(x)\rVert $ is integrable w.r.t. $\pi$;
    \item $\int_{\mathbb{R}^d} e^{-U(x)/2} dx < + \infty$ and $\lim_{\lVert x\rVert \to \infty} U(x) = +\infty$;
    \item There exists $\zeta \in (0,1)$ such that
    \begin{equation*}
        \liminf_{\lVert x\rVert \to \infty} \frac{\lVert\nabla U(x)\rVert}{U^{1-\zeta}(x)} > 0, \qquad \limsup_{\lVert x\rVert \to \infty} \frac{\lVert\nabla U(x)\rVert}{U^{1-\zeta}(x)} < \infty,
    \end{equation*}
    and 
    \begin{equation*}
        \limsup_{\lVert x\rVert \to \infty} \frac{\lVert\nabla^2 U(x)\rVert}{U^{1-2\zeta}(x)} < \infty .
    \end{equation*}
\end{enumerate}
\end{assumption}
\begin{theorem}\label{thm:ergodicityABPS}
Let $\mathcal{M}$ be a compact set of positive definite matrices and $\Lambda_r=[\lambda_\textnormal{min},\lambda_\textnormal{max}]$ for $0<\lambda_\textnormal{min} \leq \lambda_\textnormal{max} <\infty$. Let $\mathcal{P} = \{P_{M,\lambda_r}: M \in \mathcal{M},\lambda_r \in \Lambda_r \}$ be a family of preconditioned BPS's as defined in Equation \eqref{eq:genABPS}, where $\lambda_r$ is the refreshment rate. Suppose Assumption~\ref{ass:BPS_growth} holds and let $\Delta t$ be the discretisation step. Assume that  $\psi = \mathcal{N}(0,\mathbbm{1}_d)$. If adaptations are allowed only inside of a compact set as explained in Equation \eqref{eq:adap_compact}, then the containment condition holds. Furthermore, for the strategy discussed in Section~\ref{sec:algorithms} the diminishing adaptation holds as long as $p_n\to 0 $ as $n\to \infty$. Thus the ABPS is ergodic in the sense of Equation \eqref{eq:AZZisergodic} and satisfies a WLLN of the form \eqref{eq:AZZsatisfiesWLLN} for any bounded and measurable $f:\mathbb{R}^d \to \mathbb{R}$. 
\end{theorem}
\begin{remark}
Proving ergodicity of the adaptive BPS with refreshments from $\mathbb{S}^{d-1}$, i.e. the unit sphere centred at the origin, is more challenging due to the more involved drift condition proved in \cite{BPSexp_erg}. In particular, it is not straightforward to convert it into a simultaneous drift condition as required in assumption 2(b) of Theorem \ref{thm:ergodicity_continuoustime}.
\end{remark}

\section{Proofs of the main theorems}\label{sec:proofs}

\subsection{Proof of Theorem~\ref{thm:ergADZZ}}
In order to prove the theorem we show that, for suitable families of preconditioners, either Assumption~\ref{ass:SGE} holds for the family of discretised ZZ processes, or condition (1) in Theorem~\ref{thm:ergodicity_continuoustime} holds for the family of continuous time ZZ processes. In the next two sections we show auxiliary results, while in Section~\ref{sec:finalZZ} we assemble them to show the theorem.

\subsubsection{Minorisation condition for the ZZS}\label{sec:smallZZ}
The following lemma shows that a simultaneous small set condition holds for the family of ZZ processes. 
The strategy of the proof is to reduce the $d$-dimensional minorisation condition to $1$-dimensional conditions for every component of the process. Then we can take advantage of Lemma \ref{lemmma:smallsetZZ1D} in Appendix \ref{sec:app_smallset_1d}, which establishes that a simultaneous minorisation condition holds for a $1$-dimensional ZZ process as long as lower and upper bounds for the switching rates are available.

\begin{lemma}\label{lemma:smallZZ}
Let $U\in\mathcal{C}^1$. Consider the family of $d$-dimensional Zig-Zag processes with generators $\{\mathcal{L}_{M,\gamma}: M \in \mathcal{M}, \gamma\in\Lambda \}$, in which $\mathcal{M}$ is a compact set of positive definite matrices and $\Lambda$ is a set of switching rates $\gamma:E\to\mathbb{R}^d_+$. Assume that there are $\gamma_{\textnormal{min}}, \gamma_{\textnormal{max}}$ such that for all $\gamma\in\Lambda$
\begin{equation}\notag
    0 < \gamma_{\textnormal{min}} \le \gamma(x,\theta) \le \gamma_{\textnormal{max}} < \infty \qquad \text{for all } (x,\theta)\in E.
\end{equation}
Then for any set of the form $C=D \times V $, where $D \subset \mathbb{R}^d$ is a compact set and $V \subseteq \Theta$, there exists $t_0>0$ such that for any $t \ge t_0$ there are $\delta >0$, and probability measures $\{\nu_M\}_{M\in\M}$ on $E$ such that 
\begin{equation}\notag
    P^{t}_{M,\gamma}((x,\theta),\cdot) \ge \delta \nu_M(\cdot) \qquad \text{for all } (x,\theta) \in C .
\end{equation}
In particular, $t_0$ and $\delta$ do not depend neither on $M$ nor on $\gamma$.
\end{lemma}
\begin{proof}
We first show that all compact sets are uniformly small for the family of standard Zig-Zag processes with targets in $\{\tilde{\pi}_M(\xi)=\pi(M\xi): M \in \mathcal{M} \}$. Let $C$ be a $d$-dimensional rectangle of the form $[-R,+R]^d \times V$ for some $R>0$ and $V \subseteq \Theta$.

Let $M \in \mathcal{M}$ and $\gamma \in \Lambda$ and denote as $(\Xi_{M,\gamma} (t), \Theta_{M,\gamma} (t))_{t \ge 0}$ the ZZ process that targets $\tilde{\pi}_M(\xi)$ and has switching rate $\tilde{\gamma}_M(\xi,\theta) = \gamma(M\xi,\theta)$. Let $\tilde{P}_{M,\gamma}$ be the corresponding semigroup. Then observe that for any rectangle $B= B_1 \times \dots \times B_d$, with $B_i = [R_{i,1},R_{i,2}]$ for some $R_{i,2} > R_{i,1}$ the following holds
\begin{equation}\label{eq:prodsmallset}
    \begin{aligned}
        \tilde{P}_{M,\gamma}^t((\xi,\theta),B) &= \mathbb{P}_{(\xi,\theta)}((\Xi_{M,\gamma}(t),\Theta_{M,\gamma}(t)) \in B) \\
        & = \prod_{i=1}^d \mathbb{P}_{(\xi,\theta)} \big( (\Xi_{M,\gamma}(t),\Theta_{M,\gamma}(t))_i \in B_i \,|\,  (\Xi_{M,\gamma}(t),\Theta_{M,\gamma}(t))_j \in B_j \, \text{for } j>i \big).
    \end{aligned}
\end{equation}
Then observe that, independently of the values at any time $s<t$ of the other components of the process, the $i$-th switching rate $\tilde{\lambda}_{M,i}(\xi,\theta) = (\theta_i \partial_i U(M \xi))_+ + \gamma_i (M \xi, \theta)$ satisfies the bounds 
\begin{equation}\label{eq:bds_lambdasZZ}
    0 < \gamma_{\text{min}} \le \tilde{\lambda}_{M,i}(\xi,\theta) \le \lambda_{\text{max}} < \infty,
\end{equation}
where we have defined
\begin{equation}\notag
    \lambda_{\text{max}} \coloneqq \max_{i=1,\dots,d} \,\,\max_{M \in \mathcal{M}} \,\,\max_{ \{ \xi \in \tilde{C}_t, \theta \in \Theta \}} \left\{\left( \theta_i \partial_i U(M \xi) \right)_+ \right\} + \gamma_{\textnormal{max}} 
\end{equation}
with $\tilde{C}_t = [-R-t,R+t]^d$ is the set of reachable points in time $t$. Here $\lambda_{\text{max}}$ is well defined because $U \in \mathcal{C}^1$. In particular neither of the bounds in \eqref{eq:bds_lambdasZZ} depend on the specific $\tilde{\pi}_M$ and thus hold for any $M \in \mathcal{M}$. Therefore, we can apply Lemma~\ref{lemmma:smallsetZZ1D} 
with $t_0 = 2 R + \varepsilon$ with $\varepsilon > 0$ to each 
component of the product \eqref{eq:prodsmallset}. It then follows that for any $t \ge t_0$ there exists $\delta_1>0$ that satisfies
\begin{equation}\label{eq:smallZZrecs}
    \tilde{P}_{M,\gamma}^t((\xi,\theta),B) \ge \prod_{i=1}^d \delta_1 \nu_1(B_i) = \delta_1^d \, \nu_d (B) = \delta_d\, \nu_d (B) \qquad \text{for all } (\xi,\theta) \in C,
\end{equation}
where $\nu_d = \mathrm{Leb}_d(C) \times \text{Unif}(\Theta)$ is the $d$-dimensional equivalent of $\nu_1$ that was defined in Lemma~\ref{lemmma:smallsetZZ1D}. Most importantly, $\delta$ depends only on the bounds on the switching rates, which are uniform for all $\gamma \in \Lambda$. Therefore Equation \eqref{eq:smallZZrecs} holds for the same $\delta$ and $\nu_d$ for all $M \in \mathcal{M},\, \gamma\in\Lambda$, and any rectangle $B \subset C$.
Since the set of rectangles is a $\pi$-system and generates the Borel $\sigma$-algebra, by the monotone class theorem (Theorem 6.2 in \cite{jacodprotter}) it follows that the lower bound of Equation \eqref{eq:smallZZrecs} holds for any Borel set $B$. 
Since any compact set can be included in a large enough hypercube, for all sets $C = D \times V$, with $D$ compact and $ V \subseteq \Theta$, there are $\nu_d$ and $t_0>0$ such that for all $t\geq t_0$ there exists $\delta_d>0$  for which $C$ is uniformly $(t,\delta_d,\nu_d)$-small for the family of standard ZZ processes defined above.

Now we wish to translate this result to the family of linearly transformed Zig-Zag processes with excess switching rates in $\Lambda$ and target $\pi$. Denote as $(X_{M,\gamma}(t),\Theta_{M,\gamma}(t))_{t\ge 0}$ the process with generator $\mathcal{L}_{M,\gamma}$ and observe that for any $t>0$ and any set $A = A_x \times V$, for a Borel set $A_x$, the event $\{(X_{M,\gamma}(t),\Theta_{M,\gamma}(t)) \in A\}$ is equivalent to the event $\{(\Xi_{M,\gamma}(t),\Theta_{M,\gamma}(t)) \in \tilde{A}_M \}$ for $\tilde{A}_M = \{(\xi,\theta) \in E: M \xi \in A_x, \, \theta \in V \}$. 
Assume $C \in E$ is a compact set and let $(x,\theta) \in C$. Then define $C_\mathcal{M} = \{(\xi,\theta) = (M^{-1} x ,\theta): (x,\theta) \in C, \, M \in \mathcal{M} \}$ which is itself a compact set and depends on $\mathcal{M}$, but not on the specific $M\in\mathcal{M}$.
Then, for all $t \ge t_0$, with $t_0$ large enough, it holds that
\begin{equation}\notag
    \begin{aligned}
        P_{M,\gamma}^t((x,\theta),A) &= \mathbb{P}_{(x,\theta)}((X_{M,\gamma}(t),\Theta_{M,\gamma}(t)) \in A) \\
        & = \mathbb{P}_{(\xi,\theta)}((\Xi_{M,\gamma}(t),\Theta_{M,\gamma}(t)) \in \tilde{A}_M) \\
        & \ge \delta \nu_d (\tilde{A}_M)  \hspace{86pt} \text{for all } (\xi,\theta) \in C_\mathcal{M} \\
        & = \delta \nu_M(A)  \hspace{92pt} \text{for all } (x,\theta) \in C
    \end{aligned}
\end{equation}
with $\nu_M(A) = \nu_d(\tilde{A}_M)$ and $\delta$ is chosen such that $C_\mathcal{M}$ is a uniform $(t,\delta,\nu)$-small set for the family of standard ZZ processes with targets $\tilde{\pi}_M$ and switching rates $\Tilde{\gamma}_M$. Therefore there is no dependence neither on $M$ nor on $\gamma$ in $t_0$ and $\delta$ and the proof is concluded.
\end{proof}

\subsubsection{Drift conditions for the Zig-Zag process}\label{sec:driftZZ}
If we restrict our attention to the class of diagonal matrices with positive, bounded entries, then the Lyapunov function in \cite[Lemma 11]{ZZerg} satisfies also a simultaneous drift condition. This is shown in the following lemma.
\begin{lemma}\label{lemma:SimultaneousDriftZZ}
Let Assumption~\ref{ass:ZZ_growth} hold. Consider the family of linearly transformed Zig-Zag processes with generators $\{ \mathcal{L}_{M,\gamma}: M \in \mathcal{M}, \gamma \in \Lambda \}$, where
\begin{equation}\label{eq:diagonal_preconditioners}
    \mathcal{M} = \{M\in\mathbb{R}^{d \times d}: M_{ii} \in [V_{\textnormal{min}}^i,V_{\textnormal{max}}^i], \,\,M_{jk}=0 \,\, \text{for all } j \neq k \} ,
\end{equation}
with $V_{\textnormal{max}}\ge V_{\textnormal{max}}^i \geq V_{\textnormal{min}}^i \ge V_{\textnormal{min}}>0$ for each $i = 1,\dots,d$, and where $\Lambda$ is a set of excess switching rates $\gamma:E\to\mathbb{R}^d_+$ such that for all $\gamma\in\Lambda$ it holds that
\begin{equation}\label{eq:bound_excess_switch}
    \gamma_i( x, \theta) \le \gamma_\textnormal{max} \qquad \text{for all } (x,\theta) \in E,\, i=1,\dots,d.
\end{equation}
Let $\delta>0$ and $\alpha >0 $ be such that $0 < (\delta \gamma_\textnormal{max})/V_{\textnormal{min}} < \alpha < 1$ and define $\phi(s) = \tfrac{1}{2} \textnormal{sign}(s) \ln{(1+\delta |s|)}$.
Then the function 
\begin{equation}\label{eq:lyapfnZZ}
    V(x,\theta) = \exp{\left( \alpha U(x) + \sum_{i=1}^d \phi(\theta_i \partial_i U(x)) \right) }
\end{equation}
is a simultaneous Lyapunov function for the family of ZZ processes, that is there exist $A_1>0$, $A_2>0$, a compact set $C \subset E$ such that
\begin{equation}\notag
    \mathcal{L}_{M,\gamma} V(x,\theta) \le - A_1 V(x,\theta) + A_2 \mathbbm{1}_C (x,\theta) \qquad \text{for all } (x,\theta) \in E,\, M \in \mathcal{M},\,\gamma\in\Lambda,
\end{equation}
where $A_1$, $A_2$, $C$ do not depend neither on $M$ nor on $\gamma$ (but depend on $\mathcal{M}$ and $\Lambda$).
\end{lemma}

\begin{proof}
We follow the proof of Lemma 11 in \cite{ZZerg}. Let $M\in\M$ and $\gamma \in \Lambda$. Applying the generator to $V$, which was defined in Equation \eqref{eq:lyapfnZZ}, we obtain
\begin{align}
    \notag \left( \frac{\mathcal{L}_{M,\gamma} V}{V} \right) (x,\theta) &= \sum_{i=1}^d M_{ii} \theta_i \left( \alpha \partial_i U(x) + \sum_{j=1}^d \theta_j \partial_{ij} U(x) \phi'(\theta_j \partial_j U(x)) \right) \\
    \notag & \,\, + \sum_{j=1}^d \left( M_{ii} \left(\theta_i \partial_i U(x)\right)_+ + \gamma_i(x,\theta) \right) \cdot \\
    \notag & \qquad \cdot \left( \exp{\left( \phi(-\theta_i \partial_i U(x)) - \phi(\theta_i \partial_i U(x)) \right) } - 1 \right) 
\end{align}
Now consider $s = \theta_i \partial_i U(x) \ge 0$, then 
\begin{align}
    \notag M_{ii} \alpha s + (M_{ii} s + \gamma_i)(\exp{(\phi(-s) - \phi(s))} -1 ) &=
    \notag M_{ii} \left( \alpha s + \left(s+ \frac{\gamma_i}{M_{ii}} \right) \left( \frac{1}{1+\delta s} - 1 \right) \right) \\
    \notag & = M_{ii} \left( (\alpha-1) s + \frac{s - \gamma_i/M_{ii}}{1+\delta s} \right) \\
    \notag & \le -M_{ii} ((1-\alpha)|s| + 1/\delta).
\end{align}
In case $s = \theta_i \partial_i U(x) < 0$ we obtain
\begin{align}
    \notag M_{ii} \alpha s + (M_{ii} s + \gamma_i)(\exp{(\phi(-s) - \phi(s))} -1 ) &=
    \notag M_{ii} \left( \alpha s + \frac{\gamma_i}{M_{ii}} (1+\delta |s| -1 ) \right) \\
    \notag & \le -M_{ii} \left( \alpha - \frac{\gamma_\textnormal{max}}{M_{ii}} \delta  \right) |s| .
\end{align}
Note that by assumption $( \alpha - \gamma_\textnormal{max} \delta/M_{ii}  ) \ge ( \alpha - \gamma_\textnormal{max}\delta/V_{\text{min}} ) > 0$.
For the remaining term the best we can do is to derive the following bound 
\begin{align}
     \notag \sum_{i,j} M_{ii} \theta_i \theta_j \partial_{ij} U(x) \phi'(\theta_j \partial_j U(x)) & \le \sum_{i,j} M_{ii}  \phi'(\theta_j \partial_j U(x))  | \partial_{ij} U(x) | \\
     \notag & \le V_{\text{max}} \frac{\delta}{2}  \sum_{i,j} | \partial_{ij} U(x) |,
\end{align}
where we have used that $0 \le \phi'(s) \le \delta / 2$.
Finally we obtain for any $M \in \mathcal{M}$
\begin{align}\label{eq:final_ineq_drift_ZZ}
    \left( \frac{\mathcal{L}_{M,\gamma} V}{V} \right) (x,\theta) & \le -\min\left(1-\alpha, \alpha - \frac{\gamma_\textnormal{max} \delta}{V_{\text{min}}} \right) V_{\text{min}} \sum_{i=1}^d |\partial_i U(x)|  \\
    \notag& \quad + \frac{V_{\text{max}} d }{\delta} + \frac{\delta}{2} V_{\text{max}} \sum_{i,j} \lvert \partial_{ij} U(x) \rvert,
\end{align}
which is independent of the specific $M$ and $\gamma$ and can be made arbitrarily small outside of a sufficiently large compact set $C$ by our assumptions on $U$.
\end{proof}

If we wish to consider a more general class of positive-definite matrices, then we have to settle for the following, weaker result.
\begin{lemma}\label{lemma:lyapZZ2}
Let Assumption~\ref{ass:ZZ_growth} hold. Consider a family of linearly transformed Zig-Zag processes with generators $\{ \mathcal{L}_{M,\gamma}: M \in \mathcal{M}, \gamma\in\Lambda \}$, where $\mathcal{M} \subset \mathbb{R}^{d \times d}$ is a compact space of positive definite matrices and $\Lambda$ is a space of excess switching rates $\gamma:E\to\mathbb{R}_+$ such that \eqref{eq:bound_excess_switch} is satisfied for some $\gamma_{\textnormal{max}}$.
Let $\delta>0$ and $\alpha >0 $ be such that $0 < \delta \gamma_\textnormal{max} < \alpha < 1$. Define for each $M \in \mathcal{M}$ the function
\begin{equation}\label{eq:LyapZZ_M}
    V_M(x,\theta) = \exp{\left( \alpha U(x) + \sum_{i=1}^d \phi(\theta_i \langle M_i,\nabla U(x) \rangle) \right) },
\end{equation}
where $M_i$ denotes the $i$-th column of $M$ and $\phi:\mathbb{R}\to\mathbb{R}$ was defined in Lemma~\ref{lemma:SimultaneousDriftZZ}.
Then there are $A_1>0$, $A_2>0$, and a compact set $C \subset E$ such that for all $M\in\mathcal{M}$ the following simultaneous drift condition holds:
\begin{equation}\notag
    \mathcal{L}_{M,\gamma} V_M(x,\theta) \le - A_1 V_M(x,\theta) + A_2 \mathbbm{1}_C (x,\theta) \qquad \text{for all } (x,\theta) \in E.
\end{equation}
In particular $A_1$, $A_2$, $C$ do not depend neither on $M$ nor on $\gamma$ (but depend on $\mathcal{M}$ and $\Lambda$).
\end{lemma}
\begin{proof}
Consider the change of variables $\xi = M^{-1}x$. Denote as $\tilde{\mathcal{L}}_{M,\gamma}$ the generator of a ZZ process with transformed stationary measure $\tilde{\pi}_M$ and transformed excess switching rate $\tilde{\gamma}_M(\xi,\theta) = \gamma(M\xi,\theta)$. Then transforming the function in \eqref{eq:LyapZZ_M} we obtain a Lyapunov function for the standard ZZ process with transformed target:
\begin{equation}\notag
    \tilde{V}_M(\xi,\theta) = \exp{\left( \alpha \tilde{U}_M(\xi) + \sum_{i=1}^d \phi(\theta_i \partial_i \tilde{U}_M(\xi) ) \right) },
\end{equation}
where $\tilde{U}_M (\xi) = U(M \xi)$.
Since $\mathcal{M}$ is a compact space of positive definite linear transformations, Assumption~\ref{ass:ZZ_growth} is satisfied by each $\tilde{U}_M (\cdot)$.
Then, by the proof of Lemma 11 in \cite{ZZerg}, for any constant $A_1>0$ there exists a large enough ball $\tilde{B}_M \coloneqq \{ (\xi,\theta) \in E: \theta \in \Theta, \xi \in B(0,\tilde{R}_M) \}  $ such that
\begin{equation}\notag
    \tilde{\mathcal{L}}_{M,\gamma}  \Tilde{V}_M(\xi,\theta) \le -A_1 \tilde{V}_M(\xi,\theta) \qquad \text{for all } (\xi,\theta) \notin \tilde{B}_M.
\end{equation}
In particular $\tilde{B}_M$ does not depend on $\gamma$, but only on $\gamma_\text{max}$. Observe that by the proof of Lemma 11 in \cite{ZZerg} it follows that $\tilde{R}_M$ depends continuously on $M$. Indeed one can show that $(\tilde{\mathcal{L}}_{M,\gamma}  \Tilde{V}_M(\xi,\theta))/ \tilde{V}_M(\xi,\theta)$ is smaller or equal than a sum of terms which depend continuously on the components of $\nabla \Tilde{U}_M$ and $\nabla^2 \Tilde{U}_M$. Continuity follows from the assumption that $U\in \mathcal{C}^2$ and because $M$ does not appear in other ways.
Proposition~\ref{prop:generatorZZ} then implies that $\mathcal{L}_{M,\gamma} V_M(x,\theta) = \tilde{\mathcal{L}}_{M,\gamma} \tilde{V}_M(\xi,\theta)$. Thus for each $M \in \mathcal{M}$
\begin{equation}\notag
    \mathcal{L}_{M,\gamma} V_M(x,\theta) \le -A_1 V_M(x,\theta) \qquad \text{for all } (x,\theta) \notin B_M,
\end{equation}
where $B_M \coloneqq \{ (x,\theta) \in E: (M^{-1}x,\theta) \in \tilde{B}_M  \}$ and $A_1$ does not depend on $M$. 
Finally, take $C$ to be any ball that contains all sets $B_M$. 
Then $C$ is bounded by continuity of $\Tilde{R}_M$ in $M$ and compactness of $\M$. By continuity of $\mathcal{L}_{M,\gamma} V_M(x,\theta)$ in $x,\theta,M$ it follows that 
\begin{equation}\notag
    \mathcal{L}_{M,\gamma}  V_M(x,\theta) \le -A_1 V_M(x,\theta) + A_2 \mathbbm{1}_{C} (x,\theta),
\end{equation}
in which $A_2 = \max_{ \{ M \in \mathcal{M},\,\gamma\in\Lambda,\, (x,\theta) \in C \}} \left( \mathcal{L}_{M,\gamma}  V_M(x,\theta) + A_1 V_M(x,\theta) \right)$. In particular the maximum over $\gamma \in \Lambda$ does not cause problem as the $\gamma$'s are uniformly bounded. Hence $A_2$ is independent of the specific $M$ and $\gamma$.
\end{proof}

\subsubsection{Finalising the proof of Theorem~\ref{thm:ergADZZ}}\label{sec:finalZZ}
Let us first consider the case of diagonal preconditioners. Let $\Delta t >0$ be the discretisation step. Then by Lemma~\ref{lemma:smallZZ} for any set of the form $C=D \times V$, with $D$ compact and $V \subseteq \Theta$, there exist $\delta >0$, $n_0 := \inf \{n\in\mathbb{N}: n\geq  \frac{t_0}{\Delta t}\}$, $\nu_M(\cdot)$ such that $C$ is a uniform $(\nu_M, n_0, \delta)$-small set for the family of discretised processes. Observe that no conditions on $\Delta t$ are required. Moreover, a simultaneous drift condition holds by Lemma~\ref{lemma:SimultaneousDriftZZ} combined with Lemma~\ref{lemma:drift2drift} for any $\Delta t$. The condition $\mu (V) < \infty $ is satisfied by definition of the Lyapunov function $V$, and in fact it also holds that $\sup_{(x,\theta) \in C}V(x,\theta)<\infty$ by continuity of $V$ in $x$. Therefore all the conditions of Assumption~\ref{ass:SGE} are verified, which means the family is simultaneously geometrically ergodic and by Theorem 3 in \cite{BaiRob} the containment condition is satisfied.

In the case of a non-diagonal transformation matrix, parts (a)-(c) of condition (1) in Theorem~\ref{thm:ergodicity_continuoustime} are verified by Lemmas~\ref{lemma:smallZZ} and~\ref{lemma:lyapZZ2}. It also holds that $\sup_{\{(x,\theta)\in C, M\in\M\}} V_M(x,\theta) < \infty$ for (small) sets $C= D\times V$, with $D$ compact and $V \subseteq \Theta$, because of continuity of each $V_M$ in $x$ and $M$, together with the fact that $D$ and $\M$ are compact spaces (see Lemma~\ref{lemma:lyapZZ2} for the definition of $\{V_M \}_{M\in\M}$). Moreover $\mu (V_M) = \tilde{\mu}_M (\tilde{V}_M) < \infty$, where $\tilde{\mu}_M = \tilde{\pi}_M \otimes \text{Unif}(\Theta)$ and $\Tilde{V}_M$ is a Lyapunov function for a standard ZZ process with invariant measure $\tilde{\mu}_M$.
Theorem~\ref{thm:ergodicity_continuoustime} ensures that the containment condition holds true with no restriction on $\Delta t$.

Proposition~\ref{prop:dimi_adapt_zzs_bps} implies that, under the assumption that $p_n\to 0$ as $n\to \infty$, the adaptive strategy satisfies diminishing adaptation. Therefore ergodicity follows from diminishing adaptation and containment.

\subsection{Proof of Theorem~\ref{thm:ergodicityABPS}}
In the next two sections we show respectively that condition (2) in Theorem \ref{thm:ergodicity_continuoustime} is verified for the family of preconditioned BPS and/or of BPS with refreshment rates in a compact set. In Section~\ref{sec:finalBPS} we use these auxiliary results to show the theorem.

\subsubsection{Simultaneous coupling inequality}
The next lemma states that a simultaneous coupling inequality is satisfied for the BPS with adaptive preconditioner and/or adaptive refreshment rate. The proof is based on the proof of Lemma 12 in \cite{BPS_Durmus}, which shows a coupling inequality result for the standard BPS.
\begin{lemma}\label{lemma:couplABP}
Let condition (a) in Assumption~\ref{ass:BPS_growth} hold for the energy function $U$. Consider the family of BP processes $\{P_{M,\lambda_r}^t: M \in \mathcal{M},\, \lambda_r\in\Lambda_r\}$, where $\mathcal{M}$ is a compact space of non-singular preconditioning matrices and $\Lambda_r = [\lambda_r^{\textnormal{min}},\lambda_r^{\textnormal{max}}]$ is the set of refreshment rates, for some $0<\lambda_r^{\textnormal{min}}\leq \lambda_r^{\textnormal{max}}<\infty$. Then for any compact set $K \subset \{(x,\theta)\in\mathbb{R}^d \times \mathbb{R}^d: \lVert x\rVert+\lVert\theta\rVert \leq R  \}$, with $R\geq 0$, there exists $\alpha >0$ such that for all $(x,\theta)$, $(\tilde{x}, \tilde{\theta}) \in K$, for all $t>0$, and for all $M \in \mathcal{M}$ and $\lambda_r \in \Lambda_r$
\begin{equation}\notag
    \lVert P^t_{M,\lambda_r}((x,\theta),\cdot) - P^t_{M,\lambda_r}((\tilde{x}, \tilde{\theta}), \cdot) \rVert_{\textnormal{TV}} \le  2\left(1-\alpha \right). 
\end{equation}
In particular $\alpha$ is independent of $M$ and $\lambda_r$.
\end{lemma}
\begin{proof}
Let $M \in \mathcal{M}$ and $\lambda_r \in \Lambda_r$. First note that for $(\xi,\theta) = (M^{-1} x,\theta)$ and $(\tilde{\xi}, \tilde{\theta}) = (M^{-1} \tilde{x}, \tilde{\theta})$ we have the equality
\begin{equation}\label{eq:couplBPS_to_sdt}
    \lVert P^t_{M,\lambda_r}((x,\theta),\cdot) - P^t_{M,\lambda_r}((\tilde{x}, \tilde{\theta}), \cdot) \rVert_{\textnormal{TV}} = \lVert \tilde{P}^t_{M,\lambda_r}((\xi,\theta),\cdot) - \tilde{P}^t_{M,\lambda_r}((\tilde{\xi}, \tilde{\theta}), \cdot) \rVert_{\textnormal{TV}} ,
\end{equation}
where $\tilde{P}^t_{M,\lambda_r}$ is the transition kernel of a standard BPS with energy function $\tilde{U}_M(\xi) = U(M \xi)$ as defined in Figure~\ref{fig:trans_scheme}, and refreshment rate $\lambda_r$.

Our strategy is thus to apply Lemma 12 in \cite{BPS_Durmus} to each semigroup $\tilde{P}^t_{M,\lambda_r}$ and then take advantage of Equation \eqref{eq:couplBPS_to_sdt}. Observe that it is possible to apply the lemma because part (a) of Assumption~\ref{ass:BPS_growth} implies that for all $M\in\M$
\begin{equation*}
    \int  \lVert \nabla \tilde{U}_M (\xi) \rVert \tilde{\pi}_M (d\xi) = \int  \lVert \nabla U(x) \rVert \pi (dx) < \infty,
\end{equation*}
for $\tilde{\pi}_M(\xi)= \exp(- \tilde{U}_M(\xi))/\tilde{Z}_M$, with $\tilde{U}_M(\xi)=U(M\xi)$ and $\tilde{Z}_M=Z/\lvert\text{det}(M)\rvert$. Hence the integrability condition holds for each $\tilde{P}^t_{M,\lambda_r}$ with respect to the corresponding target $\tilde{\pi}_M$. Applying the lemma and Equation \eqref{eq:couplBPS_to_sdt} it follows that for any compact set $K \subset \{(x,\theta)\in\mathbb{R}^d \times \mathbb{R}^d: \lVert x\rVert+\lVert\theta\rVert \leq R  \}$, with $R\geq 0$, for each $M\in\M$ and $\lambda_r\in\Lambda_r$ there exists $\alpha = \alpha(\lambda_r,M,K)>0$ such that for all $(x,\theta)$, $(\tilde{x}, \tilde{\theta}) \in K$, for all $t>0$
\begin{equation}\label{eq:aux_couplineq_BPS}
    \lVert P^t_{M,\lambda_r}((x,\theta),\cdot) - P^t_{M,\lambda_r}((\tilde{x}, \tilde{\theta}), \cdot) \rVert_{\textnormal{TV}} \le  2\left(1-\alpha(\lambda_r,M,K) \right). 
\end{equation}
Observe that there is a dependence on $M$ due to the fact that the set of initial conditions $K$ is transformed to $K_M \coloneqq \{(\xi,\theta)\in\mathbb{R}^d\times\mathbb{R}^d: \lVert M^{-1}\xi\rVert + \lVert \theta\rVert\leq R\}$ in Equation \eqref{eq:couplBPS_to_sdt}. This translates to a dependence on $M$ in the coefficient $\alpha$, as a consequence of the dependence of $\alpha$ on set $K_M$. This inconvenience can be avoided if we choose $\alpha$ such that the coupling inequality is satisfied for all initial conditions $(\xi,\theta)\in K_\mathcal{M}$, where $R_{\mathcal M}$ is such that $K_\M \coloneqq \{(\xi,\theta)\in\mathbb{R}^d\times\mathbb{R}^d: \lVert \xi\rVert + \lVert \theta\rVert\leq R_\M \}$ satisfies $K_M \subset K_\M$ for all $M\in\M$. Thus for all $M\in\M$ it holds that if $(x,\theta)\in K$, then $(\xi,\theta)=(M^{-1}x,\theta) \in K_\M$.

The last thing to do is showing that there exists a constant $\alpha^*>0$ independent of $M$ and $\lambda_r$ such that $\alpha(\lambda_r,M,K) \geq \alpha^*$ for all $M\in\M$ and $\lambda_r\in\Lambda_r$.
This is indeed the case for the following reasons:
\begin{itemize}
    \item The dependence on $M$ appears then in the following term, which in the statement of Lemma 12 in \cite{BPS_Durmus} appears as a factor in $\alpha(\lambda_r,M,K)$:
    \begin{equation*}
        g(r) = \mathbb{P}\left( E_3 \leq r \Tilde{N} \sup_{\{\xi:\lVert\xi\rVert\leq  (1+E_1/\lambda_r)R_\M + (r/\lambda_r)\Tilde{N}\}} \lVert \nabla \tilde{U}_M (\xi)\rVert  \right).
    \end{equation*}
    Here $\Tilde{N}= N+(1+E_1/\lambda_r)R_\M$ for some $N>0$, and $E_1,E_3$ are independent exponential random variables with parameter $1$. Because $g(r)$ is a factor in $\alpha(\lambda_r,M,K)$, we wish to bound it from below, and hence we should bound the supremum of $\lVert \nabla \tilde{U}_M (\xi)\rVert$ from below. Denoting $\zeta(\lambda_r) = (1+E_1/\lambda_r)R_\M + (r/\lambda_r)\Tilde{N}$, this can be done as follows:
    \begin{align*}
        \sup_{\{\xi:\lVert\xi\rVert\leq  \zeta(\lambda_r)\}} \lVert \nabla \tilde{U}_M (\xi)\rVert
        & = \sup_{\{x:\lVert M^{-1} x\rVert \leq  \zeta(\lambda_r)\}} \lVert M^T \nabla  U (x)\rVert \\
        & \geq \left( \min_{M\in\M} \frac{1}{\lVert M^{-T}\rVert} \right) 
        \sup_{\{x:\lVert x\rVert \leq \min_{M\in\M}(\lVert M\rVert) \zeta(\lambda_r)\}} \lVert \nabla  U (x)\rVert,
    \end{align*}
    where we used that $\tilde{U}_M (\xi)=M^T \nabla  U (x)$, that $\lVert M^{-1}x\rVert \geq \lVert x\rVert/\lVert M\rVert$, and the compactness of $\M$. This is sufficient to eliminate any dependence on $M$ in $\alpha$. 
    
    \item Depending on the specific factors in the statement of the lemma, the switching rate $\lambda_r$ can be conveniently bounded either above by $\lambda_r^{\text{max}}$ or below by $\lambda_r^{\text{min}}$ in order to bound $\alpha(\lambda_r,M,K)$ from below. This can be done in every term in which $\lambda_r$ appears and it follows that the dependence on it can be easily eliminated.
\end{itemize}
We have thus shown that we can choose $\alpha^* > 0$ such that in the same setting of \eqref{eq:aux_couplineq_BPS}, and for all $M\in\M$ and all $\lambda_r\in\Lambda_r$
\begin{equation}\notag
    \lVert P^t_{M,\lambda_r}((x,\theta),\cdot) - P^t_{M,\lambda_r}((\tilde{x}, \tilde{\theta}), \cdot) \rVert_{\textnormal{TV}} \le  2\left(1-\alpha^* \right). 
\end{equation}
\end{proof}

\subsubsection{Drift condition for the BPS}
The second condition we need is uniformity of the constants in the drift condition for the family of preconditioned BPS and/or for the family of BPS with different refreshment rate. To this end, we go through the proof of Lemma 7 from \cite{BPS_Durmus} to show that this is indeed the case.
\begin{lemma}\label{lemma:driftABP}
Consider a family of BP processes with generators $\{\mathcal{L}_{M,\lambda_r}: M \in \mathcal{M}, \, \lambda_r \in \Lambda_r\}$, where $\mathcal{M}$ is a compact space of non-singular matrices that act as preconditioners, $\lambda_r$ is the refreshment rate and $\Lambda_r=[\lambda_{\text{min}},\lambda_{\text{max}}]$ for some $0<\lambda_{\text{min}}\leq\lambda_{\text{max}}<\infty$. Let Assumption~\ref{ass:BPS_growth} hold and let $\psi = \mathcal{N}(0,\mathbbm{1}_d)$. Then there are $A_1,A_2>0$ and a class of functions $\{V_{M,\lambda_r}: M \in \mathcal{M}, \, \lambda_r \in \Lambda_r \}$ such that for each $M \in \mathcal{M}$ and $\lambda_r \in \Lambda_r$ it holds that
\begin{equation*}
    \mathcal{L}_{M,\lambda_r} V_{M,\lambda_r}(x,\theta) \leq -A_1 V_{M,\lambda_r}(x,\theta) + A_2 \qquad \text{for all } (x,\theta) \in \mathbb{R}^d \times \mathbb{R}^d,
\end{equation*}
where in particular $A_1,A_2$ do not depend on $M$.
\end{lemma}

\begin{proof}
We follow the same underlying idea that was used in the proof of Lemma~\ref{lemma:lyapZZ2}. It is in fact sufficient that there exist $A_1,A_2>0$, both independent of $M$ and $\lambda_r$, such that for all $M$ and $\lambda_r$ there exists a function  $\Tilde{V}_{M,\lambda_r}(\xi,\theta)$ such that 
\begin{equation}\label{eq:driftBPstd}
    \tilde{\mathcal{L}}_{M,\lambda_r} \tilde{V}_{M,\lambda_r}(\xi,\theta) \leq -A_1 \tilde{V}_{M,\lambda_r}(\xi,\theta) + A_2 \qquad \qquad \text{for all } (\xi,\theta) \in \mathbb{R}^d \times \mathbb{R}^d,
\end{equation}
where $\tilde{\mathcal{L}}_{M,\lambda_r}$ is the generator of a standard BPS with target $\Tilde{U}_M(\xi) = U(M \xi)$ and refreshment rate $\lambda_r$. Indeed, we can then define $V_{M,\lambda_r}(x,\theta)= \tilde{V}_{M,\lambda_r}(M^{-1}x,\theta)$ to obtain by Proposition~\ref{prop:generatorBPS}
\begin{equation}\label{eq:finalBPdrift}
    \mathcal{L}_{M,\lambda_r} V_{M,\lambda_r}(x,\theta) = \tilde{\mathcal{L}}_{M,\lambda_r} \tilde{V}_{M,\lambda_r}(\xi,\theta) \leq -A_1 V_{M,\lambda_r}(x,\theta) + A_2 
\end{equation}
for all $(x,\theta) \in \mathbb{R}^d \times \mathbb{R}^d$.
In order to show the inequality in \eqref{eq:driftBPstd} we rely on \cite[Lemma 7]{BPS_Durmus}. In particular we begin by showing that the constants $c_1,c_2,c_3,c_4>0$, $R>0$ defined in \cite[Assumption A8]{BPS_Durmus} can be chosen independently of the specific $M$ for the class of standard BP samplers with targets $\Tilde{\pi}_M$. This implies that $A_1,A_2$ can then be chosen to be the same for every $M\in\M$. The dependence on $\lambda_r$ is dealt with as a second step. For ease of notation from now on we denote the corresponding energy function as $U_M (\xi)$. Let us restrict our attention to the case $\psi=\mathcal{N}(0,\mathbbm{1}_d)$ and thus, using the notation in \cite{BPS_Durmus}, we take $\ell(\xi)=1$, $H(\lVert \theta \rVert)= \eta \lVert \theta \rVert^2$ for some $\eta \in (0,1)$ small enough such that $\int_{\mathbb{R}^d} \exp{(\eta \lVert \theta\rVert^2)} \psi(d\theta) < \infty$, and $\Bar{U}_M(\xi) = U_M^{\zeta}(\xi)$ for $\zeta \in (0,1)$ chosen as in part (c) of Assumption~\ref{ass:BPS_growth}. Observe that $\eta, \zeta$ are uniform in $M$.

We start by considering $c_1$, which must be such that $\lVert\nabla_\xi \Bar{U}_M(\xi)\rVert \geq c_1$ for all points outside of a ball of radius $R_1$. In particular we wish to show that there exist $c_1,R_1$ that satisfy the property above for all $M\in \M$. Observing that for $x=M\xi$ it holds that $\nabla_\xi \Bar{U}_M(\xi) = \zeta U^{\zeta-1}_M(\xi) \nabla_\xi U_M(\xi)$ and $\nabla_\xi U_M(\xi) = M^T \nabla_x U(x)$, hence we have that for any $M\in\M$
\begin{align*}
    \lVert\nabla_\xi \Bar{U}_M(\xi)\rVert &= \zeta \frac{\lVert\nabla_\xi U_M(\xi) \rVert}{U^{1-\zeta}_M(\xi)} \geq \frac{\zeta}{\lVert M^{-T}\rVert} 
    \, \frac{\lVert\nabla_x U(x) \rVert}{U^{1-\zeta}(x)} \numberthis \label{eq:condholdbisABPS}  \geq C \frac{\lVert\nabla_x U(x) \rVert}{U^{1-\zeta}(x)},
\end{align*}
where $C = \zeta \min_{M\in\M} (1/\lVert M^{-T}\rVert) >0$. Assumption~\ref{ass:BPS_growth}(c) implies that there exist $\Tilde{c}_1,\Tilde{R}_1>0$ such that $\lVert\nabla_x U(x) \rVert/U^{1-\zeta}(x) \geq \Tilde{c}_1$ for any $x$ such that $\lVert x\rVert \geq \Tilde{R}_1$, and because $\M$ is a compact space, we can choose $c_1=C\, \Tilde{c}_1$ and $R_1=\tilde{R}_1$, which are then independent of $M$.

The constant $c_2$ must be such that $\ell(\xi)\le c_2$, and because for a Gaussian $\psi$, we may take as written above $\ell(\xi)=1$. Therefore $c_2=1$ for each $M\in\M$.

Then, $c_3$ must be such that $(\lVert\nabla_\xi U_M(\xi) \rVert/\lVert\nabla_\xi \Bar{U}_M(\xi) \rVert ) \ge c_3$ for all $\xi$ such that $\xi > R_3$ for some $R_3 >0$. Indeed for $x=M\xi$ we have
\begin{align*}
    \frac{\lVert\nabla_\xi U_M(\xi) \rVert}{\lVert\nabla_\xi \Bar{U}_M(\xi) \rVert} &= \frac{\lVert\nabla_\xi U_M(\xi) \rVert}{\zeta U^{\zeta-1}_M(\xi) \lVert\nabla_\xi U_M(\xi) \rVert} = \frac{1}{\zeta} U^{1-\zeta}_M(\xi) = \frac{1}{\zeta} U^{1-\zeta}(x).
\end{align*}
By part (b) of Assumption~\ref{ass:BPS_growth} we have that $\lim_{\lVert x \rVert \to \infty} U(x)=+\infty$ and since $\zeta \in (0,1)$ there exist $c_3$ and $R_3$ large enough such that outside of a ball of radius $R_3$ we have $(\lVert\nabla_\xi U_M(\xi) \rVert/\lVert\nabla_\xi \Bar{U}_M(\xi) \rVert) \ge c_3$ for any $M\in\M$.

It is left to show that $c_4$ can be chosen uniform. For $A_{x,M} \coloneqq \{ \theta \in\mathbb{R}^d:\eta\lVert\theta\rVert^2\leq3\Bar{U}_M(M^{-1} x)\}$, $c_4$ must be such that 
\begin{equation}\label{eq:c_4}
    \lVert \nabla^2 \Bar{U}_M(\xi) \rVert \left(\sup_{\theta \in A_{x,M}} \lVert\theta\rVert^2 \right)\leq c_4 \qquad \textnormal{for any }\xi\textnormal{ such that }\lVert\xi\rVert>R_4
\end{equation}
for some $R_4>0$. Observing that for $\theta \in A_{x,M}$ it holds that $\sup_{\theta \in A_{x,M}} \lVert\theta\rVert^2 \leq \frac{3}{\eta}U_M^\zeta(\xi)$, we obtain 
\begin{align*}
    \lVert \nabla^2 \Bar{U}_M(\xi) \rVert \left(\sup_{\theta \in A_{x,M}} \lVert\theta\rVert^2 \right) &\leq \Big( \zeta(1-\zeta) U_M^{\zeta-2}(\xi) \lVert \nabla_\xi U_M(\xi) (\nabla_\xi U_M(\xi))^T \rVert \\
    & \quad+ \zeta U_M^{\zeta-1}(\xi) \lVert \nabla_\xi^2 U_M(\xi)\rVert \Big) \frac{3}{\eta}U_M^\zeta(\xi) \\
    & \leq 3\frac{\zeta(1-\zeta)}{\eta} \left(\frac{\lVert \nabla_\xi U_M(\xi)\rVert}{U_M^{1-\zeta}(\xi)} \right)^2 + 3\frac{\zeta}{\eta} \frac{\lVert \nabla_\xi^2 U_M(\xi)\rVert}{U_M^{1-2\zeta}(\xi)} \\
    & \leq 3\frac{\zeta(1-\zeta)}{\eta} \lVert M^T\rVert^2 \left(\frac{\lVert \nabla_x U(x)\rVert}{U^{1-\zeta}(x)} \right)^2 + 3 \lVert M\rVert \lVert M^T\rVert \frac{\zeta}{\eta} \frac{\lVert \nabla_x^2 U(x)\rVert}{U^{1-2\zeta}(x)} \\ 
    &\leq 3\frac{\zeta(1-\zeta)}{\eta} D^2 \left(\frac{\lVert \nabla_x U(x)\rVert}{U^{1-\zeta}(x)} \right)^2 + 3 D^2 \frac{\zeta}{\eta} \frac{\lVert \nabla_x^2 U(x)\rVert}{U^{1-2\zeta}(x)}.  \numberthis \label{eq:c3ABPS}
\end{align*}
In the second to last inequality we used that $\nabla_\xi^2 U_M(\xi) = M^T \nabla_x^2 U(x) M$ with $x=M\xi$, and in the last inequality we defined $D=\max_{M\in\M} \{\lVert M\rVert \vee \lVert M^T\rVert\}$. Part (c) of Assumption~\ref{ass:BPS_growth} implies that there exist $\tilde{c}_4,\tilde{R}_4$ such that 
\begin{equation*}
    \frac{\lVert \nabla_x U(x)\rVert}{U^{1-\zeta}(x)} \leq \tilde{c}_4, \quad \frac{\lVert \nabla_x^2 U(x)\rVert}{U^{1-2\zeta}(x)} \leq \tilde{c}_4 \qquad \textnormal{for all } x \textnormal{ such that } \lVert x\rVert \geq  \tilde{R}_4.
\end{equation*}
As a consequence, because $\M$ is a compact space and by \eqref{eq:c3ABPS}, there exist $c_4,R_4$ independent of $M$ such that \eqref{eq:c_4} holds for all $M \in\M$. 
It is now sufficient to take $R=\max\{R_1,R_3,R_4\}$, and we have shown that $c_1,c_2,c_3,c_4,R$ can be picked uniformly for all standard BP samplers with energy function $U_M$. In particular there is no dependence on $\lambda_r$ in all these constants.

It is important to observe that part (c) of Assumption~\ref{ass:BPS_growth} is satisfied by all BP samplers with targets in $\{\Tilde{\pi}_M\}_{M\in\M}$ because we have bounds on $\lVert M \rVert$ and $\lVert M^{-1} \rVert$. 
Moreover, parts (a) and (b) of Assumption~\ref{ass:BPS_growth} are trivially verified because of the transformation scheme in Figure~\ref{fig:trans_scheme} that defines $\tilde{U}_M$. It is then possible to apply \cite[Lemma 7]{BPS_Durmus} to each process to obtain the drift condition \eqref{eq:driftBPstd}.
Therefore for each $M\in\M$ and $\lambda_r \in \Lambda_r$ we have the Lyapunov function
\begin{equation}\label{eq:lyapBPorig}
    \tilde{V}_{M,\lambda_r}(\xi,\theta) = \exp{\left(\kappa_{\lambda_r} U_M^\zeta(\xi)\right)} \varphi_{\lambda_r}\left(\frac{2}{r c_1} \langle\theta,\nabla_\xi \Bar{U}_M(\xi)\rangle \right) + \exp{\left(\eta \lVert \theta\rVert^2 \right)},
\end{equation}
where $\eta \in (0,1)$ and $r$ depend only the distribution at refreshments $\psi$, $\kappa_{\lambda_r} \in (0,1)$ depends on the $c_i$'s, and $\varphi_{\lambda_r}$ is a positive, non decreasing, continuously differentiable function as defined in \cite{BPS_Durmus}. Notice that, as a consequence of the calculations above,  $\kappa_{\lambda_r}$ and $\varphi_{\lambda_r}$ do not depend on the preconditioner.
Applying our usual transformation scheme we obtain the functions 
\begin{equation}\label{eq:lyapBPtransf}
    V_{M,\lambda_r}(x,\theta) = \exp{\left(\kappa_{\lambda_r} U^\zeta(x)\right)} \varphi_{\lambda_r}\left(\frac{2}{r c_1} \langle\theta,M^T\nabla_x\Bar{U}(x)\rangle \right) + \exp{\left(\eta \lVert \theta\rVert^2 \right)}.
\end{equation}
This means that for each $\tilde{V}_M$ as defined in \eqref{eq:lyapBPorig}, the same proof of \cite[Lemma 7]{BPS_Durmus} can be followed and thus $\tilde{V}_M$ satisfies
\begin{equation}
     \tilde{\mathcal{L}}_{M,\lambda_r} \tilde{V}_{M,\lambda_r}(\xi,\theta) \leq -A_1(\lambda_r) \tilde{V}_{M,\lambda_r}(\xi,\theta) + A_2(\lambda_r) \qquad \qquad \text{for all } (\xi,\theta) \in \mathbb{R}^d \times \mathbb{R}^d,
\end{equation}
In particular, $A_1,A_2$ both depend on $\lambda_r$, on $c_1,c_2,c_3,c_4,R$, and on other constants that depend only on $\psi$ and thus are independent of $M$. Therefore $A_1,A_2$ can be chosen uniformly in $M\in\M$. 

The last step is now to eliminate the dependence on $\lambda_r$ in $A_1,A_2$. Observe that $A_1(\lambda_r)$ is defined in the proof of \cite[Lemma 7]{BPS_Durmus} as a minimum of several constants. It is enough for our needs to notice that the constants have a continuous dependence in $\lambda_r$ and that we are considering $\lambda_r \in [\lambda_{\text{min}},\lambda_{\text{max}}]$. On the other hand, $A_2(\lambda_r)$ can be chosen independently of $\lambda_r$ by continuity of $V_{M,\lambda_r}$ in $\lambda_r$.
We can therefore choose $A_1,A_2$ independently both of $M$ and $\lambda_r$, thus the wanted simultaneous drift condition follows by \eqref{eq:finalBPdrift}.
\end{proof}

\subsubsection{Finalising the proof of Theorem~\ref{thm:ergodicityABPS}}\label{sec:finalBPS}
Let $\Delta t>0$ be a discretisation step. Lemma~\ref{lemma:driftABP} gives the drift condition, that is condition (b) in Theorem~\ref{thm:ergodicity_continuoustime}. Then Lemma~\ref{lemma:couplABP} implies that a coupling inequality holds for any compact set. Sets of the form $V_{M,\lambda_r}(x,\theta) + V_{M,\lambda_r}(x,\theta) \leq C_2$ are compact by definition of the class of Lyapunov functions $\{V_{M,\lambda_r}:M\in\M,\lambda_r\in\Lambda_r \}$. We are in particular free to choose the constant $C_2$ as large as we wish. Note also that the coupling inequality holds for all $t>0$, hence there are no constraints on the choice of $\Delta t$.
Moreover $\mu(V_{M,\lambda_r})=\Tilde{\mu}_M(\Tilde{V}_{M,\lambda_R})<\infty$ for all $M\in\M$ and $\lambda_r\in\Lambda_r$. Here $\Tilde{V}_{M,\lambda_R}$ is the Lyapunov function of a standard BPS with refreshment rate $\lambda_r$ and target $\Tilde{\mu}_M = \tilde{\pi}_M \times \Psi$.
The containment condition is thus verified as all conditions in part (2) of Theorem~\ref{thm:ergodicity_continuoustime} hold.
Proposition~\ref{prop:dimi_adapt_zzs_bps} implies the diminishing adaptation condition, and thus ergodicity follows.

\subsection{Proving the diminishing adaptation condition}
A key part of Theorem~\ref{thm:RobRos} is condition (b), i.e. the diminishing adaptation condition. For the adaptive scheme described in Section \ref{sec:algorithms} the condition can be easily shown to be true as the adaptation happens with diminishing probability.
\begin{proposition}\label{prop:dimi_adapt_zzs_bps}
Consider the adaptive schemes in Section~\ref{sec:adap_schemes}. In particular, assume that $\{ p_n\}_{n\geq 0}$, i.e. the sequence of probabilities of updating the adaptation parameters, is such that $p_n \to 0$ as $n\to\infty$. Then the diminishing adaptation holds for any $t\geq 0$.
\end{proposition}
\begin{proof}
Consider for example the adaptive BPS. 
Observe that $M_{n+1}= M_n$ and $\lambda_r^{n+1}=\lambda_r^{n}$ with probability $1-p_{n+1}$ and thus $$\lVert P^{\Delta t}_{M_{n+1},\,\lambda_r^{n+1}}((x,\theta),\cdot) - P^{\Delta t}_{M_n,\,\lambda_r^{n}}((x,\theta),\cdot)\rVert_{\textnormal{TV}} \leq 2 p_{n+1} \to 0 \qquad \text{as } n\to \infty .$$ Thus the diminishing adaption holds. The same reasoning works for the adaptive ZZS.
\end{proof}

\section{Numerical experiments}\label{sec:num_exp}
In this section we test the empirical performance of the adaptive schemes we defined in Section \ref{sec:algorithms}. All experiments are implemented in Julia and the corresponding codes can be found at  \url{https://github.com/andreabertazzi/Adaptive_PDMC_samplers}.
Let us state some settings that hold for all experiments below. The time horizon is set to $T=10^5$ for all processes. This is large enough for the adaptive PDMC samplers to learn and take advantage of the covariance structure. When considered fixed, the refreshment rate of the BPS is taken to be $\lambda_r=1$. The excess switching rate for ZZS is set to $0$ in all experiments. The discretisation step is chosen to be $\Delta t =0.5$, which in our experiments turns out to be a good choice for a wide range of targets. Moreover, we set adaptation times to be every $t_{\text{adap}}=2000$ continuous time units. The probability of adapting decays as $\mathcal{O}(\log\log n)$. Finally, no normalisation in the sense discussed at the end of Section \ref{sec:transf_matrix} is employed.
The performance measures we consider are the \textit{effective sample size per second} (ESS/sec) for the mean and for the radius statistic  $t(x)= \sum_{i=1}^d x_i^2$. In Sections \ref{sec:gauss_exp} and \ref{sec:logistic} these are computed in continuous time as discussed in \cite{ZZ} by estimating the asymptotic variance with the batch means method, and the variance of the Monte Carlo estimate on the continuous time trajectories of the processes. On the other hand, in Section \ref{sec:gauss_mixture} we compute the mean squared error (MSE) in discrete time for the sample mean and radius statistic and take advantage of the fact that in the large time horizon regime the MSE is approximately given by the asymptotic variance of the observable divided by the number of generated samples. This alternative way to compute the ESS avoids poor convergence of the batch means method in the multimodal case.
Finally, in all settings we repeat the same task $20$ times and report all the results in box-plots.

\subsection{Multidimensional Gaussian target}\label{sec:gauss_exp}
In this section we focus on two different kinds of multivariate Gaussian target distributions. The first one, denoted by $\textbf{MG1}$, has unitary variances and correlation $\rho$ between all components. Denoting the covariance matrix by $\Sigma$, this means that $\Sigma_{ii}=1$ for each $i=1,\dots,d$ and $\Sigma_{ij}=\rho$ for all $i\neq j$. We study how the adaptive PDMC algorithms compare to their non-adaptive counterparts for different values of $\rho$ and different dimensionalities. In this setting we focus on adaptive algorithms that estimate the full covariance matrix. The second Gaussian target we consider has variances $0.5,1,5,10,15$ repeated depending on the dimension, together with a milder correlation between components. This setting is denoted as $\textbf{MG2}$ and is useful to compare all kinds of adaptive algorithms we introduced. 
\subsubsection{\textbf{MG1} target}
\begin{figure}
\centering
\begin{subfigure}{\textwidth}
  \centering
    \includegraphics[width=\textwidth]{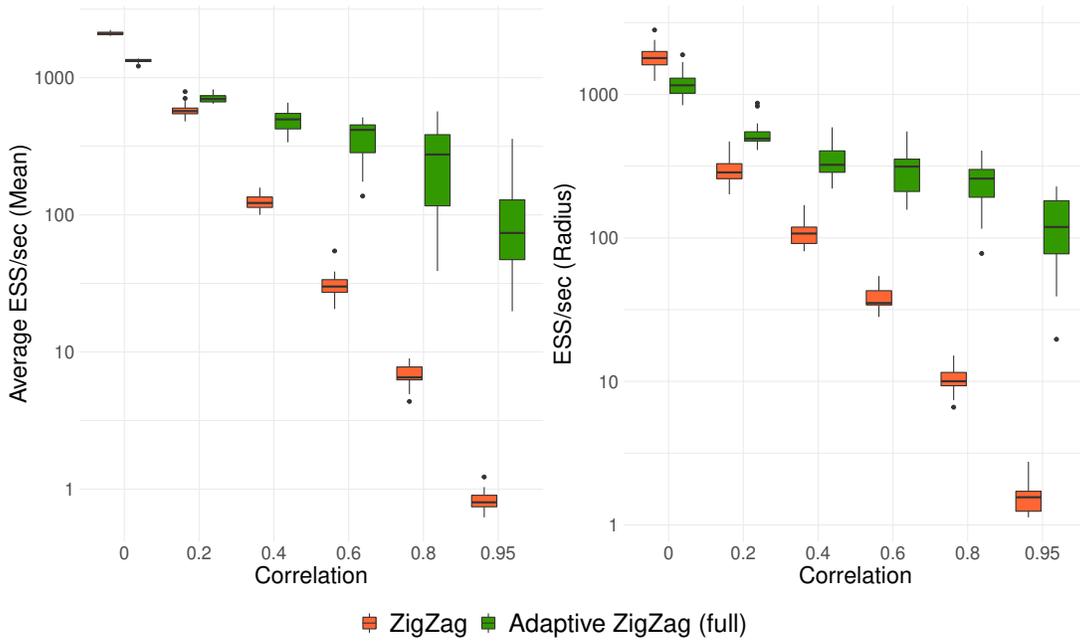}
    \caption{Comparison between the ZZS and the ZZS with adaptive preconditioner learning the \textbf{full} covariance matrix. }
    \label{fig:AZZcorr_ESS_dim50}
\end{subfigure}
\begin{subfigure}{\textwidth}
  \centering
    \includegraphics[width=\textwidth]{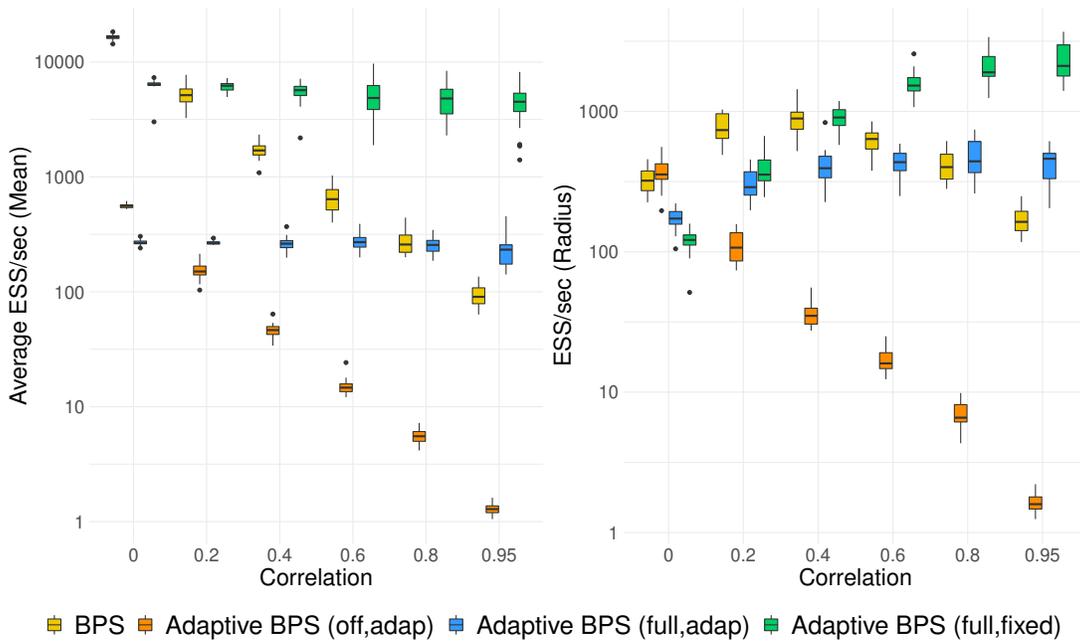}
    \caption{Comparison between the BPS and various alternative adaptive BPS's. \textbf{Adaptive BPS (full,adap)} denotes the BPS that learns the entire covariance matrix (\textbf{full}) and with adaptive refreshment rate (\textbf{adap}). Adaptation of the preconditioner can be turned \textbf{off}, and similarly the refreshment rate can be \textbf{fixed}.}
    \label{fig:ABPScorr_ESS_dim50}
\end{subfigure}
\caption{Results as a function of the correlation for \textbf{MG1} targets.}
\label{fig:MG1_corr}
\end{figure}

In the first experiment we consider a $50$-dimensional $\textbf{MG1}$ target for different values of $\rho$. In Figure~\ref{fig:MG1_corr} the average ESS/sec and the ESS/sec for the radius statistic are shown. As expected, the performance of the ZZS is degrading as the correlation increases. This behaviour is for the most part caused by the very large number of events that have to be simulated for very narrow targets. The adaptive ZZS successfully improves over this inconvenience and is stable with respect to the increasing correlation. The standard BPS with fixed refreshment rate shows a decaying average ESS/sec, while the ESS/sec for the radius statistic appears to increase as $\rho$ grows up until $\rho=0.4$ and then becomes smaller. This behaviour is likely due to the fact that the choice $\lambda_r=1$ is more suited for the estimation of the radius in case of a more concentrated target rather than for a standard Gaussian. A similar behaviour is shown by the adaptive BPS's. Overall we notice a marked improvement for the BPS with adaptive preconditioner and fixed $\lambda_r$. Choosing to adapt only the refreshment turns out to be a detrimental decision when the target is correlated. Indeed the optimality criterion derived \cite{bierkens2018highdimensional} assumes a standard Gaussian target.

\begin{figure}
  \centering
    \includegraphics[width=\textwidth]{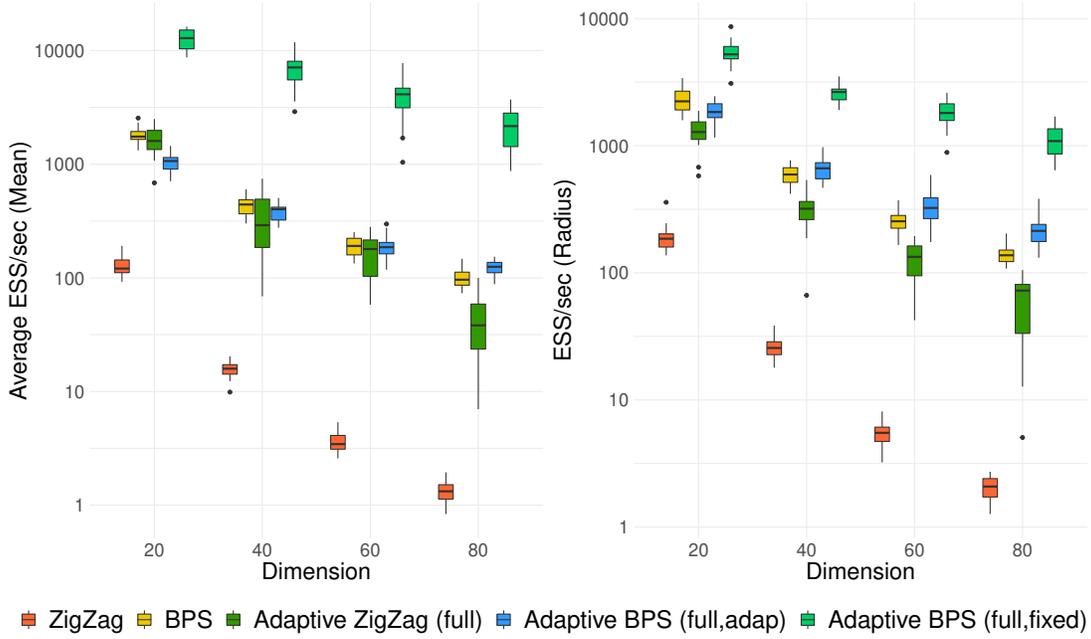}
    \caption{\textbf{MG1} target with $\rho = 0.8$ and different dimensionalities.}
    \label{fig:MG1_dim}
\end{figure}

In Figure \ref{fig:MG1_dim} we study how the adaptive schemes compare to the standard ones for an \textbf{MG1} target with correlation $\rho = 0.8$ and increasing dimensionalities of target.
The plots show that when the target is strongly correlated the effect of the adaptation shows no sign of diminishing. It also seems clear that the sampler of choice in this case should be the adaptive BPS with fixed, rather than adaptive, refreshment rate. This could be due to the fact that, when the refreshment rate is updated adaptively and the target is anisotropic, a too large $\lambda_r$ is chosen at first due to the high number of reflections, thus slowing down the estimation of the covariance matrix. As suggested by a reviewer, one could avoid this issue by keeping the refreshment rate fixed until the estimate of the covariance stabilises, and only then starting to learn the optimal $\lambda_r$.
It is worth pointing out that the performance of the BPS with adaptive preconditioner and refreshment improves compared to the BPS as the dimension increases. This is according to the theoretical results in \cite{bierkens2018highdimensional}, which are indeed obtained in the high dimensional limit. Therefore we expect that for a large $d$ it is reasonable to apply both the transformation scheme and the tuning of $\lambda_r$. 
\begin{figure}
\centering
\begin{subfigure}{0.9\textwidth}
  \centering
    \includegraphics[width=\textwidth]{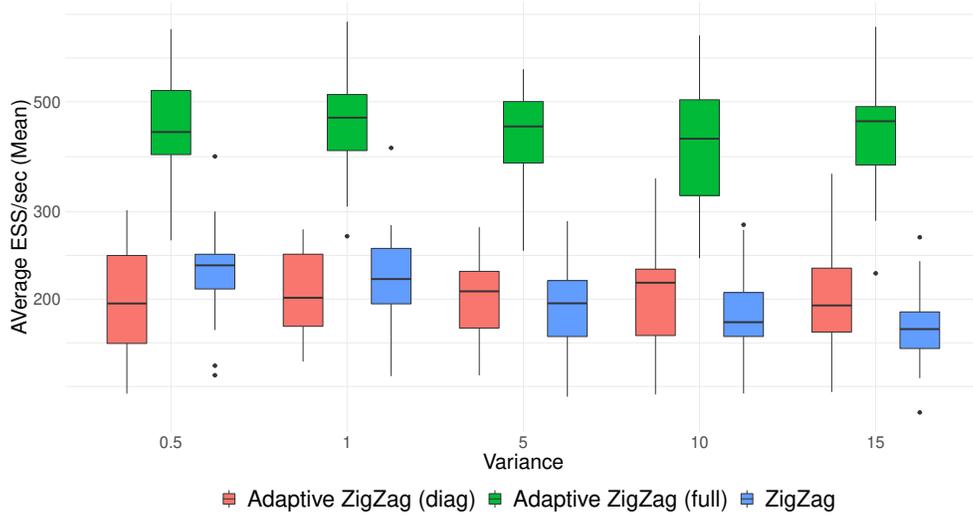}
    \caption{Results of the \textbf{MG2} experiment for the ZZS. The option \textbf{diag} refers to the adaptive algorithm that learns only the diagonal of the covariance matrix. }
    \label{fig:AZZSvars_dim50}
\end{subfigure}
\begin{subfigure}{0.9\textwidth}
  \centering
    \includegraphics[width=\textwidth]{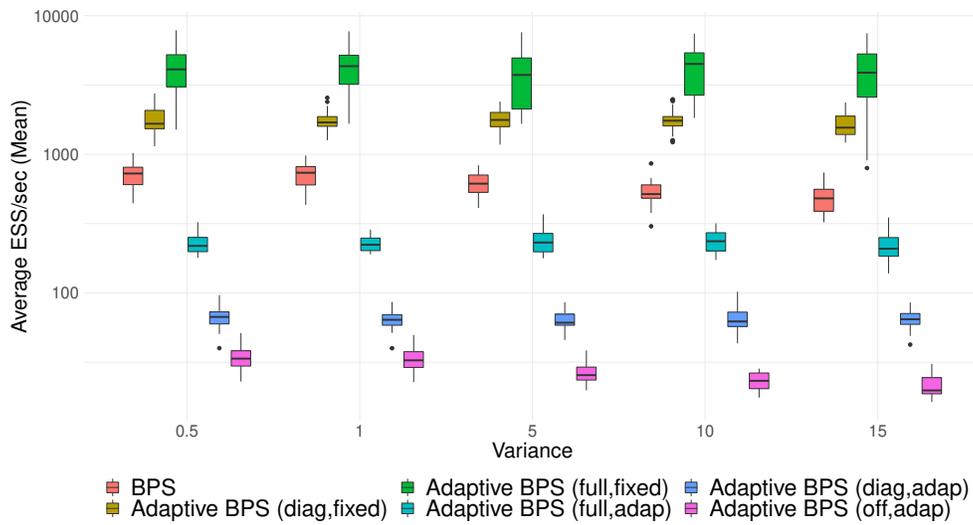}
    \caption{Results of the \textbf{MG2} experiment for the BPS.}
    \label{fig:ABPSvars_dim50}
\end{subfigure}
\caption{Comparison of several samplers in the context of Section \ref{sec:MG2}.}
\label{fig:MG2}
\end{figure}

\subsubsection{\textbf{MG2} target}\label{sec:MG2}
Let us now consider a $50$-dimensional \textbf{MG2} target with a mild correlation set to $\rho=0.3$. Figure~\ref{fig:MG2} shows the results for several adaptive PDMC samplers. The adaptive algorithms that learn the entire covariance matrix show the largest gain in terms of ESS/sec. For the BPS the choice of learning only the variance of each component of the target seems interesting, also in view of larger dimensions. As in the previous section, we observe that the adaptation of the refreshment rate seems to have a bad effect for anisotropic targets.



\subsection{Logistic regression with correlated data}\label{sec:logistic}
\begin{figure}
  \centering
    \includegraphics[width=0.9\textwidth]{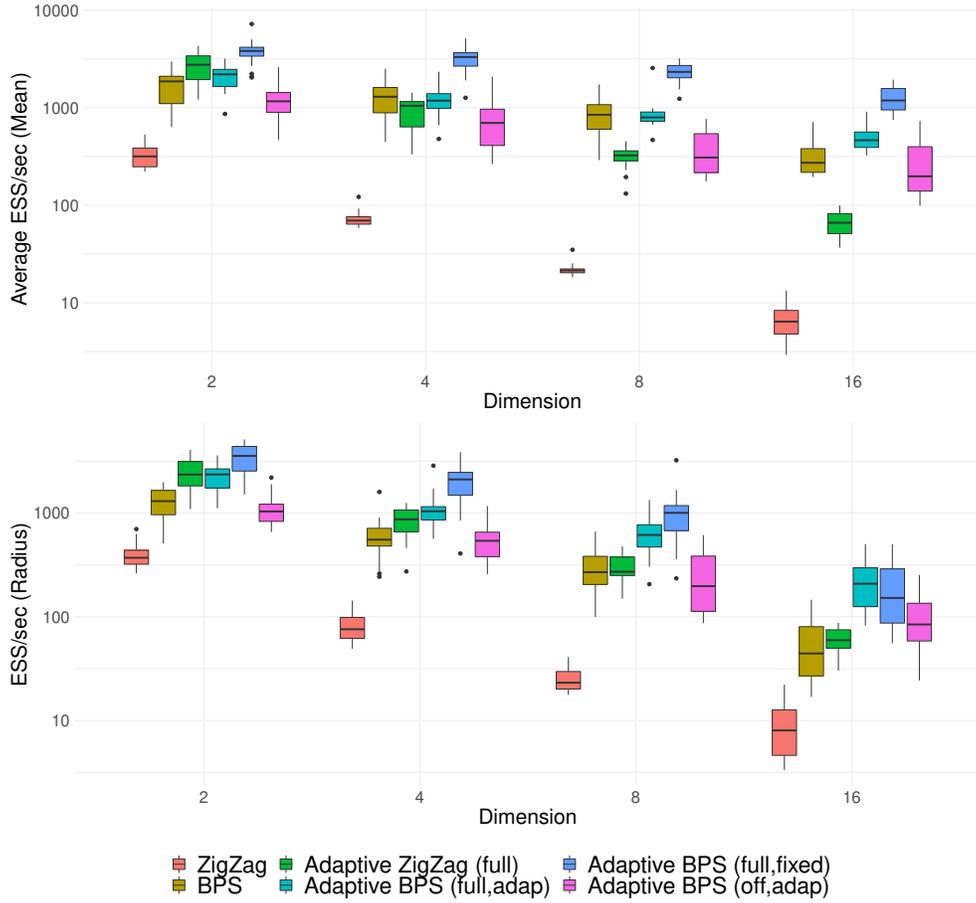}
    \caption{Logistic regression task of Section \ref{sec:logistic}.}
    \label{fig:logistic_by_dimension}
\end{figure}
The next numerical experiment we consider is a Bayesian logistic regression task. In this setting for $j=1,\dots,n_{\textnormal{obs}}$ a binary output value $y_j\in\{0,1\}$ has distribution
$$
\mathbb{P}(Y_j = 1 | \beta) = \frac{1}{1+\exp{(-\beta^T x_j)}},
$$
where $\{x_j \}_{j=1}^{n_\textnormal{obs}}$ are known covariates, and $\beta \in \mathbb{R}^d$ is an unknown parameter. We take a flat prior and thus obtain the posterior 
$$
\pi(\beta|\{y_j\}_{j=1}^{n_\textnormal{obs}}) \propto \prod_{j=1}^{n_\textnormal{obs}} \frac{ \exp{(-y_j\beta^T x_j)}}{1+\exp{(-\beta^T x_j)}}.
$$
We force correlation between some components of the parameter by taking, for $j=1,\dots,n_{\textnormal{obs}}$ and $i=1,\dots,d$, $(x_j)_i= 1 + \varepsilon N_{ji}$, where $N_{ji} \sim \mathcal{N}(0,1)$ and $\varepsilon = 0.1$. The results of the experiment are reported in Figure \ref{fig:logistic_by_dimension}, in which the samplers are tested with targets as above with $d=2,4,8,16$ and $n_{\textnormal{obs}}=1000$. 
The adaptation of the refreshment rate follows the alternative scheme discussed in Appendix \ref{sec:adap_strategy_appendix}. This scheme seems more stable as the refreshment rate is update gradually and cannot jump immediately to very large or small values.
Although the dimensionality is small and the correlation is limited to a subset of the coordinates, we observe that the adaptive schemes outperform their standard counterparts.

\subsection{Mixture of Gaussian distributions}\label{sec:gauss_mixture}
\begin{figure}[ht]
\centering
\begin{subfigure}{0.9\textwidth}
  \centering
    \includegraphics[width=\textwidth]{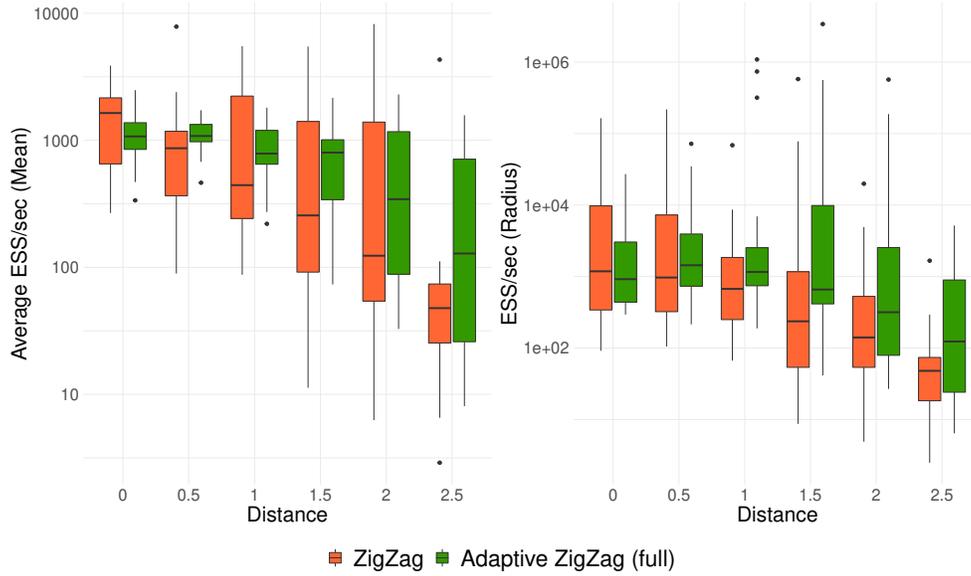}
    \caption{Results for the ZZS with time horizon $T=10^5$. }
    \label{fig:mixture_BPS}
\end{subfigure}
\begin{subfigure}{0.9\textwidth}
  \centering
    \includegraphics[width=\textwidth]{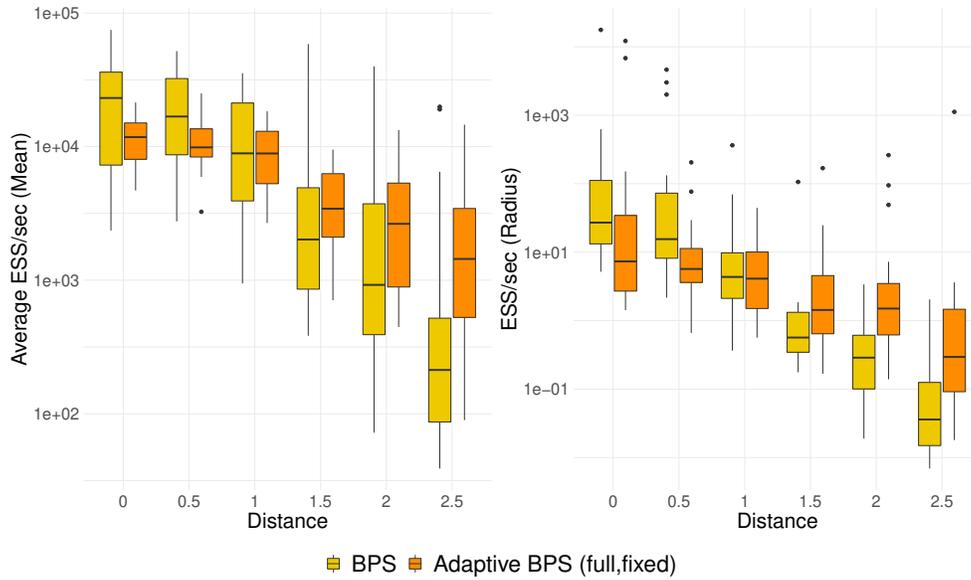}
    \caption{Results for the BPS with time horizon $T=4\times 10^5$.}
    \label{fig:mixture_ZZS}
\end{subfigure}
\caption{Numerical results for a mixture of two Gaussian distributions as described in Section \ref{sec:gauss_mixture}. }
\label{fig:mixture}
\end{figure}
Consider a mixture of two $30$-dimensional Gaussian distributions $\mathcal{N}(0_d,\Sigma)$ and $\mathcal{N}(\mu,\Sigma)$, both with weight $\frac{1}{2}$. Here we take $\Sigma$ with $\Sigma_{ii}=1$ and $\Sigma_{ij}=0.25$ for $j\neq i$. Moreover we take $\mu=a\times(1,\dots,1)$, where $a$ is a parameter that determines the distance between the two means. We investigate the performance of the adaptive schemes compared to standard ones as a function of the parameter $a$. The details on the implementation of this experiment can be found in Appendix \ref{app:gauss_mixture}. In this experiment we test the robustness of the algorithm in a case in which the target distribution is cigar shaped, but multimodal. The time horizon for ZZS is $10^5$, while for BPS it is $4 \times 10^5$. 
Figure \ref{fig:mixture} shows the results for the adaptive samplers in this setting. 
We observe that as the distance between the means increases the performance of the adaptive samplers improves over the standard ones. We remark that for small values of the distance $a$ one cannot expect improvements of the adaptive algorithms as the correlation is small and thus the target is not very anisotropic.

\section{Discussion}
In this paper we proposed adaptive schemes to overcome two of the current issues with the BPS and ZZS. We have shown that the refreshment rate and the excess switching rate can be tuned on the fly as long as the updates or the probabilities of updating decrease to $0$. With this approach the user does not have to worry about tuning the refreshment rate. A current limitation is that more theory or experiments are needed to determine a criterion that works well with anisotropic targets. In addition to this, we have proposed a way to make the PDMC samplers learn and take advantage of the covariance structure of the target. The theoretical results stated in Theorems \ref{thm:ergADZZ} and \ref{thm:ergodicityABPS} ensure that the adaptive samplers are ergodic and thus converge to the correct measure $\pi$. It is challenging to prove theoretical statements regarding performance improvements of the adaptive schemes over the standard PDMC samplers. However, the numerical experiments we conducted suggest that our adaptive algorithms can lead to a significant performance improvement when there are strong anisotropies. 
An alternative approach could be to use an optimisation algorithm to obtain an estimate of the covariance matrix by computing the Hessian of the target at its point of maximum. However, this optimisation step can be expensive and moreover for realistic problems it is not a given that this estimator is a good approximation of the posterior covariance. In particular in our theory we do not assume log convexity of the target, and also we do not assume that we are in the large sample regime where a Bernstein-von Mises theorem holds.  
In addition, the adaptive schemes discussed in this paper can be applied to the Boomerang sampler \cite{bierkens2020boomerang}. This would result in elliptical dynamics that are adapted to resemble the (unimodal) target at hand. Naturally, the adaptive algorithms should be run with a time horizon that is large enough to benefit from the adaptation. Two other important settings are the discretisation step and the time between two adaptations. Based on our experience with the experiments, we suggest $\Delta t = \mathcal{O}(10^{-1})$ and $t_{\text{adaps}}= \mathcal{O}(10^{3})$. For concentrated targets, as for instance posteriors when there is a very large number of data points, it is suggested to choose both values small. We remark that in very high dimensional settings it may be unfeasible to let the samplers learn the full covariance matrix, as the computation of $M$ entails calculating the square root of the empirical covariance matrix. In such cases we suggest either learning only the diagonal elements or blocks of the covariance.
We remark that the adaptive PDMC algorithms with subsampling are applicable in the setting of tall data, that is when data-set is made of a large number of observations, but with a moderate dimensionality. In such settings subsampling can be shown empirically to result in an improved efficiency (in terms of ESS per second); a result which is backed by a heuristic argument, based on posterior contraction, i.e., the Bernstein-Von Mises theorem; see \cite{ZZ} for details. More research is necessary to understand in which situations subsampling can lead to improved efficiency, and in particular if improved efficiency is possible in cases for which the Bernstein-von Mises theorem does not apply; see also \cite{bardenet2017markov,johndrow2020no} for a critical discussion of subsampling methods.

A question that one could naturally ask is how applicable this transformation scheme is in case of a multimodal target. The answer depends on the specific target at hand, but one can design a target as a mixture of Gaussian distributions for which applying the transformation scheme would not speed up the convergence of the sampler. However, when the target is multimodal it is possible for instance to use the adaptive PDMC samplers together with the framework proposed in \cite{pompe2018framework}. In this framework, the adaptive PDMC samplers would be beneficial since the regions around each mode would be explored more efficiently by taking advantage of the covariance structure of the specific mode. 

Finally, we remark that the idea of learning the covariance structure of the target on the fly could be applied to obtain adaptive versions of the Hamiltonian Monte Carlo (HMC) algorithm \cite{Neal} and of the Metropolis Adjusted Langevin Algorithm (MALA) \cite{MALA}. In particular both the HMC algorithm and the MALA are sensitive to correlation in the target and can thus benefit from a suitable preconditioner, as argued respectively in Section 4.1 of \cite{Neal} and in \cite{RobertsStramer}.
Moreover, the preconditioner could be chosen to take advantage of the geometry of the target, as proposed in \cite{girolamicalderhead} for HMC and MALA. The preconditioner could be estimated adaptively with an appropriate adaptation strategy, together with similar ideas presented in this manuscript.

\appendix

\section{Implementation of adaptive PDMC algorithms}\label{app:implementation}
The main issue in the exact simulation of PDMPs lies in the simulation of the switching times. It is generally proposed to use Poisson thinning to overcome this difficulty. The idea is to find upper bounds for the switching rates that are more tractable, then simulate a Poisson process with said rate and finally adjust with an acceptance-rejection step. Applying a preconditioning matrix to PDMPs results in a modification of the switching rates and thus the bound proposed in \cite{ZZ,BPS} do not directly apply to our proposed algorithms. In the following two sections we give two important examples to motivate that it is possible to find bounds with similar ideas to those used for standard PDMC algorithms. In Appendix \ref{sec:adap_strategy_appendix} we discuss a different adaptation strategy for the refreshment rate.
\subsection{Dominated Hessian of the negative log-likelihood}\label{sec:bdd_hessian}
Let us consider the case in which the Hessian of the negative log-likelihood is dominated by a positive definite matrix. Denoting the Hessian by $H_U(x)=(\partial_i\partial_j U(x))_{i,j=1}^d$, this means that for any $x\in\mathbb{R}^d$ it holds that $\langle H_U(x) u,u\rangle \leq \langle Qu,u\rangle$ for all $u \in \mathbb{R}^d$. Then we say that $-Q\preceq H_U(x)\preceq Q$. Assuming $Q$ is symmetric, we have for all $u,v \in \mathbb{R}^d$ that $\langle u,H_U(x)v\rangle \leq \lVert u\rVert_2 \lVert Qv\rVert_2$.

Let us consider the switching rates of a ZZ process with preconditioner $M$, i.e. $\lambda_i^M(x,\theta)=(\theta_i \langle M_i,\nabla U(x)\rangle )_+$, where $M_i$ is the $i$-th column of $M$. Remember that the deterministic trajectory of this process is $x(t)=x+M\theta t $, where $(x,\theta)$ is the initial condition. We have
\begin{align*}
    \partial_i U(x(t)) &= \partial_i U(x) + \int_0^t \sum_{j=1}^d \partial_j\partial_i U(x(s)) (M\theta)_j ds \\
    & = \partial_i U(x) + \int_0^t \langle H_U(x(s)) e_i, M\theta\rangle  ds,
\end{align*}
where $e_i$, is the $i$-th vector of the canonical basis, i.e. with zeros in all components except for a $1$ in the $i$-th component. Taking $a_i=\theta_i \langle M_i, \nabla U(x)\rangle$ and $b_i= \sqrt{d}\lVert M\rVert_2 \lVert Q M_i\rVert_2$ it follows that
\begin{align*}
    \lambda_i^M(x(t),\theta) &= \left( \theta_i \langle M_i, \nabla U(x)\rangle +\theta_i \sum_{j=1}^d M_{ji} \int_{0}^t \langle H_U(x(s)) e_j, M\theta\rangle  ds \right)_+ \\
    & \leq \left( \theta_i \langle M_i, \nabla U(x)\rangle +  \int_{0}^t \langle H_U(x(s)) M_i, M\theta\rangle  ds \right)_+ \\
    &\leq \left( a_i + t \lVert M\theta\rVert_2 \lVert Q M_i\rVert_2  \right)_+ \leq (a_i+b_it)_+ .
\end{align*}
It is possible to have a tighter bound taking $b_i = \lVert Q M_i\rVert_2 \lVert M\theta\rVert_2$, but this entails having to update $b_i$ after every switching time, which would make the overall simulation more expensive.
Note that in particular for a diagonal $M$ we have that $\lVert M\theta\rVert_2 = (\sum_{i=1}^d M_{ii}^2)^{1/2}$ and we can take $a_i=\theta_i M_{ii} \partial_i U(x)$ and $b_i= M_{ii} \lVert Q_i\rVert_2 (\sum_{i=1}^d M_{ii}^2)^{1/2}$. It is worth observing that $\lVert Q\rVert_2$ can  be computed once before running the algorithm and then stored. Terms $\lVert M\rVert_2$ and $\lVert M_i\rVert_2$ can in principle be estimated by taking a maximum over the class of matrices $\mathcal{M}$. However, having loose computational bounds leads to a lower percentage of accepted proposed events and thus to an increased computational burden.

For BPS with preconditioner $M$ we showed that $\lambda_M(x,\theta)=(\langle M\theta,\nabla U(x)\rangle)_+$. With computations similar to the ones above we obtain
\begin{align*}
    \lambda_M(x(t),\theta) &= \left( \langle M\theta,\nabla U(x) \rangle + \int_0^t \langle M\theta ,  H_U(x(s)) M\theta  \rangle ds  \right)_+ \\
    & \leq \left( \langle M\theta,\nabla U(x) \rangle + t \langle M\theta ,  Q M\theta  \rangle \right)_+ \leq (a+bt)_+
\end{align*}
with $a = \langle M\theta,\nabla U(x) \rangle$ and $b= \langle M\theta ,  Q M\theta  \rangle$. Note that $M\theta$ is the velocity of the process and needs to be computed in any case to determine the trajectories.

\subsection{Subsampling techniques}\label{sec:subsampling}
PDMC algorithms are particularly interesting because they allow for exact subsampling. This means that it is possible to modify PDMC samplers such that the correct invariant measure is maintained but without going through the entire data-set at each iteration. This was illustrated in \cite{ZZ,BPS} for ZZS and BPS. Subsampling techniques are particularly helpful when the number of data points $n$ is much larger than the dimensionality $d$ of the posterior density function. This is an interesting scenario for the adaptive algorithms that are here introduced, as for a small $d$ the additional computations necessary to learn and take the square root of (part of) the covariance matrix are not overpowering. In this section we explain how adaptive PDMC with subsampling can be implemented.

Suppose the partial derivatives of the posterior distribution can be written in the form 
\begin{equation}
    \partial_i U(x) = \frac{1}{n} \sum_{j=1}^n E_i^j(x),
\end{equation}
where $E_i^j:\mathbb{R}^d \to \mathbb{R}^d$ are continuous mappings. This holds for example when the target satisfies
\begin{equation*}
    U(x) = \frac{1}{n}\sum_{j=1}^d U^j(x).
\end{equation*}
This is the case for instance when the target is a posterior density of a model with iid observations. In such cases one can choose $U^j(x) = -\log(\pi_0(x))-n\log{(l(y^j|x))}$, where $\pi_0$ is a prior of the parameters and $l$ is the likelihood associated to observation $y^j$. Then the idea is to define a collection of switching rates along a trajectory.
\begin{equation*}
    m_i^j(t) = (\theta_i \langle M_i, E^j(x+M\theta t))_+,
\end{equation*}
where $E^j(x)= (E_1^j(x),\dots,E_d^j(x))^T$. If we can find uniform bounds $H_i(t)$ such that $m_i^j(t) \leq H_i(t)$ for all $j$, then the subsampling procedure is successful. One can simulate IPPs with rates $H_i(t)$ and then accept the proposed switch of component $i_0$ with probability $m_{i_0}^J(t)/H_{i_0}(t)$ where $J\sim\textnormal{Unif}\{1,\dots,n \}$. The next proposition shows that stationarity of $\pi$ is preserved for the ZZS with subsampling  also introducing a preconditioner.
\begin{proposition}
The Zig-Zag sampler with subsampling and preconditioner $M$ has invariant distribution $\mu = \pi \times \text{Unif}(\Theta)$. Moreover it coincides with a Zig-Zag process with switching rates
\begin{equation*}
    \lambda_i^M(x,\theta) = \frac{1}{n}\sum_{j=1}^n (\theta_i \langle M_i, E^j(x)\rangle)_+ \quad \text{for } (x,\theta)\in E,\, i=1,\dots,d.
\end{equation*}
\end{proposition}
\begin{proof}
The statements can be derived following the proof of Theorem 4.1 in \cite{ZZ} and is thus omitted.
\end{proof}
An interesting situation is that of a negative log-likelihood with Lipschitz partial derivatives, i.e. $\lvert \partial_i U(x) - \partial_i U(y)\lvert \leq C_i \lVert x-y\rVert_p$. In this case we can take a reference point $x^*$ and take
\begin{align*}
    E^j(x) = \nabla U(x^*) + \nabla U^j(x) - \nabla U^j(x^*).
\end{align*}
This allows us to obtain the following uniform bound on the switching rates
\begin{align*}
    m^j_i(t) &= (\theta_i \langle M_i, \nabla U(x^*) + \nabla U^j(x+M\theta t) - \nabla U^j(x^*)\rangle )_+\\
    & \leq \left(\theta_i \langle M_i,\nabla U(x^*)\rangle + \lvert \langle M_i,\nabla U^j(x+M\theta t)-\nabla U^j(x^*)\rangle \rvert \right)_+\\
    & \leq  \left(\theta_i \langle M_i,\nabla U(x^*)\rangle + \sum_{l=1}^d C_l\lvert  M_{li}\rvert  \lVert x+M\theta t - x^* \rVert_p \right)_+\\
    & \leq \left(\theta_i \langle M_i,\nabla U(x^*)\rangle + \langle \lvert M_i\rvert,C\rangle \lVert x-x^*\rVert_p + t \langle \lvert M_i\rvert,C\rangle \lVert M\theta \rVert_p \right)_+.
\end{align*}
Therefore we can simulate $d$ IPPs with rates $H_i(t)=(a_i+b_it)_+$, where $a_i = \theta_i \langle M_i,\nabla U(x^*)\rangle  + \langle \lvert M_i\rvert,C\rangle \lVert x-x^*\rVert_p$ and $b_i = d^{1/p} \langle \lvert M_i\rvert,C\rangle \lVert M \rVert_p$.

The same estimator for the gradient of the negative log-density can be used for a preconditioned BPS. With computations analogous to above we find
\begin{align*}
    m_i^j(t) \leq \left( \langle M\theta, \nabla U(x^*)\rangle + \langle \lvert M\theta \rvert, C\rangle \lVert x-x^*\rVert_p + t \langle \lvert M\theta \rvert, C\rangle \lVert M \theta \rVert_p \right)_+ \leq H_i(t).
\end{align*}
Although establishing ergodicity of these adaptive algorithms would be an interesting research question, we leave it for future research due to the lack of theoretical results for their standard counterparts.

\subsection{A different adaptation strategy for the refreshment rate}\label{sec:adap_strategy_appendix}
The adaptive scheme for the refreshment rate described in Section \ref{sec:algorithms} is a natural implementation of the results in \cite{bierkens2018highdimensional}. However, it can be unstable when combined with an adaptive preconditioner. Here we define an alternative adaptive strategy and we prove that the diminishing adaptation condition from Theorem \ref{thm:erg_adaptiveMCMC} is satisfied.

Assume that the refreshment rate of the BPS is constant in the position and velocity spaces. We update it iteratively as follows:
\begin{equation}\label{eq:adap_scheme_refresh}
    \lambda^r_{n+1} = \begin{cases}
    &\lambda^r_{n} + q_{n+1} \qquad \textnormal{ if } \frac{n_{\textnormal{refresh}}(n+1)}{n_{\textnormal{events}}(n+1)} < \lambda^*, \\
    & \lambda^r_{n} - q_{n+1} \qquad \textnormal{ if } \frac{n_{\textnormal{refresh}}(n+1)}{n_{\textnormal{events}}(n+1)} > \lambda^*,
    \end{cases}
\end{equation}
in which $n_{\textnormal{refresh}}(n)$ and $n_{\textnormal{events}}(n)$ are respectively the number of refreshments and the number of events up to time $n$, $q_{n}$ is a positive, decreasing sequence such that $q_n\to 0$, and $\lambda^* = 0.7812$. 

For the strategy in Section \ref{sec:algorithms} we used that at time $n$ there is a probability of adapting equal to $p_n$, such that $p_n \to 0$.
In the remainder of this section we show that for the  adaptive scheme described in \eqref{eq:adap_scheme_refresh} the diminishing adaptation condition holds even if the refreshment rate is updated with probability $1$. 
\begin{lemma}\label{lemma:bdd_pert}
Let $\pi$ be a probability distribution. Denote as $P^t_\gamma$ the semigroup of a ZZ process with excess switching rate $\gamma : E\to \mathbb{R}^d_+$. Let $\gamma_1,\gamma_2$ be two bounded excess switching rates. Then for any $t \geq 0$
\begin{equation}\label{eq:conv_zzs_perturbation}
    \sup_{(x,\theta)\in E} \lVert \delta_{(x,\theta)} P_{\gamma_1}^{t} - \delta_{(x,\theta)} P_{\gamma_2}^{t} \rVert_{\textnormal{TV}} \to 0 \qquad \text{as } \lVert \gamma_1-\gamma_2\rVert_\infty \to 0.
\end{equation}

Similarly, denote now by $P^t_\lambda$ the semigroup of a BPS with refreshment rate $\lambda : E\to \mathbb{R}_+$. Let $\lambda_1,\lambda_2$ be two bounded refreshment rates. Then for any $t \geq 0$
\begin{equation}\label{eq:conv_bps_perturbation}
    \sup_{(x,\theta)\in E} \lVert \delta_{(x,\theta)} P_{\lambda_1}^{t} - \delta_{(x,\theta)} P_{\lambda_2}^{t} \rVert_{\textnormal{TV}} \to 0 \qquad \text{as } \lVert \lambda_1-\lambda_2\rVert_\infty \to 0.
\end{equation}
\end{lemma}
\begin{proof}
Consider first the case of the BPS.
Consider the generators of two BPS's with refreshment rates $\lambda_1$ and $\lambda_2$ and denote them respectively by $\mathcal{L}_{\lambda_1}$ and $\mathcal{L}_{\lambda_2}$.
Now observe that $\mathcal{L}_{\lambda_2} = \mathcal{L}_{\lambda_1} + B$, where $B$ is a bounded operator such that 
\begin{equation*}
    Bf(x,\theta) = (\lambda_2(x,\theta)-\lambda_1(x,\theta)) \int_{\mathbb{R}^d} \left(f(x,\theta')-f(x,\theta)\right)\psi(d\theta').
\end{equation*}
In other words, $\mathcal{L}_{\lambda_2}$ is a perturbed version of $\mathcal{L}_{\lambda_1}$, with bounded perturbation  $B$. Since $\mathcal{L}_{\lambda_1}$ and $\mathcal{L}_{\lambda_2}$ generate strongly continuous semigroups $(P^t_{\lambda_1})_{t\geq 0}, (P^t_{\lambda_2})_{t\geq 0}$, we can apply the reasoning in the proof of Corollary 1.11 in Chapter 3 of \cite{engel_nagel} to obtain that, for any $t>0$, $\lVert P^{t}_{\lambda_1}- P^{t}_{\lambda_2} \rVert \to 0$ as $B\to 0$, which is equivalent to 
\begin{equation*}
   \sup_{f\in \mathcal{C}_0(E), \lvert f\rvert \leq 1} \, \sup_{(x,\theta)\in E} \lvert P^{t}_{\lambda_1}f(x,\theta) - P^{t}_{\lambda_2}f(x,\theta) \rvert \to 0 \qquad \text{as } \lVert \lambda_1-\lambda_2\rVert_\infty \to 0,
\end{equation*}
where $\mathcal{C}_0(E)$ is the space of continuous functions on $E$ that vanish at infinity.
Inverting the order of the two suprema and applying the Riesz-Markov representation theorem (see e.g. \cite{Arveson96noteson}) we obtain the result in  \eqref{eq:conv_bps_perturbation}.

Now consider the case of two ZZS with two different bounded excess switching rates $\gamma_1:E\to\mathbb{R}^d_+$ and $\gamma_2:E\to\mathbb{R}^d_+$. The perturbation is now given by $$\tilde{B}f(x,\theta) = \sum_{i=1}^d ((\gamma_2(x,\theta))_i-(\gamma_1(x,\theta))_i) (f(x,F_i \theta)-f(x,\theta)).$$ 
Then the result \eqref{eq:conv_zzs_perturbation} follows by the same arguments as above.
\end{proof}

\begin{proposition}\label{prop:dim_adap_appendix}
Consider the adaptive PDMC algorithms defined in Section \ref{sec:algorithms}, with the following modification. Let the adaptation of the refreshment rate be performed with probability $1$ at adaptation times, but in such a way that $\sup_{(x,\theta)\in E} \lvert\lambda_r^{n+1}(x,\theta) - \lambda_r^n(x,\theta)\rvert \leq \delta_{n+1}$, where $(\delta_n)_{n\geq 1}$ is such that $\delta_n \to 0$ as $n\to\infty$, and similarly for the refreshment rate of the ZZS.
Then for any $\Delta t>0$
\begin{equation}\notag
        \lim_{n \to \infty} \left( \sup_{(x,\theta)\in E} \lVert P^{\Delta t}_{M_{n+1},\,\lambda_r^{n+1}}((x,\theta),\cdot) - P^{\Delta t}_{M_n,\,\lambda_r^{n}}((x,\theta),\cdot)\rVert_{\textnormal{TV}} \right)=0 \quad \text{ in probability},
    \end{equation}
    and similarly for the ZZS.
\end{proposition}
\begin{proof}
By the triangle inequality for the total variation distance we obtain 
\begin{align}\label{eq:ineq_dim_adapt}
\begin{split}
    \lVert \delta_{(x,\theta)}P^{\Delta t}_{M_{n+1},\,\lambda_r^{n+1}} - \delta_{(x,\theta)} P^{\Delta t}_{M_n,\,\lambda_r^{n}}\rVert_{\textnormal{TV}} &\leq \lVert P^{\Delta t}_{M_{n+1},\,\lambda_r^{n+1}}((x,\theta),\cdot) - P^{\Delta t}_{M_{n+1},\,\lambda_r^{n}}((x,\theta),\cdot)\rVert_{\textnormal{TV}} \\
    &  \quad + \lVert P^{\Delta t}_{M_{n+1},\,\lambda_r^{n}}((x,\theta),\cdot) - P^{\Delta t}_{M_n,\,\lambda_r^{n}}((x,\theta),\cdot)\rVert_{\textnormal{TV}}.
\end{split}
\end{align}
The inequality above allows us to deal with the two adaptations separately. The second term in the right hand side of \eqref{eq:ineq_dim_adapt} is handled by Proposition \ref{prop:dimi_adapt_zzs_bps}.

Now focus on the first term in the right hand side of \eqref{eq:ineq_dim_adapt}, for which we want to apply Lemma \ref{lemma:bdd_pert}. It is sufficient to observe that supremum norm of the perturbation goes to zero in probability when $n\to \infty$. Indeed the sequence of refreshment rates is Cauchy in probability, as for all $n$ we have
$\lVert \lambda_r^{n+1} - \lambda_r^n \rVert_\infty \leq  \delta_{n+1} $. Convergence in probability follows from the fact that $(\delta_{n})_{n\geq 0}$ is a convergent sequence and therefore the diminishing adaptation condition holds.

\end{proof}

\subsection{Details on the implementation of a Gaussian mixture target}\label{app:gauss_mixture}
Consider the target of Section \ref{sec:gauss_mixture}, that is 
\begin{align*} \pi(x) & \propto \lambda \exp\left( - \frac 1 2 (x-\mu_1)^T \Sigma^{-1} (x- \mu_1) \right) + (1-\lambda)  \exp \left(- \frac 1 2 (x-\mu_2)^T \Sigma^{-1} (x- \mu_2) \right). \end{align*}
Here we wish to find upper and lower bounds to the Hessian of the negative log-density, which can be used to simulate the adaptive ZZS and BPS following the approach of Appendix \ref{sec:bdd_hessian}.
Start by writing $w = \frac 1 2 (\mu_2 + \mu_1)$ and $v = \frac 1 2 (\mu_2 - \mu_1)$.
Then we have
\begin{align*}
    \pi(x) & \propto \exp \left( -\frac 1 2 (x - w)^T \Sigma^{-1} (x-w) \right) \Big[ \alpha \exp \left(\mu_1^T \Sigma^{-1}x - w^T \Sigma^{-1} x \right)  \\
    & \qquad + \beta \exp \left( \mu_2^T \Sigma^{-1} x - w^T \Sigma^{-1} x \right)  \Big] \\
    & = \exp\left( -U_1(x) - U_2(x)\right)
\end{align*}
where
\begin{align*}
    & \alpha = \lambda \exp\left(\frac{1}{2}v^T\Sigma^{-1}(w+\mu_1)\right), \\
    & \beta = (1-\lambda) \exp\left(-\frac{1}{2}v^T\Sigma^{-1}(w+\mu_2)\right),
\end{align*}
while 
\[ U_1( x) = \frac 1 2 (x-w)^T \Sigma^{-1}(x-w) \]
and
\[ U_2(x) = - \log \left[  \alpha \exp(-(\Sigma^{-1} v)^T x) + \beta \exp((\Sigma^{-1} v)^T x) \right].\]
Clearly $\nabla^2 U_1(x) = \Sigma^{-1}$. Write $m(x) = (\Sigma^{-1} v)^T x$.  Then $\nabla_x m(x) = \Sigma^{-1} v$. 
We have
\[ \nabla U_2(x) = (\Sigma^{-1}v) \frac{ \alpha \exp(-m(x)) - \beta \exp(m(x)) }{ \alpha \exp(-m(x)) + \beta \exp(m(x)) },\]
and
\begin{align*}
    \nabla^2 U_2(x) & = (\Sigma^{-1} v)(\Sigma^{-1} v)^T \left[ -1  + \frac{(\alpha \exp(-m(x)) - \beta\exp(m(x)))^2}{(\alpha \exp(-m(x)) + \beta \exp(m(x)))^2} \right] \\
    & =  (\Sigma^{-1} v)(\Sigma^{-1} v)^T \left[ \frac{ -4 \alpha \beta }{(\alpha \exp(-m(x)) + \beta \exp(m(x)))^2} \right].
\end{align*}
It remains to find the minimum of
\[ m \mapsto  \frac{-4 \alpha\beta}{(\alpha \exp(-m) + \beta \exp(m))^2}.\]
This is achieved at $m = \frac 1 2 \ln \frac \alpha \beta$, yielding that
\[ - \frac 1 4  \left( \Sigma^{-1}(\mu_2-\mu_1)\right) \left( \Sigma^{-1}(\mu_2-\mu_1) \right)^T \preceq \nabla^2 U_2(x)  \preceq 0.\]
We conclude that
\begin{equation} \label{eq:bound} \Sigma^{-1} - \frac 1 4  \left( \Sigma^{-1}(\mu_2-\mu_1) \right) \left( \Sigma^{-1}(\mu_2-\mu_1) \right)^T  \preceq \nabla^2 U(x)  \preceq \Sigma^{-1}.\end{equation}

\begin{remark}[Some special situations]
Consider the case in which $v = \frac 1 2 (\mu_2 - \mu_1)$ is an eigenvector of $\Sigma^{-1}$ with eigenvalue $\gamma$.
Let $Q$ denote the positive definite matrix with eigenvalues identical to those of $\Sigma^{-1}$ along the directions orthogonal to $v$, and with eigenvalue $\max(\gamma, \gamma^2\|v\|^2-\gamma)$ along $v$. Then $-Q \preceq \nabla^2 U(x) \preceq Q$. In particular
\[ \|\nabla^2 U(x)\|_2 \le \max (\kappa, \gamma, \gamma^2\|v\|^2 - \gamma),\]
where $\kappa$ denotes the maximal eigenvalue of $\Sigma^{-1}$ restricted to the orthogonal complement of $v$. 
 
 In the special case for which $\Sigma = I_d$, we find for the lower bound
 \[ \nabla^2 U \succeq I_d - \frac 1 4(\mu_2 -\mu_1)(\mu_2-\mu_1)^T, \]
 with eigenvalues $1$ and $1 -  \frac 1 4 \|\mu_2 -\mu_1\|^2_2$. So if $\|\mu_2 -\mu_1\|^2_2 \le 2$, then $-I_d \preceq \nabla^2 U \preceq I_d$.
 In general
 \[ \|\nabla^2 U (x)\|_2 \le \max (1, 1/4 \|\mu_2-\mu_1\|^2_2 -1) \quad \text{for all $x\in\mathbb{R}^d$}. \]
\end{remark}

\section{Omitted proofs and technical results}\label{sec:other_proofs} 
\subsection{Proofs of Section~\ref{sec:adap_schemes}}\label{app:proofs_section_2}
\begin{proof}[Proof of Proposition~\ref{prop:generatorZZ}]
The extended generator $(\Tilde{\mathcal{L}}_M,\mathcal{D}(\Tilde{\mathcal{L}}_M))$ of the standard ZZ process is by \cite[Definition 14.15]{Davis1993MarkovM} such that for all $\Tilde{f} \in \mathcal{D}(\Tilde{\mathcal{L}}_M)$
\begin{equation}\label{eq:extgenZZ}
    M_t^{\Tilde{f}} = \Tilde{f}(\Xi(t),\Theta(t)) - \Tilde{f}(\xi,\theta) - \int_0^t \mathcal{\Tilde{\mathcal{L}}}_M \Tilde{f}(\Xi(s),\Theta(s)) ds
\end{equation}
is a local martingale. For any $\Tilde{f} \in \mathcal{D}(\Tilde{\mathcal{L}}_M)$ we define a function $h(x,\theta) = \Tilde{f}(\xi,\theta)$ with $x=M \xi$, for any $M \in \mathcal{M}$. Observe that $h(X(t),\Theta(t)) =\Tilde{f}(\Xi(t),\Theta(t))$ since $X(t) = M \, \Xi(t)$. The generator of the transformed process then satisfies Condition \eqref{eq:extgenZZ} provided that
\begin{equation} \notag
    \mathcal{L}_M h(x,\theta) = \Tilde{\mathcal{L}}_M \Tilde{f}(\xi,\theta) \qquad  \text{for all } (x,\theta) \in E.
\end{equation}
Therefore
\begin{equation}\notag
    \begin{aligned}
       \Tilde{\mathcal{L}}_M \Tilde{f}(\xi,\theta) &=  \langle \theta, \nabla_\xi h(M \xi,\theta) \rangle  + \sum_{i=1}^d \Tilde{\lambda}_{M,i}(\xi, \theta) (h(M \xi, F_i \theta)-h(M \xi,\theta))  \\
        & = \langle M \theta, \nabla_x h(x,\theta) \rangle + \sum_{i=1}^d \Tilde{\lambda}_{M,i} (M^{-1}x,\theta) (h(x,F_i \theta)-h(x,\theta)) \\
        &= \mathcal{L}_M h(x,\theta).
    \end{aligned}
\end{equation}
\end{proof}

\begin{proof}[Proof of Proposition~\ref{prop:invarianceAZZ}]
Following the approach in the proof of Theorem 2.2 in \cite{ZZ}, we want to check that for any $M \in \mathcal{M}$, and any $h \in \mathcal{D}(\mathcal{L}_M)$, it holds that $\int_E \mathcal{L}_M h(x,v) d\mu =0$. By a change of variable and using the same reasoning as in the proof of Proposition~\ref{prop:generatorZZ}, for any 
\begin{equation}\notag
    \begin{aligned}
        \int \mathcal{L}_M h(x,\theta) d\mu(x,\theta) & = \frac{1}{Z} \sum_{\theta \in \{-1,+1\}^d}  \int  \mathcal{L}_M h(x,\theta) \exp(-U(x)) dx \\
        &= \frac{1}{\Tilde{Z}} \sum_{\theta \in \{-1,+1\}^d} \int \Tilde{\mathcal{L}}_M \tilde{f}(\xi,\theta) \exp(- \Tilde{U}_M(\xi)) d\xi \\
        &= \int \Tilde{\mathcal{L}}_M \tilde{f}(\xi,\theta) d \Tilde{\mu}_M(\xi,\theta) \,\,=\,\,0,
    \end{aligned}
\end{equation}
where $\Tilde{\mu}_M = \Tilde{\pi}_M \otimes \textnormal{Unif}(\{-1,+1\}^d)$ the last equality was obtained by the invariance of $\Tilde{\mu}_M$ for the standard Zig-Zag process.
\end{proof}

\begin{proof}[Proof of Proposition~\ref{prop:generatorBPS}]
Following the same approach as in the proof of Proposition~\ref{prop:generatorZZ}, for any $\Tilde{f} \in \mathcal{D}(\Tilde{\mathcal{L}}_M)$ we define a function $h(x,\theta) = \Tilde{f}(\xi,\theta)$ with $x=M \xi$. The generator of the transformed process then satisfies Condition \eqref{eq:extgenZZ} provided that
\begin{equation} \notag
    \mathcal{L}_{M} h(x,\theta) = \Tilde{\mathcal{L}}_M \Tilde{f}(\xi,\theta) \qquad \text{for all } (x,\theta) \in E.
\end{equation}
In this case
\begin{equation}\notag
    \begin{aligned}
       \Tilde{\mathcal{L}}_M \Tilde{f}(\xi,\theta) &=  \langle \theta, \nabla_\xi h(M \xi,\theta) \rangle  + \Tilde{\lambda}_M (\xi, \theta) (h(M \xi, \tilde{R}(\xi) \theta)-h(M \xi,\theta))  \\
       & \quad + \lambda_r(\xi,\theta) \int (h(M \xi,\theta')-h(M \xi,\theta)) \psi(d\theta') \\
       & = \langle M \theta, \nabla_x h(x,\theta) \rangle + \Tilde{\lambda}_M(M^{-1}x,\theta) (h(x,\tilde{R}(M^{-1}x) \theta)-h(x,\theta)) \\
       & \quad + \lambda_r(M^{-1} x,\theta) \int (h(x,\theta')-h(x,\theta)) \psi(d\theta') \\
       &= \mathcal{L}_{M} h(x,\theta).
    \end{aligned}
\end{equation}
\end{proof}

\begin{proof}[Proof of Proposition~\ref{prop:invarianceABPS}]
As for Proposition~\ref{prop:invarianceAZZ}, it suffices to show that for all $M \in \mathcal{M}$ $\int_E \mathcal{L}_{M} h(x,v) d\mu =0$. By a change of variable we obtain
\begin{equation}\notag
    \begin{aligned}
        \int \mathcal{L}_{M} h(x,\theta) d\mu(x,\theta) & = \frac{1}{Z}  \int_\mathcal{X} \int_\mathcal{V}  \mathcal{L}_{M}h(x,\theta) \exp(-U(x)) \psi(\theta) d\theta dx \\
        &= \frac{1}{\Tilde{Z}} \int_\mathcal{X} \int_\mathcal{V} \Tilde{\mathcal{L}}_M \tilde{f}(\xi,\theta) \exp(- \Tilde{U}_M(\xi)) \psi(\theta) d\theta d\xi \\
        &= \int \Tilde{\mathcal{L}}_M \tilde{f}(\xi,\theta) d \Tilde{\mu}(\xi,\theta) \,\,=\,\,0,
    \end{aligned}
\end{equation}
where $\Tilde{\mu}_M = \Tilde{\pi}_M \otimes \psi$ and the last equality follows by the invariance of $\Tilde{\mu}_M$ for the BPS as shown in \cite{BPS}.
\end{proof}

\subsection{Proof of Theorem~\ref{thm:erg_adaptiveMCMC}}\label{app:proof_erg_adaptiveMCMC}
Let us initially separate the proof under Assumption~\ref{ass:mineZZ} and next under Assumption~\ref{ass:mineassBPS}. First we show that in both cases the containment condition holds if the sequence $\{V_{\Gamma_n}(Z_n) \}_{n\geq 0}$ is bounded in probability.
\subsubsection{The case of Assumption~\ref{ass:mineZZ}}
Let $\gamma \in \mathcal{Y}$ and let $C$ be the uniform $(\nu_\gamma,\delta,n_0)$-small set as defined in Assumption~\ref{ass:mineZZ}. As described in \cite{rob_ros_mcmc_survey,rosenthal_minorisation} we define the following coupling between two chains with kernel $P_\gamma$. The two chains evolve independently each with kernel $P_\gamma$ when they are outside of $C$. When both chains are inside $C$, then with probability $\delta$ we let the two chains couple after $n_0$ steps by drawing $X_{n+n_0}=Y_{n+n_0}$ from $\nu_\gamma$, or with probability $(1-\delta)$ we independently draw  $X_{n+n_0} \sim (P_\gamma^{n_0}(X_n,\cdot)-\delta \nu_\gamma(\cdot)) / (1-\delta) $ and $Y_{n+n_0} \sim (P_\gamma^{n_0}(Y_n,\cdot)-\delta \nu_\gamma(\cdot)) / (1-\delta) $.
Then by taking $h_\gamma(x,y) = (V_\gamma(x)+V_\gamma(y))/2$ as suggested in the proof of Theorem 3 in \cite{BaiRob} we have
\begin{equation}\notag
    \mathbb{E}_\gamma \left( h_\gamma (X_1,Y_1) \big| X_0 = x,Y_0 = y \right) \le \lambda h_\gamma (x,y) \qquad \text{for all } (x,y) \notin C \times C .
\end{equation}
Define now $A(\gamma) \coloneqq \sup_{(x,y) \in C \times C}  \mathbb{E}_\gamma \left( h_\gamma (X_{n_0},Y_{n_0}) \big| X_0 = x,Y_0 = y \right)$. 
Then by \cite[Theorem 5]{rosenthal_minorisation} we have for each $\gamma \in \mathcal{Y}$ the following bound on the total variation distance
\begin{equation}\notag
    ||P_\gamma^n (x,\cdot) - \mu(\cdot)||_{TV} \le (1-\delta)^{ \floor*{\frac{\sqrt{n}}{n_0}} }  +\lambda^{n - \sqrt{n}n_0 +1} (A(\gamma))^{\sqrt{n} - 1}(V_{\gamma}(z)+\pi(V_\gamma))/2.
\end{equation}
The dependence on $\gamma$ in $A(\gamma)$ can be eliminated by noting that $$A(\gamma) \leq \sup_{\gamma \in \mathcal{Y}} (A(\gamma)) \leq\, \lambda^{n_0} \sup_{ \{ \gamma \in \mathcal{Y},\, x \in C \}} V_\gamma (x) + b n_0=:B,$$ which is independent of $\gamma$. In particular by our assumptions we have that $B<\infty$ and $\pi(V_\gamma)<\infty$ for each $\gamma$. 

\subsubsection{The case of Assumption~\ref{ass:mineassBPS}}
Define for each $V_\gamma$ and any measurable function $\varphi:E \to \mathbb{R}$ and any $\beta \geq 0$ the norm
\begin{equation}\label{eq:phinorm}
    \lVert \varphi \rVert_{\beta,V_\gamma} = \sup_{z \in E} \frac{\lvert \varphi(z) \rvert}{1 + \beta V_\gamma(z)} .
\end{equation}
For measures $\mu_1$, $\mu_2$ on $E$ such that, for each $\gamma \in \mathcal{Y}$, $\mu_1(V_\gamma) < \infty$, $\mu_2(V_\gamma) < \infty$, consider the $V_\gamma$ weighted norm 
\begin{equation}\label{eq:rhonorm}
    \rho_\beta(\mu_1,\mu_2) = \sup_{\{\varphi: \lVert \varphi \rVert_{\beta,V_\gamma} \le 1 \}} \left( \mu_1(\varphi) - \mu_2(\varphi)  \right).
\end{equation}
Note that $\rho_\beta$ depends on $\gamma$. Now for each $\gamma$ we can apply \cite[Theorem 24]{BPS_Durmus} and therefore there exist $\beta^*>0$ and $\kappa \in (0,1)$ such that
\begin{equation*}
    \rho_{\beta^*}(\mu_1 P_\gamma,\mu_2 P_\gamma) \leq \kappa \rho_{\beta^*}(\mu_1,\mu_2),
\end{equation*}
where in particular $\kappa$ and $\beta^*$ depend only on $\alpha, \lambda, C_1$ as defined in Assumption~\ref{ass:mineassBPS}. Therefore $\kappa$ and $\beta^*$ do not depend on $\gamma$ by uniformity of said constants.
The norm \eqref{eq:phinorm} is such that if $\lVert \varphi \rVert_{0,V_\gamma} \leq 1$, then $\lVert \varphi \rVert_{\beta,V_\gamma} \leq 1$ for any $\beta>0$. Choosing $\mu_1=\delta_z$,  $\mu_2 = \pi$, for any $z \in E$  we obtain 
\begin{align*}
    \lVert P_\gamma^n (z,\cdot) - \pi(\cdot) \rVert_{\textnormal{TV}} &\le \rho_{\beta^*}(\delta_z P^n_\gamma,\pi) \\
    & \leq \kappa^n \rho_{\beta^*}(\delta_z,\pi) \\
    & = \kappa^n \sup_{\{\varphi: \lVert \varphi \rVert_{\beta^*,V_\gamma} \le 1 \}} \left( \varphi(z) - \pi(\varphi)  \right) \\
    & \leq \kappa^n \left( 2 + \beta^* V_\gamma(z) + \beta^*\pi(V_\gamma) \right) 
\end{align*}
Here we used \eqref{eq:rhonorm} in the third equality, and twice \eqref{eq:phinorm} in the last  inequality. Moreover by assumption $\pi(V_\gamma) < + \infty$ for all $\gamma \in \mathcal{Y}$.

\subsubsection{Boundedness in probability of the sequence of Lyapunov functions}
In both cases the containment condition (assumption (a) in Theorem~\ref{thm:RobRos}) is thus satisfied if the process $\left\{ V_{\Gamma_n} (Z_n) \right\}_{n\ge 0}$ is bounded in probability. This is indeed implied by our assumption that $\Gamma_n$ is updated only if $Z_n \in B$, where $B$ is a compact set, as suggested in  \cite{craiu2015,chimisov2018adapting,pompe2018framework}. For the sake of completeness we report the main steps. 

By \cite[Lemma 3]{Roberts2007CouplingAE} it is enough to show that $\sup_{n \in \mathbb{N}} \mathbb{E}(V_{\Gamma_n}(Z_n)) < \infty$. By our assumptions and letting $D \coloneqq \sup_{\{z \in B,\, \gamma \in \mathcal{Y} \}} V_\gamma(z) < \infty$
\begin{align}\notag
     \mathbb{E}(  &V_{\Gamma_{n+1}}(Z_{n+1}) |Z_n=z, \Gamma_n= \gamma) = \\ 
    & \notag =\mathbb{E}(V_{\Gamma_{n+1}}Z_{n+1} \left( \mathbbm{1}(Z_{n+1} \notin B) + \mathbbm{1}(Z_{n+1} \in B) \right)| Z_n = z, \Gamma_n= \gamma ) \\
    & \notag \le \mathbb{E}(V_{\Gamma_{n}}(Z_{n+1}) \mathbbm{1}(Z_{n+1} \notin B) |Z_n = z, \Gamma_n= \gamma) + D \\
    & \notag \le \lambda V_\gamma(z) + b + D
\end{align} 
Taking expectations on both sides one obtains
$$
    \mathbb{E}(  V_{\Gamma_{n+1}}(Z_{n+1})) \le \lambda  \mathbb{E}( V_{\Gamma_{n}}(Z_n)) + b + D.
$$
This shows that the sequence is contracting since $\lambda\in (0,1)$ and this is enough by
\cite[Lemma 2]{Roberts2007CouplingAE} to show our boundedness in probability of the sequence.

\subsection{Proof of Theorem~\ref{thm:ergodicity_continuoustime}}\label{app:proof_ergodicity_continuoustime}
To prove Theorem~\ref{thm:ergodicity_continuoustime} it is sufficient to verify that the discretised process with time step $\Delta t$ satisfies the assumptions of Theorem~\ref{thm:erg_adaptiveMCMC}.

Let us first consider the assumptions in alternative (1) of Theorem~\ref{thm:ergodicity_continuoustime}. The uniform small set condition of the discretised process (Assumption~\ref{ass:mineZZ}(a)) follows by choosing $n_0= \frac{t_0}{\Delta t}$. The geometric drift condition (Assumption~\ref{ass:mineZZ}(b)) is a consequence of Lemma~\ref{lemma:drift2drift}.

Now consider the second set of assumptions, that is alternative (2) in the theorem. In this case we wish to show that Assumption~\ref{ass:mineassBPS} holds for the discretised process. Part (a) of Assumption~\ref{ass:mineassBPS} is immediately verified by the fact that $\Delta t = t_0$. For part (b), first observe that from Lemma~\ref{lemma:drift2drift} by our assumptions we have
$$
    P_\gamma^{\Delta t} V_\gamma (z) \leq e^{-A_1 \Delta t}\, V_\gamma(z) + \frac{A_2}{A_1} (1-e^{-A_1\Delta t}).
$$
In the notation of Assumption~\ref{ass:mineassBPS}(b) we have $\lambda = e^{-A_1 \Delta t}$ and $C_1=\frac{A_2}{A_1}$. In particular we have $2C_1<C_2$ since we assumed that $C_2> 2 A_2/A_1$.

The thesis now follows in both cases by Theorem~\ref{thm:erg_adaptiveMCMC}.

\subsection{Other technical results}

\subsubsection{Small set condition for a family of one-dimensional ZZ processes}\label{sec:app_smallset_1d}
In this section we show a simultaneous minorisation condition for the ZZ process in the one-dimensional case. This result is essential to generalise the result to the multidimensional case (see Lemma \ref{lemma:smallZZ} and its proof).

\begin{assumption}[Assumption 3 in \cite{ZZCurie}]\label{ass:smallsetZZ}
Let $U \in \mathcal{C}^2(\mathbb{R})$. Define the switching rates of a $1$-d ZZ process as $\lambda(x,\theta)=(\theta U'(x))_+$, where $x\in\mathbb{R}$ and $\theta \in \{-1,+1\}$. There exists $x_0 > 0$ such that
\begin{align}
    \notag \inf_{x \ge x_0} \lambda(x,+1) &> \sup_{x \ge x_0} \lambda(x,-1), \\
    \notag \inf_{x \ge -x_0} \lambda(x,-1) &> \sup_{x \le -x_0} \lambda(x,+1).
\end{align}
\end{assumption}
\begin{remark}
Assumption~\ref{ass:smallsetZZ} requires that there exists a point $x_0\in\mathbb{R}$ such that the process has a strictly positive probability of changing direction if it is outside of and moving outwards of the set $[-x_0,+x_0]$. The lemma below states that a small set condition follows from this assumption.
\end{remark}

\begin{lemma}\label{lemmma:smallsetZZ1D}
Consider the family of 1-dimensional Zig-Zag processes with generators $\{\mathcal{L}_m: m\in\M \}$ where $\M=[V_{\textnormal{min}},V_{\textnormal{max}}]$ for some $0<V_{\textnormal{min}}\le V_{\textnormal{max}}<\infty$. Thus for $f \in \mathcal{D}(\mathcal{L}_m)$ and $m\in [V_{\textnormal{min}},V_{\textnormal{max}}]$
\begin{equation}\notag
    \mathcal{L}_m f(x,\theta) = m\theta f'(x,\theta) + (m(\theta U'(x))_+ + \gamma(x)) (f(x,-\theta)-f(x,\theta)).
\end{equation}
Assume either of the two conditions:
\begin{itemize}
    \item $\gamma(x) = 0$ for all $x\in\mathbb{R}$, but Assumption~\ref{ass:smallsetZZ} holds;
    \item there exists $\gamma_{\textnormal{min}}>0$ such that $\gamma(x) \ge \gamma_{\textnormal{min}}$ for any $x\in \mathbb{R}$, and $\gamma(\cdot)$ is bounded on compact sets.
\end{itemize}
Then for any set of the form $C=D \times V $, where $D \subset \mathbb{R}$ is a compact set and $V \subseteq \{-1,+1 \}$, there exists $t_0>0$ such that for any $t \ge t_0$ there are $\delta >0$, and a probability measure $\nu$ on $E$ such that 
\begin{equation}\notag
    P^{t}_m((x,\theta),\cdot) \ge \delta \nu(\cdot) \qquad \textit{for all }(x,\theta) \in C \textit{ and all }  m\in\M.
\end{equation}

Moreover, consider $\{\mathcal{L}_{m,\gamma}: m \in \M, \gamma \in \Lambda \}$, where $\Lambda$ is a family of switching rates $\gamma:\mathbb{R} \to \mathbb{R}_+$ such that for all $\gamma \in \Lambda$ it holds that $0<\gamma_{\textnormal{min}}\leq\gamma(x)\leq\gamma_{\textnormal{max}}<\infty$ for all $x\in \mathbb{R}$. Then for any compact set $C$ as above, there exists $t_0>0$ such that for any $t \ge t_0$ there are $\delta >0$, and a probability measure $\nu$ on $E$ that satisfy
\begin{equation}\notag
    P^{t}_{m,\gamma}((x,\theta),\cdot) \ge \delta \nu(\cdot) \qquad \textit{for all }(x,\theta) \in C, \, \gamma \in \Lambda,   \, m\in\M.
\end{equation}
In particular $t_0$ depends on $V_{\textnormal{min}}$, and $\delta$ depends on $V_{\textnormal{min}},V_{\textnormal{max}},\gamma_{\textnormal{min}},\gamma_{\textnormal{max}}$, but neither depends on $m$ or $\gamma$.
\end{lemma}
\begin{proof}
Let $R>0$ and $C=[-R,R] \times V$, where $V \subseteq \{-1,+1 \}$. Consider the family $\{\mathcal{L}_m: m\in\M \}$. For $m \in [V_{\textnormal{min}},V_{\textnormal{max}}]$, the process $\mathcal{L}_m$ moves with velocity $m\theta$, so either $+m>0$ or $-m <0$. Consider first the case in which Assumption~\ref{ass:smallsetZZ} is satisfied for some $x_0 >0$. Note that as a consequence it is satisfied for any $\mathcal{L}_m$. We consider the case $x_0 > R$ because it is the most general. Indeed one is always free to take $\Tilde{x}_0>R>x_0$ and apply the reasoning below. Alternatively the case $x_0 \leq R$ follows by the same method of proof. 

For $B=A \times V \subset E$ let $\nu(B)= (\text{Leb}(A \cap [-R,R])/\text{Leb}([-R,R])) \times (\delta_{1}(\theta) + \delta_{-1}(\theta))$ be a probability measure on $E$. Define $T = T(\varepsilon) \coloneqq  (4x_0+2R)/V_{\text{min}} + \varepsilon$ where $\varepsilon > 0$. This choice of the time horizon allows the process to reach the region with strictly positive probability of a velocity flip for any initial position $(x,\theta)$ with $x\in[-R,R]$ and $\theta \in \{-1,+1 \}$. The addition of $\varepsilon>0$ is necessary to introduce a margin where the flip can take place. See Figure~\ref{fig:smallset} for a visual aid. 

We show that $C$ is a uniform $(\nu,\delta,T)$-small set. Since the proof applies for all $\varepsilon>0$, although resulting indifferent $\delta$'s, one is free to choose $t_0=T(\varepsilon)$ freely. Let us consider the two cases below.\smallskip
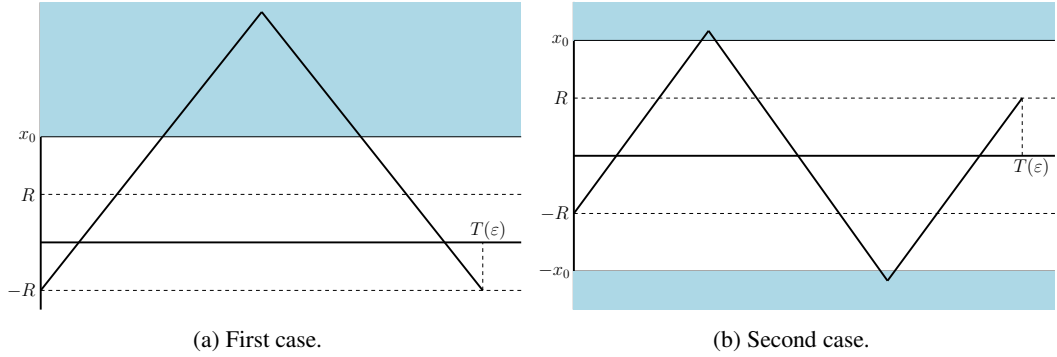
\begin{figure}[h]
\begin{subfigure}[b]{0.49 \textwidth}
\resizebox{ \textwidth}{!}{ \ifx\du\undefined
  \newlength{\du}
\fi
\setlength{\du}{15\unitlength}
\begin{tikzpicture}
\pgftransformxscale{1.000000}
\pgftransformyscale{-1.000000}
\definecolor{dialinecolor}{rgb}{0.000000, 0.000000, 0.000000}
\pgfsetstrokecolor{dialinecolor}
\definecolor{dialinecolor}{rgb}{1.000000, 1.000000, 1.000000}
\pgfsetfillcolor{dialinecolor}
\pgfsetlinewidth{0.100000\du}
\pgfsetdash{}{0pt}
\pgfsetdash{}{0pt}
\pgfsetbuttcap 
{
\definecolor{dialinecolor}{rgb}{0.000000, 0.000000, 0.000000}
\pgfsetfillcolor{dialinecolor}
\definecolor{dialinecolor}{rgb}{0.000000, 0.000000, 0.000000}
\pgfsetstrokecolor{dialinecolor}
\draw (6.000000\du,12.000000\du)--(6.000000\du,28.000000\du);
}
\pgfsetlinewidth{0.100000\du}
\pgfsetdash{}{0pt}
\pgfsetdash{}{0pt}
\pgfsetbuttcap  
{
\definecolor{dialinecolor}{rgb}{0.000000, 0.000000, 0.000000}
\pgfsetfillcolor{dialinecolor}
\definecolor{dialinecolor}{rgb}{0.000000, 0.000000, 0.000000}
\pgfsetstrokecolor{dialinecolor}
\draw (6.000000\du,24.5000000\du)--(31.000000\du,24.5000000\du);
}
\Large
\definecolor{dialinecolor}{rgb}{0.000000, 0.000000, 0.000000}
\pgfsetstrokecolor{dialinecolor}
\node[anchor=west] at (4.6000000\du,22.000000\du){$R$};
\definecolor{dialinecolor}{rgb}{0.000000, 0.000000, 0.000000}
\pgfsetstrokecolor{dialinecolor}
\node[anchor=west] at (3.900000\du,27.000000\du){$-R$};
\definecolor{dialinecolor}{rgb}{0.000000, 0.000000, 0.000000}
\pgfsetstrokecolor{dialinecolor}
\node[anchor=west] at (4.500000\du,19.000000\du){$x_0$};
\pgfsetlinewidth{0.000000\du}  
\pgfsetdash{}{0pt}
\pgfsetdash{}{0pt}
\pgfsetmiterjoin
\definecolor{dialinecolor}{rgb}{0.678431, 0.847059, 0.901961}
\pgfsetfillcolor{dialinecolor}
\fill (6.000000\du,12.000000\du)--(6.000000\du,19.000000\du)--(31.000000\du,19.000000\du)--(31.000000\du,12.000000\du)--cycle;

\pgfsetlinewidth{0.00000\du}
\pgfsetdash{}{0pt}
\pgfsetdash{}{0pt}
\pgfsetbuttcap 
{
\definecolor{dialinecolor}{rgb}{0.000000, 0.000000, 0.000000}
\pgfsetfillcolor{dialinecolor}
\definecolor{dialinecolor}{rgb}{0.000000, 0.000000, 0.000000}
\pgfsetstrokecolor{dialinecolor}
\draw (6.000000\du,19.000000\du)--(31.000000\du,19.000000\du);
}



\pgfsetlinewidth{0.00000\du}
\pgfsetdash{}{0pt}
\pgfsetdash{}{0pt}
\pgfsetbuttcap 
{
\definecolor{dialinecolor}{rgb}{0.000000, 0.000000, 0.000000}
\pgfsetfillcolor{dialinecolor}
\definecolor{dialinecolor}{rgb}{0.000000, 0.000000, 0.000000}
\pgfsetstrokecolor{dialinecolor}
\draw[dashed] (6.000000\du,22.000000\du)--(31.000000\du,22.000000\du);
}

\pgfsetlinewidth{0.00000\du}
\pgfsetdash{}{0pt}
\pgfsetdash{}{0pt}
\pgfsetbuttcap 
{
\definecolor{dialinecolor}{rgb}{0.000000, 0.000000, 0.000000}
\pgfsetfillcolor{dialinecolor}
\definecolor{dialinecolor}{rgb}{0.000000, 0.000000, 0.000000}
\pgfsetstrokecolor{dialinecolor}
\draw[dashed] (6.000000\du,27.000000\du)--(31.000000\du,27.000000\du);
}


\pgfsetlinewidth{0.100000\du}
\pgfsetdash{}{0pt}
\pgfsetdash{}{0pt}
\pgfsetbuttcap
{
\definecolor{dialinecolor}{rgb}{0.000000, 0.000000, 0.000000}
\pgfsetfillcolor{dialinecolor}
\definecolor{dialinecolor}{rgb}{0.000000, 0.000000, 0.000000}
\pgfsetstrokecolor{dialinecolor}
\draw (6.000000\du,27.000000\du)--(17.500000\du,12.5000000\du);
}
\pgfsetlinewidth{0.100000\du}
\pgfsetdash{}{0pt}
\pgfsetdash{}{0pt}
\pgfsetbuttcap
{
\definecolor{dialinecolor}{rgb}{0.000000, 0.000000, 0.000000}
\pgfsetfillcolor{dialinecolor}
\definecolor{dialinecolor}{rgb}{0.000000, 0.000000, 0.000000}
\pgfsetstrokecolor{dialinecolor}
\draw (17.50000\du,12.5000000\du)--(29.00000\du,27.000000\du);
}

\definecolor{dialinecolor}{rgb}{0.000000, 0.000000, 0.000000}
\pgfsetstrokecolor{dialinecolor}
\node[anchor=west] at (28.00000\du,23.9000000\du){$T(\varepsilon)$};

\pgfsetlinewidth{0.00000\du}
\pgfsetdash{}{0pt}
\pgfsetdash{}{0pt}
\pgfsetbuttcap
{
\definecolor{dialinecolor}{rgb}{0.000000, 0.000000, 0.000000}
\pgfsetfillcolor{dialinecolor}
\definecolor{dialinecolor}{rgb}{0.000000, 0.000000, 0.000000}
\pgfsetstrokecolor{dialinecolor}
\draw[dashed] (29.00000\du,27.000000\du)--(29.00000\du,24.500000\du);
}
\end{tikzpicture}}
\caption{First case.}
\end{subfigure}
\begin{subfigure}[b]{0.49 \textwidth}
\resizebox{\textwidth}{!}{  \ifx\du\undefined
  \newlength{\du}
\fi
\setlength{\du}{15\unitlength}
\begin{tikzpicture}
\Large
\pgftransformxscale{1.000000}
\pgftransformyscale{-1.000000}
\definecolor{dialinecolor}{rgb}{0.000000, 0.000000, 0.000000}
\pgfsetstrokecolor{dialinecolor}
\definecolor{dialinecolor}{rgb}{1.000000, 1.000000, 1.000000}
\pgfsetfillcolor{dialinecolor}
\pgfsetlinewidth{0.100000\du}
\pgfsetdash{}{0pt}
\pgfsetdash{}{0pt}
\pgfsetbuttcap 
{
\definecolor{dialinecolor}{rgb}{0.000000, 0.000000, 0.000000}
\pgfsetfillcolor{dialinecolor}
\definecolor{dialinecolor}{rgb}{0.000000, 0.000000, 0.000000}
\pgfsetstrokecolor{dialinecolor}
\draw (6.000000\du,12.000000\du)--(6.000000\du,28.000000\du);
}
\pgfsetlinewidth{0.100000\du}
\pgfsetdash{}{0pt}
\pgfsetdash{}{0pt}
\pgfsetbuttcap  
{
\definecolor{dialinecolor}{rgb}{0.000000, 0.000000, 0.000000}
\pgfsetfillcolor{dialinecolor}
\definecolor{dialinecolor}{rgb}{0.000000, 0.000000, 0.000000}
\pgfsetstrokecolor{dialinecolor}
\draw (6.000000\du,20.000000\du)--(31.000000\du,20.000000\du);
}
\definecolor{dialinecolor}{rgb}{0.000000, 0.000000, 0.000000}
\pgfsetstrokecolor{dialinecolor}
\node[anchor=west] at (4.6000000\du,17.000000\du){$R$};
\definecolor{dialinecolor}{rgb}{0.000000, 0.000000, 0.000000}
\pgfsetstrokecolor{dialinecolor}
\node[anchor=west] at (3.900000\du,23.000000\du){$-R$};
\definecolor{dialinecolor}{rgb}{0.000000, 0.000000, 0.000000}
\pgfsetstrokecolor{dialinecolor}
\node[anchor=west] at (4.500000\du,14.000000\du){$x_0$};
\definecolor{dialinecolor}{rgb}{0.000000, 0.000000, 0.000000}
\pgfsetstrokecolor{dialinecolor}
\node[anchor=west] at (3.800000\du,26.000000\du){$-x_0$};
\pgfsetlinewidth{0.000000\du}  
\pgfsetdash{}{0pt}
\pgfsetdash{}{0pt}
\pgfsetmiterjoin
\definecolor{dialinecolor}{rgb}{0.678431, 0.847059, 0.901961}
\pgfsetfillcolor{dialinecolor}
\fill (6.000000\du,14.000000\du)--(6.000000\du,12.000000\du)--(31.000000\du,12.000000\du)--(31.000000\du,14.000000\du)--cycle;

\pgfsetlinewidth{0.00000\du}
\pgfsetdash{}{0pt}
\pgfsetdash{}{0pt}
\pgfsetbuttcap 
{
\definecolor{dialinecolor}{rgb}{0.000000, 0.000000, 0.000000}
\pgfsetfillcolor{dialinecolor}
\definecolor{dialinecolor}{rgb}{0.000000, 0.000000, 0.000000}
\pgfsetstrokecolor{dialinecolor}
\draw (6.000000\du,14.000000\du)--(31.000000\du,14.000000\du);
}

\pgfsetlinewidth{0.00000\du}
\pgfsetdash{}{0pt}
\pgfsetdash{}{0pt}
\pgfsetbuttcap 
{
\definecolor{dialinecolor}{rgb}{0.000000, 0.000000, 0.000000}
\pgfsetfillcolor{dialinecolor}
\definecolor{dialinecolor}{rgb}{0.000000, 0.000000, 0.000000}
\pgfsetstrokecolor{dialinecolor}
\draw (6.000000\du,26.000000\du)--(31.000000\du,26.000000\du);
}

\pgfsetlinewidth{0.000000\du}  
\pgfsetdash{}{0pt}
\pgfsetdash{}{0pt}
\pgfsetmiterjoin
\definecolor{dialinecolor}{rgb}{0.678431, 0.847059, 0.901961}
\pgfsetfillcolor{dialinecolor}
\fill (6.000000\du,26.000000\du)--(6.000000\du,28.000000\du)--(31.000000\du,28.000000\du)--(31.000000\du,26.000000\du)--cycle;

\pgfsetlinewidth{0.00000\du}
\pgfsetdash{}{0pt}
\pgfsetdash{}{0pt}
\pgfsetbuttcap 
{
\definecolor{dialinecolor}{rgb}{0.000000, 0.000000, 0.000000}
\pgfsetfillcolor{dialinecolor}
\definecolor{dialinecolor}{rgb}{0.000000, 0.000000, 0.000000}
\pgfsetstrokecolor{dialinecolor}
\draw[dashed] (6.000000\du,17.000000\du)--(31.000000\du,17.000000\du);
}

\pgfsetlinewidth{0.00000\du}
\pgfsetdash{}{0pt}
\pgfsetdash{}{0pt}
\pgfsetbuttcap 
{
\definecolor{dialinecolor}{rgb}{0.000000, 0.000000, 0.000000}
\pgfsetfillcolor{dialinecolor}
\definecolor{dialinecolor}{rgb}{0.000000, 0.000000, 0.000000}
\pgfsetstrokecolor{dialinecolor}
\draw[dashed] (6.000000\du,23.000000\du)--(31.000000\du,23.000000\du);
}


\pgfsetlinewidth{0.100000\du}
\pgfsetdash{}{0pt}
\pgfsetdash{}{0pt}
\pgfsetbuttcap
{
\definecolor{dialinecolor}{rgb}{0.000000, 0.000000, 0.000000}
\pgfsetfillcolor{dialinecolor}
\definecolor{dialinecolor}{rgb}{0.000000, 0.000000, 0.000000}
\pgfsetstrokecolor{dialinecolor}
\draw (6.000000\du,23.000000\du)--(13.000000\du,13.5000000\du);
}
\pgfsetlinewidth{0.100000\du}
\pgfsetdash{}{0pt}
\pgfsetdash{}{0pt}
\pgfsetbuttcap
{
\definecolor{dialinecolor}{rgb}{0.000000, 0.000000, 0.000000}
\pgfsetfillcolor{dialinecolor}
\definecolor{dialinecolor}{rgb}{0.000000, 0.000000, 0.000000}
\pgfsetstrokecolor{dialinecolor}
\draw (13.000000\du,13.5000000\du)--(22.300000\du,26.5000000\du);
}
\pgfsetlinewidth{0.100000\du}
\pgfsetdash{}{0pt}
\pgfsetdash{}{0pt}
\pgfsetbuttcap
{
\definecolor{dialinecolor}{rgb}{0.000000, 0.000000, 0.000000}
\pgfsetfillcolor{dialinecolor}
\definecolor{dialinecolor}{rgb}{0.000000, 0.000000, 0.000000}
\pgfsetstrokecolor{dialinecolor}
\draw (22.300000\du,26.5000000\du)--(29.3000000\du,17.000000\du);
}

\definecolor{dialinecolor}{rgb}{0.000000, 0.000000, 0.000000}
\pgfsetstrokecolor{dialinecolor}
\node[anchor=west] at (28.500000\du,20.6000000\du){$T(\varepsilon)$};

\pgfsetlinewidth{0.00000\du}
\pgfsetdash{}{0pt}
\pgfsetdash{}{0pt}
\pgfsetbuttcap
{
\definecolor{dialinecolor}{rgb}{0.000000, 0.000000, 0.000000}
\pgfsetfillcolor{dialinecolor}
\definecolor{dialinecolor}{rgb}{0.000000, 0.000000, 0.000000}
\pgfsetstrokecolor{dialinecolor}
\draw[dashed] (29.300000\du,17.000000\du)--(29.300000\du,20.000000\du);
}
\end{tikzpicture}}
\caption{Second case.}
\end{subfigure}
\caption{Illustration of the proof of Lemma~\ref{lemmma:smallsetZZ1D}.}
\label{fig:smallset}
\end{figure}

\textit{First case:} let the initial condition be $(x,+1)$ with $x \in [-R,R]$ and we want to be at time $T$ in $B= A \times \{ -1 \}$ with $A\in \mathcal{B}([-R,+R])$ any Borel set. We can then use the following inequality
\begin{equation}\label{eq:smallset_E_1}
    \mathbb{P}_{(x,+1)}((X(T),\Theta(T)) \in B) \ge \mathbb{P}_{(x,+1)}( X(T) \in A, E_1),
\end{equation}
where $E_1$ is the event that exactly one velocity switch takes place. We are then in the case of Figure~\ref{fig:smallset}(a). Observe that by the choice of $T$ the process has enough time to travel the longest path, i.e. from $(-R,+1)$ to $(-R,-1)$ with the smallest allowed velocity $V_{\text{min}}$. In order to compute the r.h.s. in \eqref{eq:smallset_E_1} one can compute $\mathbb{P}_{(x,+1)}( X(T) \leq y, E_1)$ and then differentiate with respect to $y$. To do this we assume only one velocity switch and impose that $X(T) \le y$, resulting in the condition
\begin{equation}\notag
    X(T) = x + m t - m(T-t) \leq y,
\end{equation}
where $t$ is the time at which the velocity switch takes place. By rearranging we obtain the condition $t \leq \tfrac{y-x}{2 m} + \tfrac{T}{2}=:\Bar{t}(y)$. Observe that for any $m\in[V_{\text{min}},V_{\text{max}}]$ the process has enough time to reach the region where it is assumed there is a strictly positive probability of a velocity flip. Indeed using that $y \in [-R,R]$ we find
\begin{equation}\notag
    \Bar{t}(y) \geq \frac{x_0-x}{m} + \frac{y+x}{2 m} + \frac{2x_0+2R}{2V_{\text{min}}} + \frac{\varepsilon}{2} \geq \frac{x_0-x}{m} + \frac{x_0}{V_{\text{min}}} + \frac{\varepsilon}{2}> \frac{x_0-x}{m}.
\end{equation}
Therefore 
\begin{equation}\notag
\begin{aligned}
    \mathbb{P}_{(x,+1)}( X(T) \leq y, E_1) = & \int_0^{\Bar{t}(y)} \lambda(x+m s,1) \exp \left( -\int_0^s \lambda(x+m u,1)du \right)\\
    &  \exp \left( -\int_0^{T-s} \lambda(x+m s-m u,-1)du \right) ds.
\end{aligned}
\end{equation}
After differentiation one obtains
\begin{align}
        \notag \mathbb{P}_{(x,+1)}( & X(T) \in A, E_1) =\\
        \notag & = \int_{A} \frac{1}{2 m}  \lambda(x+m \Bar{t}(y),1) \exp \left( -\int_0^{\Bar{t}(y)} \lambda(x+ m u,1)du \right) \\
        \notag & \hspace{30pt} \exp \left( -\int_0^{T-\Bar{t}(y)} \lambda(x+m \Bar{t}(y)-m u,-1)du \right) dy \\
        \notag &\ge \frac{1}{2 V_{\text{min}}} \int_{A}  \lambda(x+m \Bar{t}(y),1) \exp (-\Bar{t}(y) \lambda_{\text{max}} - (T-\Bar{t}(y))\lambda_{\text{max}} ) dy \\
        \notag & = \frac{\exp (- \lambda_{\text{max}} T)}{2 V_{\text{min}}} \int_{A} \lambda(x+m \Bar{t}(y),1) dy \\
        \notag & \ge 2 \, \frac{ \lambda_{\text{min}} R  \exp (- \lambda_{\text{max}} T)}{V_{\text{min}}} \left( \frac{1}{2} \cdot \frac{1}{2 R} \int_{A} dy \right) \\
        \notag & = 2 \,  \frac{ \lambda_{\text{min}} R  \exp (- \lambda_{\text{max}} T)}{V_{\text{min}}} \,\, \nu(B).
\end{align}
where 
\begin{equation}\label{eq:bounds_lambdas}
\begin{aligned}
    \lambda_{\text{max}} & \coloneqq \max_{ \{ x\in \tilde{C}_T, \theta \in \{-1,+1\} \}} \left( V_{\text{max}} (\theta U'(x))_+ + \gamma(x)\right) < \infty, \\
    \lambda_{\text{min}}  & \coloneqq \min \left\{\inf_{x \ge x_0} (V_{\text{min}} U'(x))_+\,,\, \inf_{x \le -x_0} (V_{\text{min}} U'(x))_+ \right\}  > 0.
\end{aligned}
\end{equation}
where $\tilde{C}_T \coloneqq [-R -V_{\text{max}}T, R+V_{\text{max}}T ]$ is the set of points that can be reached in time T by a process starting in $x\in[-R,R]$ with velocity $V_{\text{max}}$ (and thus for any value $m\in[V_{\text{min}},V_{\text{max}}]$). Thus $\lambda_{\text{max}}$ is the maximum switching rate achieved within time $T$. \smallskip

\textit{Second case:} again we consider an initial condition with positive velocity, but now also a positive velocity at time $T$, i.e. we consider a set $B=A \times \{+1\}$. Taking advantage of the same idea as before we use the bound
\begin{equation}\notag
    \mathbb{P}_{(x,+1)}((X(T),\Theta(T)) \in B) \ge \mathbb{P}_{(x,+1)}( X(T) \in A, E_2),
\end{equation}
where $E_2$ is the event that exactly two switches take place. This case corresponds to Figure~\ref{fig:smallset}(b). Let $t_1$ and $t_2$ be the times of first and second switch respectively. Then $X(T) \le y$ when 
\begin{equation}\notag
    X(T) = x + m t_1 -m (t_2-t_1) + m (T-t_2) \le y,
\end{equation}
which can be rearranged as $t_2 \ge t_1 + \tfrac{x-y}{2 m} + \tfrac{T}{2} =: \bar{t_2}(y)$. We can obtain a bound for $t_1$ by imposing that $\bar{t_2}(y) < T$. The resulting condition is $t_1 < T/2 - (x-y)/(2 m) =: \bar{t}_1(y)$. We observe that by definition of $T$ it follows that for any $x,y\in[-R,R]$
$$\bar{t}_1(y) -\frac{x_0-x}{m} \geq  \frac{2x_0}{V_{\text{min}}} - \frac{x_0-x}{m} + \frac{\varepsilon}{2} > 0,$$ which means that there is enough time for the process to reach cyan region in Figure~\ref{fig:smallset}(b) and have the first velocity flip. The same holds true for the second velocity flip.
Now compute the distribution function as
\begin{equation}\notag
    \begin{aligned}
        \mathbb{P}_{(x,+1)} (X(T) \le y, E_2) &= \int_{s=0}^{\bar{t}_1(y)} \int_{u=\tilde{t}_2(y)}^{T-s} \lambda(x+m s,1) \exp{ \left(- \int_{0}^s \lambda(x+m l,1)dl \right)} \\
        & \lambda(x+m s-m u,-1) \exp{ \left(- \int_{0}^u \lambda(x+m s- m l,-1) dl \right) } \\
        & \exp{ \left( - \int_{0}^{T-u-s} \lambda(x+m s-m u+m l,1)dl \right) } ds \, du,
    \end{aligned}
\end{equation}
where $\tilde{t}_2(y) :=\bar{t_2}(y)-t_1= \tfrac{x-y}{2 m} + \tfrac{T}{2} $. Then by differentiating we obtain
\begin{align}\notag
        \mathbb{P}_{(x,+1)} & (X(T) \in A,  E_2) =\\
        & \notag =\int_A \int_0^{\bar{t}_1(y)} \frac{1}{2 m} \lambda(x+m s,1) \exp{ \left(- \int_{l=0}^s \lambda(x+m l,1)dl \right)} \\
        & \notag \qquad \lambda(x+ws-w \tilde{t}_2(y),-1) \exp{ \left(- \int_{l=0}^{\tilde{t}_2(y)} \lambda(x+m s-m l,-1) dl \right) } \\
        & \notag \qquad \exp{ \left( - \int_{l=0}^{T-\tilde{t}_2(y)-s} \lambda(x+m s-m \tilde{t}_2(y)+m l,1)dl \right) } ds \, dy \\
        & \notag \ge \frac{\exp{(-\lambda_{\text{max}} T )}}{2 V_{\text{min}}} \int_A \int_0^{\bar{t}_1(y)} \lambda(x+ws,1) \lambda(x+ws-w \tilde{t}_2(y),-1) \,ds\, dy.
\end{align}
Since the integrand is non-negative we can lower bound this quantity by restricting the domain of integration corresponding to the $s$ variable. A sensible choice is $(x_0-x)/m \leq s\leq -(x_0+x)/m + \tilde{t}_2(y)$, as this would imply that both switches take place in the cyan region in Figure~\ref{fig:smallset}(b). Indeed for $s\in[(x_0-x)/m,-(x_0+x)/m + \tilde{t}_2(y)]$ we have that $x+ms \geq x_0 $ and $x+ms-m \tilde{t}_2(y) \leq -x_0$. Observe that 
$$
-\frac{x_0+x}{m} + \tilde{t}_2(y) = \frac{T}{2} + \frac{x-y}{2m} - \frac{x+x_0}{m} \leq \frac{T}{2} + \frac{x-y}{2m} - \frac{x-y}{m} \leq  \bar{t}_1(y),
$$
where we used that $-x_0 \leq y$ since $y\in[-R,R]$. Combining this with the fact that $\bar{t}_1(y) > \frac{x_0-x}{m}$, we have shown that $[(x_0-x)/m,-(x_0+x)/m + \tilde{t}_2(y)] \subset [0,\bar{t}_1(y)]$. Therefore on this interval we can bound below the switching rates as follows
\begin{align}\notag    
    \mathbb{P}_{(x,+1)} & (X(T) \in A,  E_2) \geq\\
    & \notag \ge \frac{\exp{(-\lambda_{\text{max}} T )}}{2 V_{\text{min}}} \int_A \int_{ \frac{x_0-x}{m} }^{ -\frac{x_0+x}{m} + \tilde{t}_2(y) } \lambda(x+ms,1) \lambda(x+ms-w \tilde{t}_2(y),-1) ds dy \\
    & \notag \ge \frac{\exp{(-\lambda_{\text{max}} T )}}{2 V_{\text{min}}} \lambda^2_{\text{min}} \int_A \int_{ \frac{x_0-x}{m} }^{ -\frac{x_0+x}{m} + \tilde{t}_2(y) } ds dy \\
    & \notag\ge \frac{\exp{(-\lambda_{\text{max}} T )}}{2 V_{\text{min}}} \lambda^2_{\text{min}} \int_A \frac{\varepsilon}{2} dy \\
    & \notag = \frac{R \exp{(-\lambda_{\text{max}} T )}}{ V_{\text{min}}} \lambda_{\text{min}}^2 \,\, \varepsilon\, \nu(B).
\end{align}
By symmetry the same bounds hold also when the process has initial velocity is $-1$. Therefore for any $\varepsilon>0$ we proved that $C=[-R,+R]\times \{-1,+1\}$ is a uniform $(\nu,\delta,T)$-small set with
\begin{equation}\label{eq:delta_smallset_ZZ}
    \delta = \min \left\{ 2 \,  \frac{ \lambda_{\text{min}} R  \exp (- \lambda_{\text{max}} T)}{V_{\text{min}}},  \, \frac{ \varepsilon \,\lambda_{\text{min}}^2 R \exp{(-\lambda_{\text{max}} T )}}{V_{\text{min}}}  \right \}.
\end{equation}
Notice that there is an $\varepsilon$ dependence also in $\lambda_{\text{min}}$ and $\lambda_{\text{max}}$, but there is no dependence on the specific value of $m$. 
To conclude the argument we observe that it is always possible to incorporate any compact set in a set of the same form as $C$ and therefore all compact sets are uniformly small for the family of Zig-Zag processes.

When the switching rate is strictly positive but Assumption~\ref{ass:smallsetZZ} does not hold, the same reasoning can be applied by taking $x_0=0$, $\lambda_{\text{min}}   = \gamma_{\text{min}}  > 0$, and $T  = \frac{2 R}{V_{\text{min}}} + \varepsilon$ for some $\varepsilon >0$. The switching rates can be upper-bounded because $\gamma$ is bounded on compact sets.
Thus the only difference is that in this case the process can switch at any time and does not need to escape a compact set to do so. One can again obtain the uniform small set condition with $\delta$ defined as in \eqref{eq:delta_smallset_ZZ}, and $\nu$ as above.

Finally, consider the family $\{\mathcal{L}_{m,\gamma}: m \in \M, \gamma \in \Lambda \}$. Choosing $$\lambda_{\text{max}} = \max_{ \{ x\in \tilde{C}_T, \theta \in \{-1,+1\} \}} \left( V_{\text{max}} (\theta U'(x))_+\right) + \gamma_{\text{max}}$$ and applying the same reasoning as in the case of a strictly positive refreshment shows the statement.
\end{proof}

\subsubsection{Simultaneous drift conditions, from discrete time to continuous time}

The lemma below is needed to go from simultaneous, continuous time drift conditions to simultaneous, discrete time ones.
\begin{lemma}\label{lemma:drift2drift}
Consider a family of Markov processes with state space $E$ and generators $\{ \mathcal{L}_\gamma : \gamma \in \mathcal{Y} \}$. Assume the following conditions are verified:
\begin{enumerate}[label=\alph*)]
    \item there exist $c_1>0$, $c_2>0$, a set $K \subset E$ all independent of $\gamma$, and a class of functions $\{V_\gamma:E\to [1,+\infty): \gamma \in \mathcal{Y} \}$, such that for each $\gamma \in \mathcal{Y}$ the following drift condition holds
    \begin{equation}\label{eq:lemmadrift}
        \mathcal{L}_\gamma V_\gamma(z) \le -c_1 V_\gamma (z) + c_2 \mathbbm{1}_K (z)\qquad \text{for all } z\in E.
    \end{equation}
    \item the family of processes is such that for  $K \subset E$ as in the previous point, a constant $t>0$, and all $z \in E$, there exists a set $\Tilde{K}(t) \subset E$ which depends on $t$ such that on the event $z \notin \Tilde{K}(t)$ it holds that $Z(t) \notin K$ almost surely for all $\gamma \in \mathcal{Y}$.
\end{enumerate}
%
Then for any $t>0$ there are $\lambda \in (0,1)$ and $\kappa >0$, both independent of $\gamma$, such that for each $\gamma \in \mathcal{Y}$
\begin{equation}\notag
    P_\gamma^t V_\gamma (z) \le \lambda V_\gamma(z) + \kappa \mathbbm{1}_{\tilde{K}}(z)\qquad \text{for all } z\in E.
\end{equation}
In particular, if conditions \eqref{eq:lemmadrift} hold with $V \equiv V_\gamma$ for each $\gamma$, then 
\begin{equation}\notag
    P_\gamma^t V (z) \le \lambda V(z) + \kappa \mathbbm{1}_{\tilde{K}}(z) \qquad \text{for all } z\in E.
\end{equation}
\end{lemma}

\begin{proof}
Let $t>0$ and $\gamma \in \mathcal{Y}$. As a first step, we apply Dynkin's formula to $f(z,t) = \exp{(c_1 t)} V_\gamma(z)$ as in the proof of Theorem 6.1 in \cite{meyntweedie1993} to obtain 
\begin{equation}\notag
    e^{c_1 t} P^t_\gamma V_\gamma(z) 
    = V_\gamma(z) + \int_0^t P^s_\gamma  \left( \frac{\partial}{\partial s} + \mathcal L_{\gamma}\right) \left( e^{c_1 s}V_\gamma(z) \right) ds. 
\end{equation}
By the product rule and the drift condition we can write the integrand on the right hand side as
\begin{align}
    \notag P^s_\gamma  \left( \frac{\partial}{\partial s} + \mathcal L_{\gamma}\right) \left( e^{c_1 s}V_\gamma(z) \right) &= e^{c_1 s} P_\gamma^s \mathcal{L}_\gamma V_\gamma(z) + c_1 e^{c_1 s} P_\gamma^s V_\gamma(z) \\
    \notag & \le   e^{c_1 s} P_\gamma^s \left( -c_1 V_\gamma(z) +c_2 \mathbbm{1}_K(z)  \right)+ c_1 e^{c_1 s} P_\gamma^s V_\gamma(z) \\
    \notag & = c_2 e^{c_1 s}  P_\gamma^s \mathbbm{1}_K(z).
\end{align}
Therefore by linearity of the integral
\begin{align}
    \notag P^t_\gamma V_\gamma(z) & \le e^{-c_1 t} V_\gamma(z)  + c_2 \int_0^t  e^{c_1 (s-t)}  P_\gamma^s \mathbbm{1}_K(z) ds \\
    \notag & \le e^{-c_1 t} V_\gamma(z)  + c_2 \mathbbm{1}_{\tilde{K}}(z) \int_0^t  e^{c_1 (s-t)} ds \\
    \notag & =  e^{-c_1 t} V_\gamma(z)  + \frac{c_2}{c_1} ( 1- e^{-c_1 t}) \mathbbm{1}_{\tilde{K}}(z).
\end{align}
In the second inequality we took advantage of condition (b) to conclude that $P_\gamma^s \mathbbm{1}_K(z) = \mathbb{P}_z (Z(s) \in K) \le \mathbbm{1}_{\tilde{K}}(z)$.
It is then sufficient to take $\lambda = e^{-c_1 t}$, $\kappa = \frac{c_2}{c_1} ( 1- e^{-c_1 t})$ to conclude the proof. Observe that $\lambda$ and $\kappa$ do not depend on $\gamma$ (but depend on $t$).
\end{proof}
\begin{remark}
Condition (b) in Lemma~\ref{lemma:drift2drift} is for instance satisfied by a family of preconditioned Zig-Zag processes, if the preconditioner is taken from a compact set of positive definite matrices. Indeed this claim follows from the fact that all processes travel with almost surely bounded velocity. Therefore, it is possible to choose $\tilde{K}(t)$ by adding a suitable buffer zone around the set $K$, such that the set $K$ is not reachable in time $t$ starting outside $\tilde{K}(t)$.
\end{remark}

\begin{remark}
Instead of drift conditions of the form \eqref{eq:lemmadrift}, consider the following case $$\mathcal{L}_\gamma V_\gamma(z) \le -c_1 V_\gamma (z) + c_2 \qquad \text{for all } z\in E.$$ Without condition (b) in Lemma~\ref{lemma:drift2drift} we can still conclude by the same line of reasoning that there are $\lambda \in (0,1)$, $\kappa >0$, both independent of $\gamma$, such that for each $\gamma \in \mathcal{Y}$
\begin{equation}\notag
    P_\gamma^t V_\gamma (z) \le \lambda V_\gamma(z) + \kappa \qquad \text{for all } z\in E.
\end{equation}
\end{remark}

\section*{Acknowledgements}
This work is part of the research programme `Zigzagging through computational barriers' with project number 016.Vidi.189.043, which is financed by the Dutch Research Council (NWO). We acknowledge helpful discussions with Gareth Roberts.


\bibliographystyle{imsart-number} 
\bibliography{biblio.bib}       


\end{document}